\documentclass{article}
\usepackage[utf8]{inputenc}

\usepackage[hmargin=0.75in,vmargin=1in]{geometry} 

\usepackage[round]{natbib}

\usepackage[dvipsnames]{xcolor}

\usepackage{mathtools} 
\usepackage{booktabs} 
\usepackage{tikz} 
\usepackage{caption}
\usepackage{subcaption} 
\usepackage{amsfonts} 
\usepackage{amsthm}
\usepackage{thmtools}
\usetikzlibrary{arrows.meta} 
\usepackage{algorithm}
\usepackage{algpseudocode}
\usepackage{array}
\usepackage{multirow}
\usepackage{csquotes}
\usepackage{amssymb}
\usepackage{graphbox}
\usepackage{bm}
\usepackage{authblk}

\usepackage{hyperref}
\hypersetup{
  colorlinks=true,
  linkcolor=NavyBlue,
  citecolor=gray
}
\usepackage[capitalize]{cleveref}

\theoremstyle{plain}

\theoremstyle{definition}

\newtheorem{example}{Example}

\theoremstyle{remark}
\newtheorem{remark}{Remark}[section]

\Crefname{assumption}{Assumption}{Assumptions}
\crefname{assumption}{Assumption}{Assumptions}
\Crefname{condition}{Condition}{Conditions}
\crefname{condition}{Condition}{Conditions}
\Crefname{secinapp}{Appendix}{Appendices}
\crefname{secinapp}{Appendix}{Appendices}
\renewcommand{\thmcontinues}[1]{%
\ifcsname hyperref\endcsname
    \hyperref[#1]{continued}%
\else
    continued%
\fi%
}

\def\ci{\perp\!\!\!\perp}

\newcommand{\E}{\mathbb{E}}
\newcommand{\G}{\mathcal{G}}
\newcommand{\D}{\mathcal{D}}
\newcommand{\bD}{\bm{\mathcal{D}}}
\newcommand{\W}{\mathcal{W}}

\newcommand{\M}{\mathcal{M}}
\newcommand{\btheta}{\bm{\theta}}
\newcommand{\boldeta}{\bm{\eta}}
\newcommand{\bthetahat}{\hat{\btheta}}

\usepackage{xcolor}

\usepackage[normalem]{ulem}
\newcommand{\stkout}[1]{\ifmmode\text{\sout{\ensuremath{#1}}}\else\sout{#1}\fi}

\definecolor{friendlyblue}{HTML}{1E88E5}
\definecolor{friendlyred}{HTML}{D81B60}
\definecolor{friendlyyellow}{HTML}{FFC107}

\title{Score-Based Methods for Selecting Direct Causes in Reweighted Distributions}

\begin{document}

\author[1]{Jacob M. Chen}
\author[1]{Ilya Shpitser}
\affil[1]{Department of Computer Science, Johns Hopkins University}

\maketitle

\begin{abstract}
    Scoring methods based on the Bayesian information criterion (BIC) are commonly used for model selection tasks such as choosing features in a regression model or learning directed acyclic graphs from data.
    In certain model classes, model selection based on the BIC is consistent.
    However, since the BIC is based on the observed likelihood function, it does not apply directly to reweighted distributions arising in causal inference, missing data, and domain shift applications.
    Here, we propose a generalized version of the BIC that allows for model selection in reweighted distributions.
    We prove its corresponding consistency property and demonstrate how it can be used for selecting direct causes of an outcome variable in a marginal structural model.
    Our proposed method accounts for scenarios with unmeasured confounders and where the weights must be estimated from data.
    Through simulation studies, we demonstrate the asymptotic properties of the reweighted BIC and compare it with alternative methods for model selection.
\end{abstract}

\section{Introduction}
\label{sec:introduction}

The Bayesian information criterion (BIC), first proposed in \cite{ref:schwarz1978estimating}, is a popular method for score-based model selection.
In this type of model selection, the BIC score -- computed as the maximized log-likelihood given the model penalized by its dimension -- is assigned to each candidate model, and the one with the most favorable score is chosen.
Under the assumption that the underlying probability distribution is a member of the exponential family, \cite{ref:haughton1988choice} proved that the BIC is {\it consistent}.
This means that, as the sample size approaches infinity, the BIC will select the model with the smallest dimension containing the true parameters with probability approaching one.
Recent work has also extended the BIC to cases where the number of candidate models grows as the sample size grows and where predictor variables are subject to a ranked sparsity assumption.
In such cases, additional penalty terms are added to maintain consistency \citep{ref:chen2008extended,ref:peterson2026ranked}.

The BIC has seen applications in numerous types of model selection problems, including those that aim to learn relevant features in a regression model and a directed acyclic graph (DAG) underlying the data-generating process (DGP) \citep{ref:vrieze2012model,ref:chickering2002optimal}. 
An important generalization of the former type of model selection occurs in causal inference applications where we wish to select relevant causes of an outcome among a potentially large set of candidates.
This type of problem, common in the analysis of data in healthcare, epidemiology, economics, and psychology, may be formulated as a model selection problem that selects features of a {\it marginal structural model} (MSM) 
\citep{ref:robins2000marginal,ref:hernan2000marginal}.

An MSM may be viewed as a causal regression model whose parameters are fit in a reweighted distribution corresponding to a hypothetical situation where some inputs to the regression behave as if they were randomized.
The task of selecting direct causes must be done in an MSM because associations induced by confounding in the original (unweighted) distribution may make potential causes of the outcome appear significant even when they are not.
MSMs have been used extensively in causal inference for target trial emulation, learning dynamic treatment regimes, and precision medicine \citep{ref:bodnar2004marginal,ref:buikema2018effect,ref:van2004history}.

The BIC is based on the observed likelihood function, so it cannot be directly applied to model selection problems in MSMs as they feature a reweighted distribution and, thus, a reweighted likelihood.
To address this, we propose a reweighted BIC (rBIC) and prove that it is consistent for model selection with respect to the reweighted distribution.
We further show that, with appropriate modifications to the penalty term, it remains consistent when weights are estimated from data using nuisance functions.

This paper is organized as follows.
We first define and set up the problem using a graphical models and causal inference framework.
Next, we review the derivation of the classical BIC and its consistency property.
We then present the rBIC and prove its consistency.
Throughout, we give examples that motivate our method.
Finally, we demonstrate the asymptotic properties of the rBIC in simulation studies and compare it with alternative methods for model selection.

To summarize, our main contributions are as follows:
\begin{itemize}
    \item We give a review of the derivation of the classical BIC and its proof for consistency.
    We also slightly extend the consistency proof of \cite{ref:haughton1988choice} to a more general class of distributions\footnote{Please see \cref{rm:extensions_consistency} in the Appendix for more details on how we extend the proof of \cite{ref:haughton1988choice}.}.
    \item We give examples demonstrating why the the classical BIC is not suited for certain model selection problems in causal inference, but the rBIC is.
    \item We propose the rBIC and prove that it is consistent even when weights are estimated from data using nuisance functions.
    \item In synthetic experiments, we demonstrate the asymptotic performance of the rBIC.
\end{itemize}
{\bf Related work.} 
\cite{ref:greco2020weighted} discuss methods for fitting multivariate Gaussian mixture models using weighted likelihoods.
To select the number of mixtures in their models, the authors suggest using a reweighted BIC.
Their proposed method is similar to ours but has different assumptions, uses different weights, does not discuss consistency, and is not used for selecting direct causes in causal models.

\cite{ref:bahamyirou2022doubly} use the adaptive LASSO to select effect modifiers -- covariates that change the magnitude or direction of a causal effect -- in MSMs \citep{ref:zou2006adaptive, ref:tibshirani1996regression}.
\cite{ref:jaman2025penalized} utilize the smoothly clipped absolute deviation (SCAD) penalty function to select effect modifiers in structural nested mean models with repeated outcomes \citep{ref:fan2001variable, ref:vansteelandt2014structural}.
Our problem differs from these works in that we are selecting relevant causes in reweighted distributions rather than effect modifiers in observed distributions.
Furthermore, the works above use penalized likelihoods to select effect modifiers, which reduces to an optimization problem.
In contrast, utilizing the BIC to select direct causes is a discrete search problem.
We also do not assume that all variables related to treatment assignment are observed, whereas the above works do.

\section{Motivating Examples}
\label{sec:motivating_example}

\begin{example}[Causes of Acute Kidney Injury]\label{exp:example_backdoor}
    Suppose a cardiac surgeon notices that her patients sometimes develop acute kidney injury (AKI) after cardiopulmonary bypass surgery, a type of surgery where the heart is stopped, and blood flow is facilitated by a perfusion machine.
    She suspects that poor intraoperative oxygen delivery to the patient may be causing AKI \citep{ref:magruder2015nadir, ref:ranucci2018goal}.
    Oxygen delivery, however, is measured every minute over the course of the surgery, making interventions difficult to describe and interpret.
    Thus, the researcher defines three metrics of oxygen delivery that summarizes its impacts on AKI: (i) the lowest quantity of oxygen delivered during surgery, (ii) the number of times the quantity of oxygen delivered drops below a key threshold, and (iii) the length of time the patient is connected to the perfusion machine. 
    She hopes to learn from observational data which, if any, of these intraoperative events are true causes of AKI.
    All of the confounding variables in the surgeon's dataset are observed.
\end{example}

\begin{example}[Causes of Lung Cancer]\label{exp:example_frontdoor}
    Suppose a public health researcher notices that a community near a major coal mine is recording a high incidence of lung cancer.
    He identifies three potential culprits -- a resident's smoking status, distance of a resident's home to the coal mine, and whether the resident works at the coal mine -- and would like to determine which, if any of these, are significant causes of lung cancer.
    The researcher is concerned, however, that unmeasured variables such as education and genetics could be affecting both the likelihood of a lung cancer diagnosis and the three potential causes.
    Thus, he also collects data on tar accumulation in individuals' lungs as a mediating variable between the potential causes and the outcome \citep{ref:pearl2022causal}.
    The researcher hopes to determine which of the three potential causes are significant causes of lung cancer to create policy suggestions.
\end{example}

\cref{fig:example_dag_backdoor,fig:example_dag_frontdoor}, called the backdoor and frontdoor graphs, respectively \citep{ref:pearl2009causality}, represent the DGPs of \cref{exp:example_backdoor,exp:example_frontdoor}.

\begin{figure}[t!]
    \begin{center}
        \begin{subfigure}[t]{0.225\textwidth}
            \centering
            \scalebox{0.95}{
                \begin{tikzpicture}[>=stealth, node distance=1cm]
                    \tikzstyle{square} = [draw, thick, minimum size=1.0mm, inner sep=3pt]
                    \begin{scope}
                        \path[->, very thick]
                        node[] (a2) {$A_2$}
                        node[below of=a2] (a1) {$A_3$}
                        node[above of=a2] (a3) {$A_1$}
                        node[right of=a2] (y) {$Y$}
                        node[left of=a2] (c) {${\bf C}$}
                                        
                        (a1) edge[friendlyyellow] (y)
                        (a2) edge[friendlyyellow] (y)
                        (a3) edge[friendlyyellow] (y)
                        (c) edge[friendlyblue, bend left] (y)
                        (c) edge[friendlyblue] (a1)
                        (c) edge[friendlyblue] (a2)
                        (c) edge[friendlyblue] (a3)
                        ;
                    \end{scope}
                \end{tikzpicture}
            }
            \caption{}
            \label{fig:example_dag_backdoor}
        \end{subfigure} 
        \begin{subfigure}[t]{0.225\textwidth}
            \centering
            \scalebox{0.95}{
                \begin{tikzpicture}[>=stealth, node distance=1cm]
                    \tikzstyle{square} = [draw, thick, minimum size=1.0mm, inner sep=3pt]
                    \begin{scope}
                        \path[->, very thick]
                        node[] (a2) {$A_2$}
                        node[below of=a2] (a1) {$A_3$}
                        node[above of=a2] (a3) {$A_1$}
                        node[right of=a2] (m) {$M$}
                        node[right of=m] (y) {$Y$}
                        node[left of=a2] (c) {${\bf U}$}
                                        
                        (a1) edge[friendlyyellow] (m)
                        (a2) edge[friendlyyellow] (m)
                        (a3) edge[friendlyyellow] (m)
                        (m) edge[friendlyblue] (y)
                        (c) edge[friendlyred, bend left] (y)
                        (c) edge[friendlyred] (a1)
                        (c) edge[friendlyred] (a2)
                        (c) edge[friendlyred] (a3)
                        ;
                    \end{scope}
                \end{tikzpicture}
            }    
            \caption{}
            \label{fig:example_dag_frontdoor}
        \end{subfigure}
    \end{center}
    \caption{
    DAGs representing the DGPs of \cref{exp:example_backdoor,exp:example_frontdoor}.
    In both graphs, ${\bf A} = \{A_1, A_2, A_3\}$ represents potential direct causes of the outcome.
    ${\bf C}$ and ${\bf U}$ represent observed and unobserved covariates, respectively, and $M$ represents a mediator between treatments and outcome.
    Blue and red edges are outgoing edges from observed and unobserved vertices, respectively.
    We do not know whether the yellow edges are present, and the goal of inference is to determine which, if any, of the variables in ${\bf A}$ are significant direct causes of $Y$.
    } 
    \label{fig:examples}
\end{figure}
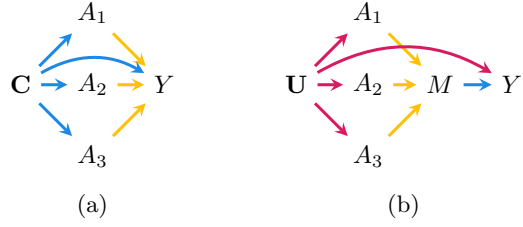

\section{Problem Setup}
\label{sec:model_and_problem_setup}

We use uppercase letters $X$ to denote random variables and uppercase letters in boldface ${\bf X}$ to denote a set of random variables.
A lowercase $x$ indicates a specific value a random variable $X$ takes, and likewise for {\bf x}.

We assume that the DGP follows that of a directed acyclic graph (DAG) $\G$ with the set of vertices ${\bf V} = {\bf C} \cup {\bf U} \cup {\bf A} \cup {\bf M} \cup \{Y\}$, where ${\bf C}$ represents observed confounders, ${\bf U}$ represents unobserved confounders, ${\bf A} = \{A_1, A_2, \dots, A_k\}$ represents $k$ potential causes of the outcome that we must select from, ${\bf M}$ represents variables mediating treatment and outcome, and $Y$ represents the outcome of interest.
A DAG $\G$ is further defined by a set of directed edges $\mathcal{E}$ connecting the vertices.
The set of variables with a directed edge to $A$ are denoted as the parents of $A$: $\text{pa}_{\G}(A)$.
We assume that the variables in $\G$ follow the partial topological ordering ${\bf C} \cup {\bf U} \prec {\bf A} \prec {\bf M} \prec Y$.
We further define the set of observed variables ${\bf O} = {\bf V} \setminus {\bf U}$.

{\bf Background on DAG models.}
The joint distribution in the statistical DAG model is factorized as
\begin{align*}
    p({\bf V}) &= \prod_{V \in {\bf V}} p(V \mid \text{pa}_{\G}(V)).
\end{align*}
Independencies in the full data distribution $p({\bf V})$ can be read off from $\G$ using the d-separation criterion \citep{ref:pearl2014probabilistic}.
DAGs are also used to represent causal relationships through the use of {\it structural causal models} (SCMs).
In an SCM, the value of each $V \in {\bf V}$ is determined by a causal mechanism function $f_V$ that takes as inputs $\text{pa}_{\G}(V)$ and an unobserved exogenous error term $\epsilon_V$ \citep{ref:pearl2009causality}.
That is, $V = f_V(\text{pa}_{\G}(V), \epsilon_V)$.
The error terms are assumed to be jointly independent of each other and ${\bf V}$.
An SCM represents causal relationships via responses to interventions, denoted by the $\text{do}()$ operator in \cite{ref:pearl2009causality}.  
An intervention that sets $A$ to a value $a$, represented by $\text{do}(a)$, replaces the structural equation $f_A$ for $A$ with another that sets $A$ to $a$, with a similar definition holding for a set ${\bf A}$.

For conciseness, we will sometimes write ${\bf a}$ as shorthand for the expression ${\bf A} = {\bf a}$.
It is well known that, given any subset ${\bf A}$ of ${\bf V}$, the interventional distribution $p({\bf V} \setminus {\bf A} \mid \text{do}({\bf a}))$ is a function of the full data distribution $p({\bf V})$. 
This function is known as the {\it g-formula} \citep{ref:robins1986new}:
\begin{align*}
    p({\bf V} \setminus {\bf A} \mid \text{do}({\bf a})) &= \prod_{V \in {\bf V} \setminus {\bf A}} p(V \mid \text{pa}_{\G}(V)) \bigg \vert_{{\bf A} = {\bf a}}.
\end{align*}

Alternatively, $p({\bf V} \setminus {\bf A} \mid \text{do}({\bf a}))$ can be represented as a joint distribution over {\it potential outcome} variables of the form ${\bf V}({\bf a}) \equiv \{ V({\bf a}) : V \in {\bf V} \setminus {\bf A} \}$, where $V({\bf a})$ means ``$V$ had ${\bf A}$, possibly contrary to fact, been set to values ${\bf a}$'' \citep{ref:splawa1990application,ref:rubin1974estimating}.
Specifically, $p({\bf V} \setminus {\bf A} \mid \text{do}({\bf a}))$ is alternatively written as $p(\{ V({\bf a}): V \in {\bf V} \setminus {\bf A} \})$.

{\it Consistency} is defined as the assumption that $Y({\bf a}) = Y$ if ${\bf A} = {\bf a}$.
This assumption allows us to relate potential outcome variables to observed variables\footnote{Note that this consistency assumption is different from consistency of the BIC. The distinction between the two should be clear from context.}.

Given an outcome $Y$ and a set of treatments ${\bf A}$, a {\it marginal structural model} (MSM) assumes that the expectation of $Y({\bf a})$ follows a parametric function indexed by a finite set of parameters $\btheta$: $\E[Y({\bf a})] = g({\bf a}; \btheta)$\footnote{See \cref{app:msms} for concrete examples of MSMs.}. 
Note that this object views ${\bf A}$ as being artificially set rather than being conditioned on, which would result in the object $E[Y \mid {\bf A}={\bf a}]$.
It is in this sense that MSMs are viewed as ``causal regressions.''

{\bf Background on exponential families.} A model for the set of random variables ${\bf X}$ is said to be a member of the exponential family if the joint probability mass function (PMF), or probability density function (PDF) for continuous variables, for ${\bf X}$ is of the form
\begin{align}
    p_{\bf X}({\bf x}, \btheta) &= h({\bf x}) \exp[{\bf T}({\bf x}) \btheta - b(\btheta)],
    \label{eq:exp_family}
\end{align}
where $\exp[x] = e^x$ is the exponential function, $h({\bf x})$ is the base measure, ${\bf T}({\bf x})$ are the sufficient statistics of the distribution, and $b(\btheta)$ is the log-partition function.
In this paper, we will use the notation $p(\cdot)$ to represent both PMFs and PDFs, and the integral and summation notations will be used interchangeably to expand expectations.
Here, ${\bf T}({\bf x})$ and $\btheta$ are both vectors, and their product in \cref{eq:exp_family} is computed as their dot product.
$\btheta \in \mathrm{int} \Theta$ -- where $\mathrm{int}\Theta$ represents the interior of the natural parameter space $\Theta$ -- is termed the {\it parameters} of the model.
The number of elements in the vector $\btheta$ is known as the dimension $d$ \citep{ref:lehmann2005testing}, and the space covered by $\Theta$ depends on the specific distribution specified in \cref{eq:exp_family}.
Common distributions in the exponential family include the normal, exponential, and Bernoulli distributions.
When the value of ${\bf X}$ depends on other variables ${\bf Z}$ in the distribution, \cref{eq:exp_family} can be extended to conditional PMFs and PDFs by allowing the parameters $\btheta$ to depend on observed values of ${\bf Z} = {\bf z}$.
Such conditional densities may be denoted as $p_{{\bf X} \mid {\bf Z}}({\bf x} \mid {\bf z}, \btheta({\bf z}))$.

For the rest of this paper, we will make the independent and identically distributed (IID) assumption.
Given $n$ IID samples of the set of random variables ${\bf X}$, the maximum likelihood estimate (MLE) for $\btheta$ is
\begin{align}
    \bthetahat_n &= \underset{\btheta}{\text{argmax }} \sum_{i=1}^n \log p_{\bf X}({\bf x}_i, \btheta).
    \label{eq:mle}
\end{align}
Under standard regularity conditions listed in \cref{app:regularity_conditions}, we can analytically derive the MLE for a set of variables whose model is a member of the exponential family and show that it is asymptotically normal, meaning $\bthetahat_n$ converges to the true parameters at rate $O_p(n^{-1/2})$.
When the model for ${\bf X}$ given ${\bf Z}$ is a member of the exponential family, one can compute the MLE for $\btheta({\bf z})$ to estimate $\E[{\bf X} \mid {\bf Z} = {\bf z}]$.

{\bf Goal of inference.}
Suppose that we are given an SCM associated with a DAG $\G$ and a set of treatment variables ${\bf A}$ where each $A \in {\bf A}$ may be a cause of the outcome $Y$.
We are also given the PDF or PMF of 
the {\it interventional distribution} $p(Y({\bf a}))$, with this object being a member of the exponential family.
Our goal of inference is to determine ${\bf A}^* \subseteq {\bf A}$ such that each $A \in {\bf A}^*$ is a significant predictor of the outcome $Y$ in the MSM $\E[Y({\bf a})]$ while each $A \in {\bf A} \setminus {\bf A}^*$ is not.
Note that significant predictors of $Y$ in the MSM correspond to significant direct causes of $Y$ in $\G$.

\section{The BIC and Consistency}
\label{sec:bic_and_consistency}

Here, we first review how the BIC arises as an approximation of the log-likelihood of the observed data given a candidate model \citep{ref:konishi2008information}.
Then, we state the theorem and conditions under which the BIC is consistent for choosing the correct model between two candidate ones \citep{ref:haughton1988choice}.

Let ${\bf O}_i$ represent a single observation of all observed variables ${\bf O}$ in the dataset, and let $\bD = ({\bf O}_1, \dots, {\bf O}_n)$ be the $n$ IID observations of each observed variable in the data.
Next, let $M_1, \dots, M_R$ be $R < \infty$ candidate models for the data where each $M_r$ is characterized by the space of some parameters $\btheta_r \in \mathrm{int} \Theta$.
Let $\pi(\btheta_r \mid M_r)$ be the prior distribution of the parameters $\btheta_r$ for the model $M_r$.
Then, the likelihood of the data $\bD$ conditional on the model $M_r$ is given by
\begin{align}
    p(\bD \mid M_r) &= \int p(\bD \mid \btheta_r, M_r) \pi(\btheta_r \mid M_r) d\btheta_r.
    \label{eq:marginal_likelihood}
\end{align}

Let $p(M_r)$ be the prior probability that the model $M_r$ is the true model.
Using Bayes' theorem, the posterior probability of $M_r$ being the true model given the observed data is $p(M_r \mid \bD) = \frac{p(\bD \mid M_r) p(M_r)}{\sum_{j=1}^R p(\bD \mid M_j) p(M_j)}$.

From a Bayesian perspective, we seek to choose the model with the greatest posterior probability conditional on the data.
Assuming that the prior for each model is the same, the Bayesian choice for the best model is the one with the greatest $\log p(\bD \mid M_r)$.
In general, $\log p(\bD \mid M_r)$ can be difficult to calculate; thus, we utilize Laplace's method for approximating integrals to get its finite sample approximation.

\begin{restatable}[BIC Approximation]{proposition}{bicDerivation}\label{prop:bic_derivation}
    Under standard regularity conditions and given a candidate model $M_r$, the log-likelihood of the observed data given $M_r$ can be approximated as
	\begin{align*}
        \log p(\bD \mid M_r) &\approx \log p(\bD \mid \hat{\btheta}_r, M_r) - \frac{d_r}{2} \log n,
	\end{align*}
    where $\hat{\btheta}_r$ is the MLE for the parameters and $d_r$ is the model dimension \citep{ref:konishi2008information}.
\end{restatable}

The proof for \cref{prop:bic_derivation} and all other proofs are presented in \cref{app:proofs}.
We further provide a proof intuition for \cref{prop:bic_derivation} in \cref{app:bic_derivation_intuition}.
The BIC is defined exactly as this finite sample approximation of the observed data log-likelihood.

\begin{restatable}[BIC]{definition}{bicDefinition}\label{def:bic_definition}
    Let $M_r$ be a model characterized by the space of parameters $\btheta_r$ for the random variables ${\bf X}$ given ${\bf Z}$, where ${\bf X}$ and ${\bf Z}$ are disjoint and ${\bf X} \cup {\bf Z} \subseteq {\bf O}$.
    Let $({\bf X}, {\bf Z})_i$ represent a single observation of all variables in ${\bf X} \cup {\bf Z}$, and let $\D = (({\bf X}, {\bf Z})_1, \dots, ({\bf X}, {\bf Z})_n)$ be the dataset of $n$ IID observations of the variables ${\bf X} \cup {\bf Z}$.
    The BIC score is defined as
    \begin{align}
        \gamma(\D, M_r) &= \log p(\D \mid \hat{\btheta}_r, M_r) - \frac{d_r}{2} \log n,
        \label{eq:bic_definition}
    \end{align}
    where $\hat{\btheta}_r$ is the MLE estimates for the parameters in the model space $M_r$, and $d_r$ is the model dimension.
    Larger values suggest more likely models.
\end{restatable}

\begin{remark}
    The BIC score is often also defined as
    \begin{align*}
        -2 \log p(\D \mid \hat{\btheta}_r, M_r) + d_r \log n.
    \end{align*}
    This is the same as \cref{eq:bic_definition} after multiplying by $-2$.
    Here, smaller values suggest more likely models.
\end{remark}

Let $\btheta^*$ be the true values of the parameters for ${\bf X}$ given ${\bf Z}$.
We now define the notion of a {\it consistent} scoring criterion and state the theorem that the BIC is consistent when ${\bf Z}=\emptyset$.

\begin{restatable}[Consistent Scoring Criterion]{definition}{correctChoice}\label{def:correct_model_choice}
    Let $\M = M_1, \dots, M_R$, where $1 < R < \infty$, be a set of models each characterized by $\btheta_r \in \mathrm{int} \Theta$.
    Suppose that $\btheta^* \in \mathrm{int} \Theta$ and that $\btheta^* \in M_r$ for at least one $r$.
    Let $M_1$ and $M_2$ be any pair of models in $\M$.
    A scoring criterion $S(\D, M_r)$ is {\it consistent} if (i) $\btheta^* \in M_2 \setminus M_1$ or (ii) $\btheta^* \in M_2 \cap M_1$ with $d_2 < d_1$, then
    \begin{align*}
        \lim_{n \rightarrow \infty} P(S(\D, M_1) < S(\D, M_2)) = 1.
    \end{align*}
\end{restatable}

\begin{restatable}[Consistency of BIC \citep{ref:haughton1988choice}]{theorem}{bicConsistency}\label{thm:bic_consistency}
    Let ${\bf Z}=\emptyset$ and ${\bf X}$ be a set of random variables, and suppose each model for ${\bf X}$ in $\M$ is a member of the exponential family.
    Then, the BIC is a consistent scoring criterion.
\end{restatable}

In \cref{app:bic_consistency_intuition}, we give a high-level overview for why \cref{thm:bic_consistency} is true.
\cref{thm:bic_consistency} and an application of Bayes' rule also allows us to conclude consistency of the BIC for conditional distributions, which we formalize in the following corollary.

\begin{restatable}[Consistency of BIC for Conditional Distributions]{corollary}{bicConditionalConsistency}\label{cor:bic_consistency_conditional}
    Let ${\bf X}$ and ${\bf Z}$ be two disjoint sets of random variables, and suppose that each model for ${\bf X}$ given ${\bf Z}$ in $\M$ is a member of the exponential family.
    Then, the BIC is a consistent scoring criterion.
\end{restatable}

\section{Reweighted BIC}
\label{sec:reweighted_bic}

We will begin this section by revisiting \cref{exp:example_backdoor,exp:example_frontdoor} to motivate the
reweighted BIC (rBIC).
Then, we will state the rBIC and show that it is consistent with both oracle and estimated weights.
Finally, we will propose an algorithm that utilizes the rBIC for selecting direct causes of the outcome in an MSM.

\begin{example}[continues=exp:example_frontdoor]
    The only observed variables in the researcher's dataset are the potential treatments, the mediator, and the outcome. 
    If the researcher compares models for the conditional mean $\E[Y \mid {\bf A}]$ using the observed data directly, variables in ${\bf A}$ may be significant in the chosen model even if their dependencies with $Y$ are non-causal\footnote{Please see \cref{app:non_causal_dependence} for a concrete example.}.
    Such non-causal associations arise due to {\it confounding bias}.
\end{example}

\begin{example}[continues=exp:example_backdoor]
    Failure to account for the covariates ${\bf C}$ raises the same issues regarding confounding bias as in \cref{exp:example_frontdoor}.
    One common approach to correct for this confounding bias is to compare conditional mean models $\E[Y \mid {\bf A}, {\bf C}]$ instead of $\E[Y \mid {\bf A}]$.
    However, in general, whether a variable is significant in a model for $\E[Y \mid {\bf A}, {\bf C}]$ and for $\E[Y({\bf a})]$ can be different\footnote{This phenomenon is also known as non-collapsibility. Please see \cref{app:different_significance} for a concrete example.}.
    The target of inference is to find which variables in ${\bf A}$ are significant in the MSM, so we must somehow select ${\bf A}^*$ using models that do not condition on ${\bf C}$.
\end{example}

In both of the examples above, the researchers would like to assess the significances of potential predictor variables in the MSM, but, regardless of whether covariates are observed, this task is hindered by the presence of confounding bias.
In the special case where the graph representing the data-generating process satisfies a property called {\it marginal unconfoundedness}, however, we can use observed associations to conclude causality.
We formalize this below.

\begin{restatable}[Marginal Unconfoundedness]{definition}{marginalUnconfoundedness}\label{def:confounding}
    Potential treatments ${\bf A}$ and an outcome $Y$ are {\it marginally unconfounded} in a DAG $\G$ if $Y \ci_{d-sep} {\bf A}$ in $\G$ when all outgoing edges from ${\bf A}$ are removed\footnote{Note that this definition is isomorphic to the criterion that $Y({\bf a}) \ci_{d-sep} {\bf A}$ in a single-world intervention graph (SWIG) \citep{ref:richardson2013single}. We define marginal unconfoundedness using edge deletions to simplify the presentation.
    }.
\end{restatable}

\begin{restatable}[Association is Causation]{proposition}{treatmentOutcomeSufficient}\label{prop:treatment_outcome_sufficient}
    Suppose there is an SCM associated with a DAG $\G$ on the variables ${\bf V}$.
    If ${\bf A}$ and $Y$ are marginally unconfounded in $\G$ and consistency holds, then $\E[Y({\bf a})] = \E[Y \mid {\bf A} = {\bf a}]$.
\end{restatable}

Thus, when treatments and outcome are marginally unconfounded, selecting significant predictors of $Y$ in the conditional mean model $\E[Y \mid {\bf A}]$ is equivalent to selecting ${\bf A}^*$ in the MSM $\E[Y({\bf a})]$.
If the assumption that the PDF or PMF of $Y$ conditional on ${\bf A}$ is a member of the exponential family also holds, we can simply employ the BIC to select ${\bf A}^*$ in the MSM.
Most observed distributions, however, do not satisfy marginal unconfoundedness, but reweighting is a common method used to obtain a marginally unconfounded distribution from a confounded one, provided a sufficiently rich set of confounding or mediating variables are observed \citep{ref:hernan2000marginal}.

Let $\G$ be the DAG associated with the SCM of the observed data.
Consider simulating treatment variables ${\bf A}'$ that have the same support as ${\bf A}$ from a user-specified distribution $p^*({\bf A}')$.
We seek weights $W({\bf O}, {\bf A}')$ that yield a {\it target reweighted distribution} $q({\bf A}', Y)$, defined as a distribution where the causal relationship between ${\bf A}'$ and ${\bf Y}$ is the same as that between ${\bf A}$ and ${\bf Y}$ in $p({\bf O})$, but ${\bf A}'$ and $Y$ are marginally unconfounded.
Specifically, these weights must satisfy $W({\bf O}, {\bf A}') p({\bf O}) = q({\bf A}', Y)$.
We now propose weights that will be valid for any $\G$ where $p(Y({\bf a}))$ is identified\footnote{For a general theory on graphs where $p(Y({\bf a}))$ is identified, we refer the reader to \cite{ref:shpitser2006identification,ref:huang2006identifiability}.}.
Let the target reweighted distribution be
\begin{align*}
    q({\bf A}'= {\bf a}', Y=y) &= p(Y({\bf A} = {\bf a}')=y) p^*({\bf A}'={\bf a}').
\end{align*}
By construction, causal relationships between ${\bf A}'$ and $Y$ are preserved in any distribution of this form, and marginal unconfoundedness is satisfied.
To see this, first notice that ${\bf A}'$ is simulated and is, thus, randomized.
Next, consider the distribution of $Y({\bf A} = {\bf a}')$.
This is the value of $Y$ in the observed distribution had we intervened on ${\bf A}$ to be the observed values of the simulated treatments ${\bf A}'$.
For a discussion on the graph corresponding to this target reweighted distribution $q({\bf A}', Y)$, please see \cref{app:graph_of_reweighted_dist}.

The weights that move the observed distribution to the target reweighted distribution are thus $W({\bf O}, {\bf A}') = \frac{q({\bf A}', Y)}{p({\bf O})}$, and multiple representations of this expression are possible; see Algorithm 2 in \cite{ref:bhattacharya2022semiparametric} for discussions on inverse-probability weighted versions of $p(Y({\bf a}))$ in graphs where it is identified.
For shorthand, we will refer to the weights as simply $W$.
We impose the constraint that, for all possible values of $({\bf O}, {\bf A}')$, $0 < W < \infty$\footnote{This is a more general version of the {\it positivity} assumption common in the causal inference literature.}.
Although these weights are valid in general, they are usually a ratio of products of many conditional distributions, which may be difficult to estimate in practice.
For instance, one representation of the weights valid for \cref{fig:example_dag_backdoor} is 
\begin{align*}
    W = \frac{\sum_C p(Y \mid {\bf A}', {\bf C}) p({\bf C}) p^*({\bf A}')}{p(Y \mid {\bf A}, {\bf C}) p({\bf A} \mid {\bf C}) p({\bf C})},
\end{align*}
which requires estimating two distinct nuisance functions.

For ease of exposition in the rest of this paper, we will consider target 
reweighted distributions of the form $q({\bf V}, {\bf A}')$, where we sometimes set ${\bf A}' \equiv {\bf A}$, that utilize 
weights that are simple functions of conditional distributions.
It is not always obvious how to find such $q({\bf V}, {\bf A}')$ target distributions with simple weights, but, in \cref{app:procedure_find_weights}, we present a heuristic procedure for finding these distributions and their corresponding weights.
We show that this procedure is valid for the DAGs shown in \cref{fig:examples} as well as the so-called napkin graph \citep{ref:pearl2018book,ref:guo2025causal}.
A general characterization of identified functionals using inverse weights with a simple form is complex and beyond the scope of this work.

\begin{example}[continues=exp:example_backdoor]
    When all confounders are observed, the researcher can use the weights $W = \frac{p^*({\bf A})}{p({\bf A} \mid {\bf C})}$, where $p^*({\bf A})$ is a user-defined distribution for the potential treatments.
    Here, we set ${\bf A}' \equiv {\bf A}$.
    To see why these weights are valid, consider multiplying them onto the DAG factorization of \cref{fig:example_dag_backdoor}:
    \begin{align*}
        &\frac{p^*({\bf A})}{p({\bf A} \mid {\bf C})} \times p(Y \mid {\bf A}, {\bf C}) p({\bf A} \mid {\bf C}) p({\bf C}) \\
        &= p(Y \mid {\bf A}, {\bf C}) p^*({\bf A}) p({\bf C}) = q(Y, {\bf A}, {\bf C}).
    \end{align*}
    $q(Y, {\bf A}, {\bf C})$ corresponds to a target reweighted distribution where marginal unconfoundedness is satisfied.
\end{example}

\begin{example}[continues=exp:example_frontdoor]
    The researcher can use the weights $W = \frac{p^*({\bf A}') p(M \mid {\bf A}')}{p(M \mid {\bf A})}$, to obtain a reweighted distribution $q({\bf A}, {\bf A}', M, Y)$ where the simulated treatments ${\bf A}'$ satisfy marginal unconfoundedness with respect to the outcome.
    At a high level, these weights work by first breaking the dependence of $M$ on the confounded treatments ${\bf A}$ and then inducing dependence of $M$ on the simulated treatments ${\bf A}'$\footnote{Due to space constraints, please see \cref{app:weights_frontdoor} for more details on why these weights yield the desired reweighted distribution $q({\bf V}, {\bf A}')$.}.
\end{example}

Now, the BIC can only compare models with respect to the observed distribution $p({\bf O})$.
To compare models with respect to a target reweighted distribution $q({\bf V}, {\bf A}')$, we need a reweighted BIC (rBIC).
We now define and prove consistency for the rBIC when the ground-truth values of the weights, or oracle values, are known.
Afterwards, we will consider the case where the weights must be estimated from data.

\begin{restatable}[Oracle rBIC]{definition}{reweightedBicDefinition}\label{def:reweighted_bic}
    Let $M$ be a model characterized by the parameters $\btheta_W$ for the random variables ${\bf X}$ given ${\bf Z}$ in a target reweighted distribution $q({\bf V}, {\bf A}')$.
    Let $\W = (W_1, \dots, W_n)$ be the oracle weights computed for each of the $n$ observations normalized such that $\sum_{i=1}^n W_i = n$\footnote{A vector of arbitrary weights $\W$ can be normalized to sum to $n$ by multiplying each $W_i$ by $n / (\sum_{i=1}^n W_i)$.}.
    Let $P = \sum_{i=1}^n W_i \log p({\bf X}_i \mid {\bf Z}_i, \hat{\btheta}_W({\bf Z}_i), M_r)$, where $\hat{\btheta}_W$ are the reweighted MLE estimates\footnote{For more details on how to compute a reweighted MLE, please see \cref{app:reweighted_mle}.}.
    The rBIC is defined as
    \begin{align}
        \gamma_W(\D, \W, M_r) &= P - \frac{d_r}{2} n^{\alpha_d},
        \label{eq:reweighted_bic}
    \end{align}
    where $0 < \alpha_d < 1$ is the {\it scaling parameter} of the penalty term.
\end{restatable}

An important change from the classical BIC to the rBIC is that the scaling of the penalty term is modified from $\log n$ to $n^{\alpha_d}$.
The reason for this change will become apparent when we discuss using the estimated weights for the rBIC rather than the oracle weights.

Let $\btheta^*_W$ be the true values of the parameters for ${\bf X}$ given ${\bf Z}$ in a target reweighted distribution $q({\bf V}, {\bf A}')$.
We now state the theorem for consistency of the rBIC.

\begin{restatable}[Consistency of rBIC]{theorem}{reweightedBicConsistency}\label{thm:reweighted_bic_consistency}
    Let ${\bf X}$ and ${\bf Z}$ be two disjoint sets of random variables, and suppose that each model for ${\bf X}$ given ${\bf Z}$ in $\M$ with respect to the reweighted distribution is a member of the exponential family.
    Then, the rBIC is a consistent scoring criterion.
\end{restatable}

Note that ${\bf Z} = \emptyset$ corresponds to the case described in \cref{thm:bic_consistency}.
We give a proof intuition for \cref{thm:reweighted_bic_consistency} in \cref{app:reweighted_bic_consistency_intuition}.
The key drawback of the oracle rBIC is that it requires access to the ground-truth weights.
In practice, we often do not have access to such information and must fit nuisance models to estimate the weights $\W$ from data.
Let $\hat{\W} = (\hat{W}_1, \dots, \hat{W}_n)$ be the set of estimated weights for $\W$.
We now show that, under flexible assumptions regarding the rate of convergence of $\hat{\W}$, the rBIC using the estimated weights is still consistent for model selection.

\begin{restatable}[Consistency of rBIC with Estimated Weights]{theorem}{consistencyEstimatedWeights}\label{thm:consistency_estimated_weights}
    Suppose that
    \begin{align}
        \bigg( \frac{1}{n} \sum_{i=1}^n ( \hat{W}_i - W_i )^2 \bigg)^{1/2} = o_p(n^{-\alpha_W}). \label{eq:weights_converge}
    \end{align}
    That is, the estimated weights are $o_p(n^{-\alpha_W})$-consistent in root-mean squared error (RMSE).
    If $\alpha_W > 0$ and $1-\alpha_W < \alpha_d$, then, under the same conditions as those of \cref{thm:reweighted_bic_consistency}, the rBIC with estimated weights is a consistent scoring criterion.
\end{restatable}

The intuition for the choices of $\alpha_W$ and $\alpha_d$ is as follows.
First, $\gamma_W(\D, \hat{\W}, M_r) = \gamma_W(\D, \hat{\W}, M_r) - \gamma_W(\D, \W, M_r) + \gamma_W(\D, \W, M_r)$.
In the difference of the first two terms, the penalty terms cancel.
Next, the maximized log-likelihood, which is a function of the reweighted mean in exponential family distributions, grows at rate $O_p(n)$.
Since the reweighted mean with estimated weights converges to the true reweighted mean at rate $O_p(n^{-\alpha_W})$, this results in the difference in log-likelihoods being $O_p(n^{1-\alpha_W})$.
To make sure that the penalty term of the remaining $\gamma_W(\D, \W, M_r)$ can dominate this growing term, we choose $\alpha_d > 1-\alpha_W$.

Valid choices for $\alpha_d$ depend on how quickly estimated weights converge to the oracle weights.
For instance, if we assume parametric models for the weights that are typically $o_P(n^{-1/2})$-consistent in RMSE, we can choose $\alpha_d=1/2+\epsilon \approx 1/2$, where $\epsilon=10^{-10}$. 
On the other hand, flexible machine learning methods under sparsity constraints can attain $o_p(n^{-1/4})$-consistency in RMSE \citep{ref:chernozhukov2018double}.
In such scenarios, $\alpha_d=3/4+\epsilon \approx 3/4$ is an appropriate choice.

\subsection{Algorithm for Selecting Direct Causes}
\label{sec:algorithm}

\begin{algorithm}[ht]
\caption{Infer ${\bf A}^*$ for the MSM $\E[Y({\bf a})]$.}\label{alg:search}
\begin{algorithmic}[1]
    \State ${\bf A}^* \gets \emptyset$
    \State $M^* \gets$ a model for $Y$ conditional on ${\bf A}^*$
    \State $\hat{\W} \gets$ estimate for $\W$ from nuisance models
    \For{each $A \in {\bf A}$}
        \State $M' \gets$ a model for $Y$ conditional on ${\bf A}^* \cup \{A\}$
        \If{$\gamma_W(\D, \hat{\W}, M') > \gamma_W(\D, \hat{\W}, M^*)$}
            \State ${\bf A}^* \gets {\bf A}^* \cup \{A\}$
            \State $M^* \gets M'$
        \EndIf
    \EndFor
    \For{each $A \in {\bf A}^*$}
        \State $M' \gets$ a model for $Y$ conditional on ${\bf A}^* \setminus \{A\}$
        \If{$\gamma_W(\D, \hat{\W}, M') > \gamma_W(\D, \hat{\W}, M^*)$}
            \State ${\bf A}^* \gets {\bf A}^* \setminus \{A\}$
            \State $M^* \gets M'$
        \EndIf
    \EndFor
    \State return ${\bf A}^*$
\end{algorithmic}
\end{algorithm}

We propose \cref{alg:search} -- following the greedy equivalence search (GES) algorithm \citep{ref:chickering2002optimal} -- to select the variables ${\bf A}^*$ that are significant in the MSM $\E[Y({\bf a})]$\footnote{For an analysis of the time and space complexities of \cref{alg:search}, please refer to \cref{app:complexity_search_alg}.}.
Importantly, the validity of \cref{alg:search} relies on the rBIC and its consistency property.

The assumptions of \cref{alg:search} are as follows: (i) samples are IID, (ii) $Y({\bf a})=Y$ if ${\bf A}={\bf a}$, (iii) the SCM of the data is associated with a DAG $\G$, (iv) the PDF or PMF of the interventional distribution $p(Y({\bf a}))$ is a member of the exponential family, (v) based on $\G$, the oracle weights satisfy $W p({\bf V}) = q({\bf V}, {\bf A}')$, (vi) for all possible values of $({\bf O}, {\bf A}')$, $0 < W < \infty$, (vii) the estimated weights $\hat{W}$ satisfy \cref{eq:weights_converge} for some $\alpha_W > 0$, and (viii) $\alpha_d$ is appropriately chosen given $\alpha_W$.
Under the above assumptions and Lemmas 9 and 10 of \cite{ref:chickering2002optimal}, \cref{alg:search} will select the correct ${\bf A}^*$ as the sample size goes to infinity.

\section{Simulation Study}
\label{sec:simulation_study}

\begin{figure*}[t]
    \begin{center}
        \begin{subfigure}[t]{0.49\textwidth}
            \centering
            \includegraphics[align=c,bmargin=6pt,scale=0.45]{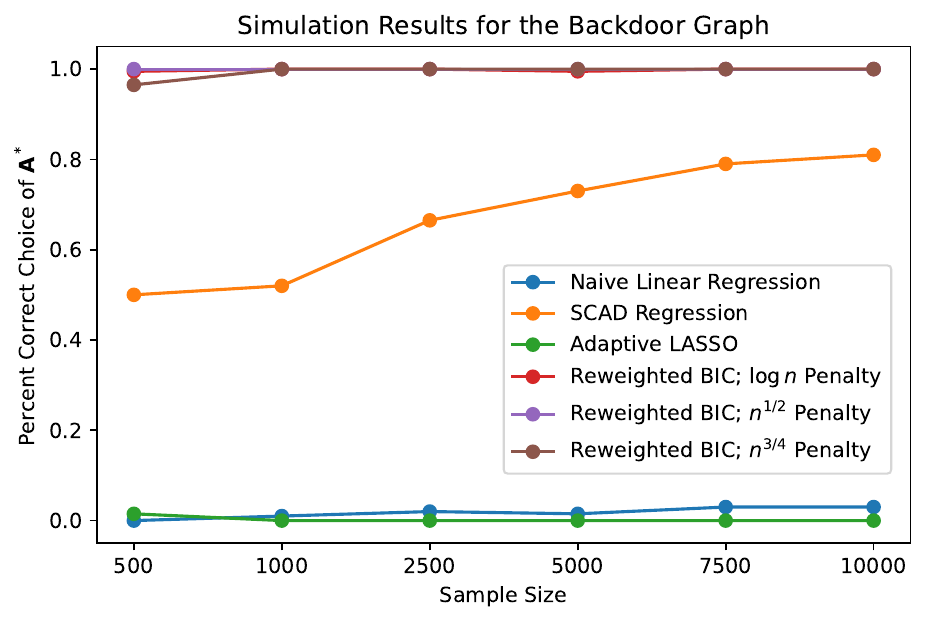}
            \label{fig:backdoor_results}
        \end{subfigure}
        \begin{subfigure}[t]{0.49\textwidth}
            \centering
            \includegraphics[align=c,bmargin=6pt,scale=0.45]{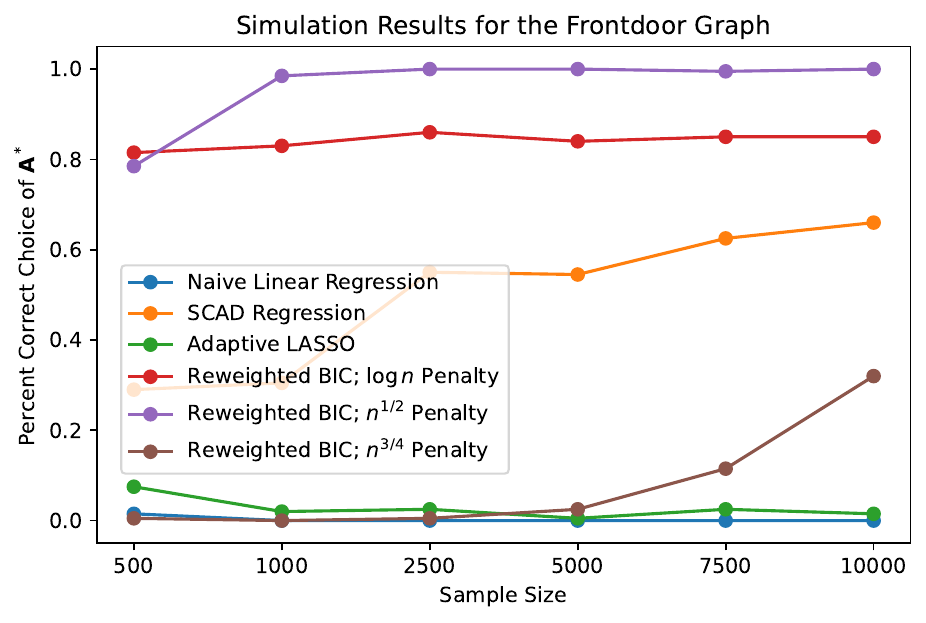}
            \label{fig:frontdoor_results}
        \end{subfigure}
    \end{center}
    \caption{Simulation results for the backdoor graph and frontdoor graph settings with estimated weights.}
    \label{fig:results}
\end{figure*}

We now present a simulation study assessing how reliably the rBIC and alternative methods for model selection with estimated weights can select ${\bf A}^*$ at various sample sizes in two DGPs\footnote{All code implementing the simulations and for generating the figures is publicly available at \url{https://github.com/jacobmchen/reweighted_bic}.}.
The DGPs generate data according to the DAGs shown in \cref{fig:example_dag_backdoor,fig:example_dag_frontdoor}.
In both simulations, $A_1$ and $A_3$ are significant predictors of the outcome in the MSM while $A_2$ is not.
Variables in ${\bf A}$ are Bernoulli random variables while all other variables are normally distributed.
We use a single confounding variable in both simulations, regardless of whether it is observed.
For precise details on the two DGPs, please see \cref{app:sims_dgp}.

After generating the data, we trial \cref{alg:search} with three choices for $\alpha_d$ in the rBIC: $\log n$, $n^{1/2}$, and $n^{3/4}$.
If the output of \cref{alg:search} satisfies ${\bf A}^* = \{A_1, A_3\}$, then we count it as a correct choice of ${\bf A}^*$.
Otherwise, it is an incorrect choice.
In both settings, the weights are estimated using parametric models.
Details on estimation of the weights are given in \cref{app:sims_weights}.

We then compare our method against three alternatives: ordinary least squares (OLS) regression, adaptive LASSO, and linear regression with the SCAD penalty.
The latter two methods are estimated using gradient descent with a reweighted loss function.
Following \cite{ref:jaman2025penalized}, an estimated parameter from the regression methods is considered to be $0$ if its absolute value is less than $0.001$.
A potential cause $A_i$ is included in a regression method's selected ${\bf A}^*$ if its estimated parameter is non-zero, and the criterion for correct choice of ${\bf A}^*$ is the same as above.
We run our experiments at sample sizes ranging from $500$ to $10000$, and, at each sample size, we run $200$ repetitions.

We chose adaptive LASSO and regression with the SCAD penalty to compare against because they possess the {\it oracle property}.
This property states that the methods will, with probability one as the sample size tends to infinity, output zero for the parameters that are zero in ground-truth \citep{ref:zou2006adaptive,ref:fan2001variable}.
More details on these regression methods and their implementations are given in \cref{app:sims_alasso_scad}.

The results of the simulations are shown in \cref{fig:results}.
In both settings, regression with the SCAD penalty achieves the best results out of the regression methods, and it becomes increasingly accurate as the sample size increases.
Adaptive LASSO and OLS are inaccurate at all trialed sample sizes across both DGPs.

In the backdoor graph, all three trialed scalings of the penalty term are accurate in selecting ${\bf A}^*$.
This is because the weights in this setting are stable, and the estimated weights converge to the oracle weights quickly.
In the frontdoor graph, on the other hand, the rBIC with the $\log n$ penalty term achieves roughly $80\%$ accuracy in selecting ${\bf A}^*$, even at large sample sizes.
The $n^{1/2}$ penalty is accurate starting at sample size $1000$, and the accuracy of the $n^{3/4}$ penalty increases at large sample sizes, which aligns with our asymptotic theory.

The weights in the frontdoor graph are a ratio of densities of a normally distributed random variable, so small deviations from the ground-truth densities propagate into large differences in the estimated weights.
Furthermore, division by small numbers occasionally results in numerical instability.
These observations explain why the rBIC with the $\log n$ penalty is inconsistent in the frontdoor graph.
The $n^{1/2}$ and $n^{3/4}$ penalties, however, remain consistent at large samples, despite the issues with the weights.
Simulation results under oracle weights are shown in \cref{app:results_oracle_weights}.

\section{Conclusion}
\label{sec:conclusion}

In this paper, we proposed a novel generalization of the Bayesian information criterion, the rBIC, that is applicable to reweighted distributions arising in causal inference and missing data, and demonstrated how it can be applied to the task of selecting significant causes of an outcome in a marginal structural model.
We proved that the rBIC is consistent for model selection with respect to the target reweighted distribution and showed how to adjust its penalty term to retain its consistency property when weights are estimated from data using nuisance functions.
Finally, we concluded with a simulation study that demonstrated the asymptotic properties of the rBIC.

There are many potential avenues for future work.
Given a DAG with unmeasured variables where the average causal effect is identified, is there a general theory for deriving ``simple'' weights that yield a marginally unconfounded distribution?
Can the rBIC consistently select models in the presence of missing data where weights for removing missing data bias are derived from a missingness graph \citep{ref:mohan2021graphical}?
Is it possible to use the rBIC to learn Verma constraints in nested Markov models \citep{ref:richardson2023nested}?
These open questions are inspirations for future investigation.

\newpage

\subsection*{AI Use Statement}

In this work, we used generative AI tools as a supplement to textbook learning of foundational mathematical and statistical concepts such as maximum likelihood estimators and Lipschitz continuity.
These topics are ingredients in our proofs of mathematical claims, and all outputs of generative AI tools on this topic were cross-checked with textbooks such as \cite{ref:van2000asymptotic}.
Additionally, we used generative AI tools to generate template code for creating scientific figures and images, but all final code was written by a human. 
We also used generative AI tools to give suggestions for rephrasing select sentences and the title of the paper; less than five sentences in total throughout this research paper were rephrased with the assistance of generative AI.

We have not used generative AI tools to generate synthetic data sets, help develop theoretical models or conceptual frameworks, formulate mathematical claims, assist in the writing of proofs, propose or refine hypotheses, design or provide feedback on research methodology or experiments, implement methods, assist with translation, clean or reformat datasets, support qualitative or thematic data analysis, nor interpret results.

The authors wrote this entire document, including the main text of the paper, the entirety of the Appendix, and this AI use statement.
The authors also wrote, with some AI assistance that provided template code for figure generation, the entirety of the accompanying code.
None of the text in this document was copy-pasted from the output of a generative AI tool or model.
We have reviewed all AI-assisted work.
We take responsibility for the final content of this work.

\subsection*{Acknowledgements}

The authors are grateful for fruitful discussions with (in alphabetical order by first name) Beatrix Wen, Helen Guo, and Trung Phung in the early stages of this project. 
This work was supported by the National Institute of Health (Grant No. 1-R01-LM-014800-01A1), the Office of Naval Research (Grant No. N000142412701), and the National Institute of Health (Grant No. 1-R56-AI-0191526-01).

\bibliography{bibliography}

@article{ref:chickering2002optimal,
  title={Optimal structure identification with greedy search},
  author={Chickering, David Maxwell},
  journal={Journal of machine learning research},
  volume={3},
  number={Nov},
  pages={507--554},
  year={2002}
}

@article{ref:haughton1988choice,
  title={On the choice of a model to fit data from an exponential family},
  author={Haughton, Dominique MA},
  journal={The annals of statistics},
  pages={342--355},
  year={1988},
  publisher={JSTOR}
}

@article{ref:schwarz1978estimating,
  title={Estimating the dimension of a model},
  author={Schwarz, Gideon},
  journal={The annals of statistics},
  pages={461--464},
  year={1978},
  publisher={JSTOR}
}

@misc{ref:hernan2000marginal,
  title={Marginal structural models to estimate the causal effect of zidovudine on the survival of HIV-positive men},
  author={Hern{\'a}n, Miguel {\'A}ngel and Brumback, Babette and Robins, James M},
  journal={Epidemiology},
  volume={11},
  number={5},
  pages={561--570},
  year={2000},
  publisher={LWW}
}

@misc{ref:robins2000marginal,
  title={Marginal structural models and causal inference in epidemiology},
  author={Robins, James M and Hernan, Miguel Angel and Brumback, Babette},
  journal={Epidemiology},
  volume={11},
  number={5},
  pages={550--560},
  year={2000},
  publisher={Lww}
}

@article{ref:bahamyirou2022doubly,
  title={Doubly robust adaptive {LASSO} for effect modifier discovery},
  author={Bahamyirou, Asma and Schnitzer, Mireille E and Kennedy, Edward H and Blais, Lucie and Yang, Yi},
  journal={The International Journal of Biostatistics},
  volume={18},
  number={2},
  pages={307--327},
  year={2022},
  publisher={De Gruyter}
}

@article{ref:jaman2025penalized,
  title={Penalized {G}-estimation for effect modifier selection in a structural nested mean model for repeated outcomes},
  author={Jaman, Ajmery and Wang, Guanbo and Ertefaie, Ashkan and Bally, Mich{\`e}le and L{\'e}vesque, Ren{\'e}e and Platt, Robert W and Schnitzer, Mireille E},
  journal={Biometrics},
  volume={81},
  number={1},
  pages={ujae165},
  year={2025},
  publisher={Oxford University Press}
}

@article{ref:zou2006adaptive,
  title={The adaptive {LASSO} and its oracle properties},
  author={Zou, Hui},
  journal={Journal of the American statistical association},
  volume={101},
  number={476},
  pages={1418--1429},
  year={2006},
  publisher={Taylor \& Francis}
}

@article{ref:tibshirani1996regression,
  title={Regression shrinkage and selection via the {LASSO}},
  author={Tibshirani, Robert},
  journal={Journal of the Royal Statistical Society Series B: Statistical Methodology},
  volume={58},
  number={1},
  pages={267--288},
  year={1996},
  publisher={Oxford University Press}
}

@article{ref:fan2001variable,
  title={Variable selection via nonconcave penalized likelihood and its oracle properties},
  author={Fan, Jianqing and Li, Runze},
  journal={Journal of the American statistical Association},
  volume={96},
  number={456},
  pages={1348--1360},
  year={2001},
  publisher={Taylor \& Francis}
}

@article{ref:magruder2015nadir,
  title={Nadir oxygen delivery on bypass and hypotension increase acute kidney injury risk after cardiac operations},
  author={Magruder, J Trent and Dungan, Samuel P and Grimm, Joshua C and Harness, H Lynn and Wierschke, Chad and Castillejo, Stephen and Barodka, Viachaslau and Katz, Nevin and Shah, Ashish S and Whitman, Glenn J},
  journal={The Annals of thoracic surgery},
  volume={100},
  number={5},
  pages={1697--1703},
  year={2015},
  publisher={Elsevier}
}

@article{ref:ranucci2018goal,
  title={Goal-directed perfusion to reduce acute kidney injury: a randomized trial},
  author={Ranucci, Marco and Johnson, Ian and Willcox, Timothy and Baker, Robert A and Boer, Christa and Baumann, Andreas and Justison, George A and De Somer, Filip and Exton, Paul and Agarwal, Seema and others},
  journal={The Journal of thoracic and cardiovascular surgery},
  volume={156},
  number={5},
  pages={1918--1927},
  year={2018},
  publisher={Elsevier}
}

@book{ref:pearl2014probabilistic,
  title={Probabilistic reasoning in intelligent systems: networks of plausible inference},
  author={Pearl, Judea},
  year={2014},
  publisher={Elsevier}
}

@book{ref:pearl2009causality,
  title={Causality},
  author={Pearl, Judea},
  year={2009},
  publisher={Cambridge university press}
}

@book{ref:lehmann2005testing,
  title={Testing statistical hypotheses},
  author={Lehmann, Erich Leo and Romano, Joseph P},
  year={2005},
  publisher={Springer}
}

@book{ref:cramer1999mathematical,
  title={Mathematical methods of statistics},
  author={Cram{\'e}r, Harald},
  volume={9},
  year={1999},
  publisher={Princeton university press}
}

@book{ref:konishi2008information,
  title={Information criteria and statistical modeling},
  author={Konishi, Sadanori and Kitagawa, Genshiro},
  year={2008},
  publisher={Springer}
}

@article{ref:greco2020weighted,
  title={Weighted likelihood mixture modeling and model-based clustering},
  author={Greco, Luca and Agostinelli, Claudio},
  journal={Statistics and Computing},
  volume={30},
  number={2},
  pages={255--277},
  year={2020},
  publisher={Springer}
}

@article{ref:vansteelandt2014structural,
  title={Structural nested models and {G}-estimation: the partially realized promise},
  author={Vansteelandt, Stijn and Joffe, Marshall},
  journal={Statistical science},
  pages={707--731},
  year={2014},
  publisher={JSTOR}
}

@book{ref:van2000asymptotic,
  title={Asymptotic statistics},
  author={Van der Vaart, Aad W},
  volume={3},
  year={2000},
  publisher={Cambridge university press}
}

@incollection{ref:pearl2022causal,
  title={Causal diagrams for empirical research (with Discussions)},
  author={Pearl, Judea},
  booktitle={Probabilistic and causal inference: The works of Judea Pearl},
  pages={255--316},
  year={2022}
}

@book{ref:barndorff2014information,
  title={Information and exponential families: in statistical theory},
  author={Barndorff-Nielsen, Ole},
  year={2014},
  publisher={John Wiley \& Sons}
}

@article{ref:richardson2023nested,
  title={Nested {M}arkov properties for acyclic directed mixed graphs},
  author={Richardson, Thomas S and Evans, Robin J and Robins, James M and Shpitser, Ilya},
  journal={The Annals of Statistics},
  volume={51},
  number={1},
  pages={334--361},
  year={2023},
  publisher={Institute of Mathematical Statistics}
}

@inproceedings{ref:shpitser2006identification,
  title={Identification of joint interventional distributions in recursive semi-{M}arkovian causal models},
  author={Shpitser, Ilya and Pearl, Judea},
  booktitle={AAAI},
  pages={1219--1226},
  year={2006}
}

@article{ref:richardson2013single,
  title={Single world intervention graphs ({SWIGs}): A unification of the counterfactual and graphical approaches to causality},
  author={Richardson, Thomas S and Robins, James M},
  journal={Center for the Statistics and the Social Sciences, University of Washington Series. Working Paper},
  volume={128},
  number={30},
  pages={2013},
  year={2013},
  publisher={Citeseer}
}

@misc{ref:chernozhukov2018double,
  title={Double/debiased machine learning for treatment and structural parameters},
  author={Chernozhukov, Victor and Chetverikov, Denis and Demirer, Mert and Duflo, Esther and Hansen, Christian and Newey, Whitney and Robins, James},
  year={2018},
  publisher={Oxford University Press Oxford, UK}
}

@article{ref:vrieze2012model,
  title={Model selection and psychological theory: a discussion of the differences between the {A}kaike information criterion ({AIC}) and the {B}ayesian information criterion ({BIC}).},
  author={Vrieze, Scott I},
  journal={Psychological methods},
  volume={17},
  number={2},
  pages={228},
  year={2012},
  publisher={American Psychological Association}
}

@article{ref:bodnar2004marginal,
  title={Marginal structural models for analyzing causal effects of time-dependent treatments: an application in perinatal epidemiology},
  author={Bodnar, Lisa M and Davidian, Marie and Siega-Riz, Anna Maria and Tsiatis, Anastasios A},
  journal={American journal of epidemiology},
  volume={159},
  number={10},
  pages={926--934},
  year={2004},
  publisher={Oxford University Press}
}

@article{ref:buikema2018effect,
  title={The effect of delaying initiation with umeclidinium/vilanterol in patients with {COPD}: an observational administrative claims database analysis using marginal structural models},
  author={Buikema, Ami R and Brekke, Lee and Anderson, Amy and Koep, Eleena and Van Voorhis, Damon and Sharpsten, Lucie and Hahn, Beth and Ray, Riju and Stanford, Richard H},
  journal={Multidisciplinary Respiratory Medicine},
  volume={13},
  number={1},
  pages={38},
  year={2018},
  publisher={Springer}
}

@article{ref:chen2008extended,
  title={Extended {B}ayesian information criteria for model selection with large model spaces},
  author={Chen, Jiahua and Chen, Zehua},
  journal={Biometrika},
  volume={95},
  number={3},
  pages={759--771},
  year={2008},
  publisher={Oxford University Press}
}

@article{ref:peterson2026ranked,
  title={A Ranked Sparsity Extension to the {B}ayesian Information Criterion: A Tool for Selecting Variables from Multiple Data Modalities},
  author={Peterson, Ryan A and Bird, Sarah M and Harris, Logan M and Breheny, Patrick J and Cavanaugh, Joseph E},
  year={2026},
  journal={Preprint}
}

@article{ref:mohan2021graphical,
  title={Graphical models for processing missing data},
  author={Mohan, Karthika and Pearl, Judea},
  journal={Journal of the American Statistical Association},
  volume={116},
  number={534},
  pages={1023--1037},
  year={2021},
  publisher={Taylor \& Francis}
}

@article{ref:robins1986new,
  title={A new approach to causal inference in mortality studies with a sustained exposure period—application to control of the healthy worker survivor effect},
  author={Robins, James},
  journal={Mathematical modelling},
  volume={7},
  number={9-12},
  pages={1393--1512},
  year={1986},
  publisher={Elsevier}
}

@article{ref:splawa1990application,
  title={On the application of probability theory to agricultural experiments. {E}ssay on principles. {S}ection 9},
  author={Splawa-Neyman, Jerzy and Dabrowska, Dorota M and Speed, Terrence P},
  journal={Statistical science},
  pages={465--472},
  year={1990},
  publisher={JSTOR}
}

@article{ref:rubin1974estimating,
  title={Estimating causal effects of treatments in randomized and nonrandomized studies.},
  author={Rubin, Donald B},
  journal={Journal of educational Psychology},
  volume={66},
  number={5},
  pages={688},
  year={1974},
  publisher={American Psychological Association}
}

@book{ref:laan2003unified,
  title={Unified methods for censored longitudinal data and causality},
  author={van der Laan, Mark J and Robins, James M},
  year={2003},
  publisher={Springer}
}

@inproceedings{ref:huang2006identifiability,
  title={Identifiability in causal {B}ayesian networks: A sound and complete algorithm},
  author={Huang, Yimin and Valtorta, Marco},
  booktitle={Proceedings of the national conference on artificial intelligence},
  volume={21},
  pages={1149},
  year={2006},
  organization={Menlo Park, CA; Cambridge, MA; London; AAAI Press; MIT Press; 1999}
}

@article{ref:van2004history,
  title={History-adjusted marginal structural models and statically-optimal dynamic treatment regimes},
  author={van der Laan, Mark J and Petersen, Maya L},
  year={2004},
  journal={bepress}
}

@book{ref:pearl2018book,
  title={The book of why},
  author={Pearl, Judea and Mackenzie, Dana},
  year={2018},
  publisher={Allen Lane London}
}

@article{ref:guo2025causal,
  title={Causal Inference with the ``Napkin Graph''},
  author={Guo, Anna and Benkeser, David and Nabi, Razieh},
  journal={arXiv preprint arXiv:2512.19861},
  year={2025}
}

@article{ref:bhattacharya2022semiparametric,
  title={Semiparametric inference for causal effects in graphical models with hidden variables},
  author={Bhattacharya, Rohit and Nabi, Razieh and Shpitser, Ilya},
  journal={Journal of Machine Learning Research},
  volume={23},
  number={295},
  pages={1--76},
  year={2022}
}

\newpage

\section*{Appendix}
\maketitle

\appendix

The appendix is organized as follows.
\cref{app:add_details} gives additional details, background, and examples on topics discussed in the main text.
\cref{app:proofs} gives proofs of propositions and theorems stated in the main text.
Some proofs are preceded by their proof intuitions.
Finally, \cref{app:sims} gives precise details on the simulations described in \cref{sec:simulation_study}.
This includes the exact DGPs or SCMs, precise computation steps for the weights, details on the penalized regression methods we compare against, additional simulation results, and descriptions of the hardware used to run the simulations.

\section{Additional Details, Background, and Examples}
\label[secinapp]{app:add_details}

\subsection{Regularity Conditions for Asymptotic Normality of MLEs}
\label[secinapp]{app:regularity_conditions}

Here, we describe the regularity conditions that guarantee that the MLE of a random variable $X$ with PDF or PMF $p_X(x, \btheta)$ is asymptotically normal and converges to the true value of the parameters at a rate of $O_p(n^{-1/2})$.
Recall that there are $d$ elements in the vector $\btheta$.
All the conditions below can be extended to cases where we have a set of random variables ${\bf X}$ by replacing $X$ with ${\bf X}$.

First, a ``Condition 0'' that has been assumed throughout the paper is that each observation $X_i$ is independent and identically distributed (IID).

{\bf Condition 1}

The support of $X$ does not depend on the parameters $\btheta$ nor the observed value $X=x$.
This allows us to apply Leibniz's integral rule to interchange derivatives and integrals.

{\bf Condition 2}

$\btheta^*$, the true value of the parameters, is in the interior of the parameter space $\Theta$.

{\bf Condition 3}

The notation $\E_{\btheta}[\cdot]$ means that we are taking the expectation with respect to the parameters $\btheta$.
The first and second derivatives of $\log p_X(x, \btheta)$ satisfy

\begin{align*}
    \E_{\btheta} \bigg[\frac{\partial \log p_X(x, \btheta)}{\partial \theta_j} \bigg] &= 0, \text{for } j=1, \dots, d
\end{align*}

and

\begin{align*}
    I_{jk}(\btheta) &= \E_{\btheta} \bigg[ \frac{\partial}{\partial \theta_j} \log p_X(x, \btheta) \frac{\partial}{\partial \theta_k} \log p_X(x, \btheta) \bigg] \\
    &= \E_{\btheta} \bigg[ - \frac{\partial^2}{\partial \theta_j \partial \theta_k} \log p_X(x, \btheta) \bigg], \text{for } j,k=1, \dots, d.
\end{align*}

{\bf Condition 4}

The Fisher information matrix

\begin{align*}
    I(\btheta) &= \E \bigg[ \bigg( \frac{\partial}{\partial \btheta} \log p_X(x, \btheta) \bigg) \bigg( \frac{\partial}{\partial \btheta} \log p_X(x, \btheta) \bigg)^T \bigg]
\end{align*}

is finite and positive definite at $\btheta = \btheta^*$.

{\bf Condition 5}

There exists an open subset $\omega$ of $\Theta$ that contains the true parameter point $\btheta^*$ such that for almost all $X$ the PDF or PMF $p_X(x, \btheta)$ admits all third derivatives $\frac{\partial \log p_X(x, \btheta)}{\partial \theta_j \partial \theta_k \partial \theta_l}$ for all $\btheta \in \omega$.
Furthermore, there exists functions $M_{jkl}(X)$ such that

\begin{align*}
    \bigg\lvert \frac{\partial^3}{\partial \theta_j \partial \theta_k \partial \theta_l} \log p_X(x, \btheta) \bigg\rvert &\leq M_{jkl}(X) \text{ for all } \btheta \in \omega
\end{align*}

where $\E_{\btheta^*}[M_{jkl}(X)] < \infty$ for all $j,k,l$.

For more details on these regularity conditions, we refer readers to \cite{ref:cramer1999mathematical,ref:lehmann2005testing,ref:van2000asymptotic} and the Appendix of \cite{ref:fan2001variable}.

\subsection{Concrete Example of a Marginal Structural Model}
\label[secinapp]{app:msms}

Here, we give a toy example of an MSM to clarify how and when they are used.
MSMs are common in two settings: first, longitudinal studies where many treatments are administered over time and, second, settings where treatments are continuous.
Our example will be a longitudinal study where treatments are administered at two time points.
This example can be easily extended to settings with more treatments and where treatments are continuous.

Consider the DAG shown in \cref{fig:longitudinal_msm} that represents a longitudinal study.
Suppose for now that treatments are binary; an individual either receives a treatment, or they do not.
A treatment $A_0$ is administered at baseline, a second potentially different treatment $A_1$ is administered at time point one, and an outcome $Y$ is observed at time point two.
$L_0$ and $L_1$ represent observed covariates used to make treatment decisions at baseline and at time point one, respectively.
$H_0$, $H_1$, and $H_2$ represent unmeasured characteristics of the individual at baseline, time point one, and time point two, respectively.

\begin{figure}[ht]
    \begin{center}
        \centering
        \scalebox{0.95}{
            \begin{tikzpicture}[>=stealth, node distance=1.5cm]
                \tikzstyle{square} = [draw, thick, minimum size=1.0mm, inner sep=3pt]
                \begin{scope}
                    \path[->, very thick]
                    node[] (a0) {$A_0$}
                    node[left of=a0] (l0) {$L_0$}
                    node[right of=a0] (l1) {$L_1$}
                    node[right of=l1] (a1) {$A_1$}
                    node[right of=a1] (y) {$Y$}
                    node[above of=l0] (h0) {$H_0$}
                    node[above of=l1] (h1) {$H_1$}
                    node[above of=y] (h2) {$H_2$}
                                   
                    (l0) edge[friendlyblue] (a0)
                    (l0) edge[friendlyblue, bend right] (l1)
                    (l0) edge[friendlyblue, bend right] (a1)
                    (l0) edge[friendlyblue, bend right] (y)
                    (a0) edge[friendlyblue] (l1)
                    (a0) edge[friendlyblue, bend left] (a1)
                    (a0) edge[friendlyblue, bend left] (y)
                    (l1) edge[friendlyblue] (a1)
                    (l1) edge[friendlyblue, bend right] (y)
                    (a1) edge[friendlyblue] (y)
                    (h0) edge[friendlyred] (l0)
                    (h0) edge[friendlyred] (h1)
                    (h1) edge[friendlyred] (l1)
                    (h1) edge[friendlyred] (h2)
                    (h2) edge[friendlyred] (y)
                    ;
                \end{scope}
            \end{tikzpicture}
        }
    \end{center}
    \caption{DAG showing a longitudinal study where two treatments are administered over time and where an MSM may be used to model $\E[Y(a_0, a_1)]$.}
    \label{fig:longitudinal_msm}
\end{figure}
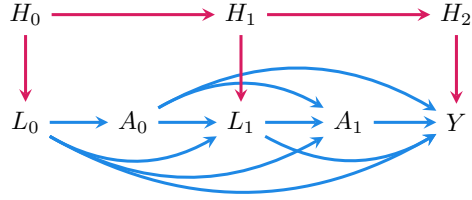

Suppose we use a linear MSM to model the expectation of the potential outcome variable $Y(a_0, a_1)$.
Let $\btheta = (\theta_0, \theta_1, \theta_2)$.
The MSM is defined as 
\begin{align*}
    \E[Y(a_0, a_1)] &= g(a_0, a_1, \btheta) \\
    &= \theta_0 + \theta_1 a_0 + \theta_2 a_1.
\end{align*}
This MSM assumes that there is some baseline level of the outcome $\theta_0$.
Furthermore, the effect of receiving the treatment at baseline is $\theta_1$, and the effect of receiving the treatment at time point one is $\theta_2$.
The effect of receiving both treatments is $\theta_1+\theta_2$.

An important feature of the DAG given in \cref{fig:longitudinal_msm} is that treatment decisions are made using only information that is observed: $L_0$ and $L_1$.
This allows us to estimate $\btheta$ using the estimating equation
\begin{align*}
    \E[\frac{p^*(A_1) p^*(A_2)}{p(A_0 \mid L_0) p(A_1 \mid L_1, A_0, L_0)} h(A_0, A_1) (Y - g(A_0, A_1; \btheta))] &= 0,
\end{align*}
where $h(A_0, A_1)$ is the vector of partial derivatives of $g(a_0, a_1, \btheta)$ with respect to each element of $\btheta$ \citep{ref:laan2003unified,ref:hernan2000marginal}.
This results in three estimating equations that we solve for simultaneously to obtain estimates for $\btheta$.
The heuristic procedure we propose for finding weights in \cref{app:procedure_find_weights} is able to find the weights $\frac{p^*(A_1) p^*(A_2)}{p(A_0 \mid L_0) p(A_1 \mid L_1, A_0, L_0)}$ that move us to a marginally unconfounded distribution in which we can fit an ordinary least squares regression.

Now, $A_0$ and $A_1$ in this toy example can be easily modified to be continuous variables.
For instance, they can represent the dosage of an administered drug rather than simply whether a drug was given to the patient.
Under such a scenario, the MSM and estimating equations remain the same.
In addition, one can easily imagine extending the DAG shown in \cref{fig:longitudinal_msm} to an arbitrary number of time points and treatments.
At each time point $i$, we can add a variable $H_i$ representing unmeasured characteristics of the individual, $L_i$ representing observed covariates used to make the treatment decision, and $A_i$ representing the treatment decision.
The parents of $L_i$ are all previously observed covariates, treatments, and $H_i$.
The parents of $A_i$ are all previously observed covariates and treatments, as well as $L_i$.
The parent of $H_i$ is $H_{i-1}$.
If the last measured time point is $i^*$, then the outcome $Y$ would be $L_{i^*}$. 

In summary, MSMs are commonly used when there are many treatments of interest or when treatments are continuous to ``smooth'' the function $g({\bf a}) = \E[Y({\bf a})]$.

\subsection{Example of Non-Causal Dependence (Confounding Bias)}
\label[secinapp]{app:non_causal_dependence}

Here, we give an example of a potential treatment variable being a significant predictor of an outcome despite it not being a direct cause of the outcome.
Namely, we will show that it is possible for a potential treatment to be significant in the conditional mean $\E[Y \mid {\bf A}]$ but not in the MSM $\E[Y({\bf a})]$.

We first define the function $\text{expit}(x) = \frac{1}{1+e^{-x}}$.
The expit function takes a real number as input and outputs a number between $0$ and $1$, which is useful when we seek to generate binary random variables as a function of other variables with support in the real line.
The notation $\mathcal{N}(\mu, \sigma^2)$ denotes a Normal random variable with mean $\mu$ and variance $\sigma^2$.
The notation $\text{Bernoulli}(p)$ denotes a Bernoulli random variable with probability of success $p$.
Consider the following data-generating process.

\begin{align*}
    C &\sim \text{Bernoulli}(0.5) \\
    A_1 &\sim \text{Bernoulli}(\text{expit}(C)) \\
    A_2 &\sim \text{Bernoulli}(\text{expit}(C)) \\
    Y &\sim A_1 + C + \mathcal{N}(0, 1)
\end{align*}

The DAG representing this data-generating process is shown in \cref{fig:app_a4}.

\begin{figure}[ht]
    \begin{center}
        \centering
        \scalebox{0.95}{
            \begin{tikzpicture}[>=stealth, node distance=1.5cm]
                \tikzstyle{square} = [draw, thick, minimum size=1.0mm, inner sep=3pt]
                \begin{scope}
                    \path[->, very thick]
                    node[] (a1) {$A_1$}
                    node[below of=a1] (a2) {$A_2$}
                    node[right of=a1] (y) {$Y$}
                    node[below of=y] (c) {$C$}
                                   
                    (a1) edge[friendlyblue] (y)
                    (c) edge[friendlyblue] (y)
                    (c) edge[friendlyblue] (a1)
                    (c) edge[friendlyblue] (a2)
                    ;
                \end{scope}
            \end{tikzpicture}
        }
    \end{center}
    \caption{DAG corresponding to the data-generating process above.}
    \label{fig:app_a4}
\end{figure}
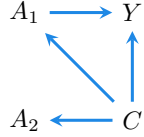

From the DAG, it is clear to see that $A_1$ is a direct cause of $Y$ while $A_2$ is not.
However, fitting the conditional mean model $\E[Y \mid A_1, A_2]$ on the observed data results in

\begin{align*}
    \E[Y \mid A_1, A_2] &= A_1 + \E[C \mid A_1, A_2] \\
    &= A_1 + p(C=1 \mid A_1, A_2) \\
    &= A_1 + \frac{p(C=1 \mid A_1, A_2) p(A_1, A_2)}{p(A_1, A_2)}.
\end{align*}

To simplify the notation as we continue to expand the expression above, let $p(A_1=1 \mid C=1) = p(A_2=1 \mid C=1) = \text{expit}(1) = p_1$ and $p(A_1=1 \mid C=0) = p(A_2=1 \mid C=0) = \text{expit}(0) = p_0$.
Continuing,

\begin{align*}
    &= A_1 + \frac{p(C=1, A_1, A_2)}{\sum_{C} p(C, A_1, A_2)} \\
    &= A_1 + \frac{p(A_1, A_2 \mid C=1) p(C=1)}{p(A_1, A_2 \mid C=0) p(C=0) + p(A_1, A_2 \mid C=1) p(C=1)} \\
    &= A_1 + \frac{p(A_1 \mid C=1) p(A_2 \mid C=1)}{p(A_1 \mid C=1) p(A_2 \mid C=1) + p(A_1 \mid C=0) p(A_2 \mid C=0)} \\
    &= A_1 + \frac{p_1^{A_1} (1-p_1)^{1-A_1} p_1^{A_2} (1-p_1)^{1-A_2}}{p_1^{A_1} (1-p_1)^{1-A_1} p_1^{A_2} (1-p_1)^{1-A_2} + p_0^{A_1} (1-p_0)^{1-A_1} p_0^{A_2} (1-p_0)^{1-A_2}}
\end{align*}

Although this is a complicated function of $A_1$ and $A_2$, it is clear that $A_2$ is a significant predictor of $Y$ in the conditional mean $\E[Y \mid A_1, A_2]$.
This dependence arises from the non-causal path $A_2 \leftarrow C \rightarrow Y$, which is known as a fork in the language of d-separation.

Note, however, that $A_2$ is not a significant predictor of $Y$ in the conditional mean model $\E[Y \mid A_1, A_2, C] = A_1 + C$.
This removal of the dependence on $A_2$ is due to the fact that $Y$ is unconfounded from $A_1, A_2$ conditional on $C$.

Now, let us consider the mean of the potential outcome $\E[Y(a_1, a_2)]$ in this DGP.
We have the exact distribution of $Y$, so we can compute the expectation of this potential outcome directly.

\begin{align*}
    \E[Y(a_1, a_2)] &= \E[a_1 + C + \mathcal{N}(0, 1)] \\
    &= a_1 + 0.5
\end{align*}

Clearly, $A_2$ is not a significant predictor of the MSM here.
This is because $A_2$ does not have a direct, or causal, path to $Y$ in the DAG.
Thus, intervening on $A_2$ does not change $Y$.

In conclusion, it is possible for a potential predictor to be significant in $\E[Y \mid {\bf A}]$ but not in $\E[Y({\bf a})]$.

\subsection{Example of Different Significant Predictors Across Different Conditional Expectations (Non-Collapsibility)}
\label[secinapp]{app:different_significance}

Here, we give an example of a potential treatment variable being significant in the conditional expectation $\E[Y({\bf a}) \mid {\bf C}]$ but not in the marginal expectation $\E[Y({\bf a})]$.
Consider the following data-generating process, and suppose that consistency holds.

\begin{align*}
    C &\sim \mathcal{N}(0, 1) \\
    A_1 &\sim \text{Bernoulli}(0.5) \\
    A_2 &\sim \text{Bernoulli}(0.5) \\
    Y &\sim A_1 + A_2 * C + \mathcal{N}(0, 1)
\end{align*}

Note that, in this DGP, both $A_1$ and $A_2$ are randomized.
We can verify this by drawing a DAG corresponding to the DGP above and seeing that the treatments $A_1$ and $A_2$ are marginally unconfounded from the outcome $Y$ (see \cref{def:confounding}).
Invoking \cref{prop:treatment_outcome_sufficient}, the conditional mean model is equivalent to the MSM in this DGP: $\E[Y \mid A_1, A_2] = \E[Y(a_1, a_2)]$.
Thus, selecting significant predictors in $\E[Y \mid A_1, A_2]$ is the same as selecting significant predictors in the MSM.
The above statements are still true if we were to condition on $C$.

In the conditional expectation $\E[Y(A_1, A_2) \mid C] = \E[Y \mid A_1, A_2, C] = A_1 + A_2 * C$, and $A_2$ is a significant predictor for the outcome.
However, $A_2$ is not a significant predictor for $Y$ in the marginal expectation 
\begin{align*}
    \E[Y(A_1, A_2)] &= \E[Y \mid A_1, A_2] \\
    &= A_1 + A_2 * \E[C \mid A_1, A_2] \\
    &= A_1 + A_2 \E[C] \\
    &= A_1.
\end{align*}
This is an example of a treatment variable having a heterogeneous effect on the outcome: the exact impact of $A_2$ on $Y$ depends on the value of the confounding variable $C$.
When we do not condition on $C$, $A_2$ has no impact on $Y$ because $C$ is independent of $A_1$ and $A_2$ and is mean zero.
The target of inference in this work is to infer the treatments that are significant in the MSM $\E[Y({\bf a})]$, not $\E[Y({\bf a}) \mid {\bf C}^\dagger]$, where ${\bf C}^\dagger \subseteq {\bf C}$.
Thus, we cannot select significant predictors of $Y$ in the model $\E[Y \mid {\bf A}, {\bf C}]$ -- even when treatments and outcome are unconfounded conditional on ${\bf C}$ -- to infer significant predictors in the MSM.

\subsection{Discussion of the Graph Corresponding to the Target Reweighted Distribution Proposed in \cref{sec:reweighted_bic}}
\label[secinapp]{app:graph_of_reweighted_dist}

An alternative way to verify that ${\bf A}'$ and $Y$ are marginally unconfounded in the target reweighted distribution
\begin{align*}
    q({\bf A}'= {\bf a}', Y=y) &= p(Y({\bf A} = {\bf a}')=y) p^*({\bf A}'={\bf a}')
\end{align*}
that we proposed in \cref{sec:reweighted_bic} is to derive the DAG $\G_q$ that it corresponds to and explicitly check that \cref{def:confounding} is satisfied.

First, we make some definitions.
An acyclic directed mixed graph (ADMG) is a graph defined by a set of vertices and a set of edges between vertices where edges may be directed or bidirected.
Given an ADMG $\G$, we can {\it latent project} out a vertex $V$ in the vertex set of $\G$ by first removing $V$ from $\G$.
Then, we add a bidirected edge between any two vertices $A$ and $B$ if there exists a path $A \leftarrow V \rightarrow B$ in $\G$, and we add a directed edge from $A$ to $B$ if there exists a path $A \rightarrow V \rightarrow B$ in $\G$ \citep{ref:richardson2023nested}.

\begin{restatable}[Graph corresponding to the target reweighted distribution $q({\bf A}', Y)$.]{proposition}{graphOfQ}\label{prop:graph_of_q}
    Let $\G$ be the DAG associated with the SCM of the full data distribution $p({\bf V})$.
    The graph that represents the SCM of the target reweighted distribution
    \begin{align*}
        q({\bf A}'= {\bf a}', Y=y) &= p(Y({\bf A} = {\bf a}')=y) p^*({\bf A}'={\bf a}')
    \end{align*}
    is the ADMG $\G_q$,
    constructed from $\G$ as follows
    \begin{itemize}
        \item all vertices in ${\bf A}$ are replaced by their corresponding vertices in ${\bf A}'$,
        \item all edges into variables in ${\bf A}'$ are deleted,
        \item and all variables in ${\bf V} \setminus ({\bf A}' \cup \{Y\})$ are latent projected out, in any order.
    \end{itemize}
\end{restatable}

\begin{proof}
    The desired result follows directly from results on single world intervention graphs (SWIGs) in \citep{ref:richardson2013single}.

    Given the SWIG $\G({\bf a})$ representing the interventional distribution $p({\bf V}({\bf a}))$, the graph representing $p(Y({\bf a}))$ may be obtained by the latent projection operation applied to $\G({\bf a})$.  The graph in $\G_q$ is then obtained by replacing every fixed vertex in ${\bf a}$ in the resulting graph by the corresponding random vertex in ${\bf A}'$.  
\end{proof}

Note that $\G_q$ is also a DAG because it has no bidirected edges, by construction.

\subsection{Discussion on a Heuristic Procedure for Finding Simple Weights}
\label[secinapp]{app:procedure_find_weights}

Here, we present a heuristic procedure for finding weights that are more practical and simple to estimate.
The procedure is as follows.

We start with the DAG $\G$ associated with the SCM and the full data distribution $p({\bf V})$, which may contain latent variables ${\bf U}$.
We initialize the weights as $W \leftarrow 1$ and a new DAG $\G_q$ as $\G_q \leftarrow \G$.
First, we decide whether to simulate ${\bf A}'$; if so, we set $W \leftarrow W \cdot p^*({\bf A}')$ and add the vertices ${\bf A}'$ to $\G_q$.
Recall that the notation $p^*(\cdot)$ refers to a user-defined probability distribution.
If we do not simulate ${\bf A}'$, set ${\bf A}' \equiv {\bf A}$.
Now, complete the following steps any number of times without choosing the same variable more than once.
Choose a variable $V \in {\bf M} \cup \{Y\}$ or $V \in {\bf C} \cup {\bf A}'$ that has no parents in ${\bf U}$ in the DAG $\G_q$.
If we selected a $V \in {\bf M} \cup \{Y\}$, then set $W \leftarrow W \cdot \frac{p(V \mid pa_{\G_q}(V) \setminus {\bf A}, {\bf A}')}{p(V \mid pa_{\G_q}(V))}$, and modify $\G_q$ such that all edges from ${\bf A}$ to $V$ are replaced with edges from ${\bf A}'$ to $V$.
If we selected a $V \in {\bf C} \cup {\bf A}$, then set $W \leftarrow W \cdot \frac{p^*(V)}{p(V \mid pa_{\G_q}(V))}$, and delete all incoming edges to $V$ in $\G_q$.
If ${\bf A}'$ and $Y$ are marginally unconfounded in $\G_q$ at the end of the procedure, then the weights $W$ are valid, and, by construction, $W p({\bf V}) = q({\bf V}, {\bf A}')$, where $q({\bf V}, {\bf A}')$ is a target reweighted distribution where ${\bf A}'$ and $Y$ are marginally unconfounded.
Otherwise, the weights are invalid for the target MSM, and we should not use them.
Note that, by design, these weights are a function of only the observed variables ${\bf O} \cup {\bf A}'$.

Notably, this procedure is able to find the weights described in \cref{exp:example_backdoor,exp:example_frontdoor} for \cref{fig:example_dag_backdoor,fig:example_dag_frontdoor}.
For \cref{fig:example_dag_backdoor}, first do not simulate ${\bf A}'$ then choose $A$.
For \cref{fig:example_dag_frontdoor}, first simulate ${\bf A}'$ then choose $M$.
Interestingly, $W = \frac{p^*({\bf A}') p(Y \mid {\bf A}', {\bf C})}{p(Y \mid {\bf A}, {\bf C})}$ is also a valid weight for the backdoor graph, \cref{fig:example_dag_backdoor}.
Our heuristic procedure above finds these weights by first simulating ${\bf A}'$ then choosing $Y$.

This procedure also yields valid and simple weights in the so-called napkin graph, which is shown in \cref{fig:napkin_graph} below \citep{ref:pearl2018book,ref:guo2025causal}.

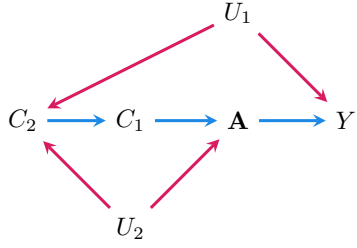
\begin{figure}[ht]
    \begin{center}
        \centering
        \scalebox{0.95}{
            \begin{tikzpicture}[>=stealth, node distance=1.5cm]
                \tikzstyle{square} = [draw, thick, minimum size=1.0mm, inner sep=3pt]
                \begin{scope}
                    \path[->, very thick]
                    node[] (a) {${\bf A}$}
                    node[left of=a] (z) {$C_1$}
                    node[left of=z] (w) {$C_2$}
                    node[right of=a] (y) {$Y$}
                    node[above of=a] (u1) {$U_1$}
                    node[below of=z] (u2) {$U_2$}
                                   
                    (a) edge[friendlyblue] (y)
                    (z) edge[friendlyblue] (a)
                    (w) edge[friendlyblue] (z)
                    (u1) edge[friendlyred] (y)
                    (u1) edge[friendlyred] (w)
                    (u2) edge[friendlyred] (w)
                    (u2) edge[friendlyred] (a)
                    ;
                \end{scope}
            \end{tikzpicture}
        }
    \end{center}
    \caption{The napkin graph.}
    \label{fig:napkin_graph}
\end{figure}

As discussed in \cref{sec:reweighted_bic}, a valid weight to achieve a target reweighted distribution $q({\bf A}', Y)$ for this DAG is, for any value of $C_1=c_1$,
\begin{align*}
    W &= \frac{\frac{\sum_{C_2} p(Y, {\bf A}' \mid C_1=c_1, C_2) p(C_2)}{\sum_{C_2} p({\bf A}' \mid C_1=c_1, C_2) p(C_2)} p^*({\bf A}')}{p(C_2, C_1, {\bf A}, Y)}
\end{align*}
\citep{ref:shpitser2006identification,ref:huang2006identifiability,ref:guo2025causal}.
However, these weights are difficult to estimate in practice.
Instead, consider the weights
\begin{align*}
    W &= \frac{p^*(C_1)}{p(C_1 \mid C_2)}
\end{align*}
The corresponding $\G_q$ that results from $W p(C_2, C_1, {\bf A}, Y) = p(Y \mid {\bf A}, C_1, C_2) p({\bf A} \mid C_1, C_2) p^*(C_1) p(C_2) = q(C_2, C_1, {\bf A}, Y)$ is

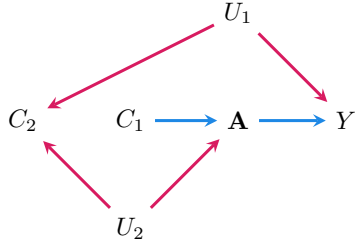
\begin{figure}[ht]
    \begin{center}
        \centering
        \scalebox{0.95}{
            \begin{tikzpicture}[>=stealth, node distance=1.5cm]
                \tikzstyle{square} = [draw, thick, minimum size=1.0mm, inner sep=3pt]
                \begin{scope}
                    \path[->, very thick]
                    node[] (a) {${\bf A}$}
                    node[left of=a] (z) {$C_1$}
                    node[left of=z] (w) {$C_2$}
                    node[right of=a] (y) {$Y$}
                    node[above of=a] (u1) {$U_1$}
                    node[below of=z] (u2) {$U_2$}
                                   
                    (a) edge[friendlyblue] (y)
                    (z) edge[friendlyblue] (a)
                    (u1) edge[friendlyred] (y)
                    (u1) edge[friendlyred] (w)
                    (u2) edge[friendlyred] (w)
                    (u2) edge[friendlyred] (a)
                    ;
                \end{scope}
            \end{tikzpicture}
        }
    \end{center}
    \caption{A DAG $\G_q$ that corresponds to the SCM of $q(C_2, C_1, {\bf A}, Y)$.}
    \label{fig:reweighted_napkin_graph}
\end{figure}

In the DAG shown in \cref{fig:reweighted_napkin_graph}, ${\bf A}$ and $Y$ are marginally unconfounded.
Thus, $\frac{p^*(C_1)}{p(C_1 \mid C_2)}$ are valid weights for the target MSM.
Our proposed heuristic procedure finds these weights by first deciding not to simulate ${\bf A}'$ then choosing $C_1$.
As stated in \cref{sec:reweighted_bic}, a general characterization or algorithm for finding weights with a simple form is beyond the scope of this work.
However, we hope the discussion presented in this subsection can be a first step towards solutions for this problem.

\subsection{Weights for \cref{exp:example_frontdoor}, the Frontdoor Graph}
\label[secinapp]{app:weights_frontdoor}

The DAG factorization for the graph representing the data-generating process of \cref{exp:example_frontdoor} (\cref{fig:example_dag_frontdoor}) is
\begin{align*}
    p({\bf U}) p({\bf A} \mid {\bf U}) p(M \mid {\bf A}) p(Y \mid M, {\bf U}).
\end{align*}
Multiplying the weights $W = \frac{p^*({\bf A}') p(M \mid {\bf A}')}{p(M \mid {\bf A})}$ onto the factorization above gives
\begin{align}
    &\frac{p^*({\bf A}') p(M \mid {\bf A}')}{p(M \mid {\bf A})} \times p({\bf U}) p({\bf A} \mid {\bf U}) p(M \mid {\bf A}) p(Y \mid M, {\bf U}) \nonumber \\
    &= p({\bf U}) p({\bf A} \mid {\bf U}) p^*({\bf A}') p(M \mid {\bf A}') p(Y \mid M, {\bf U}) \label{eq:reweighted_frontdoor}
\end{align}
The DAG corresponding to this reweighted factorization, which we can derive by following the procedure proposed in \cref{app:procedure_find_weights}, is shown in \cref{fig:reweighted_frontdoor}.

\begin{figure}[ht]
    \begin{center}
        \centering
        \scalebox{0.95}{
            \begin{tikzpicture}[>=stealth, node distance=1cm]
                \tikzstyle{square} = [draw, thick, minimum size=1.0mm, inner sep=3pt]
                \begin{scope}
                    \path[->, very thick]
                    node[] (a2) {$A_2$}
                    node[right of=a2] (a2p) {$A_2'$}
                    node[below of=a2] (a1) {$A_3$}
                    node[below of=a2p] (a1p) {$A_3'$}
                    node[above of=a2] (a3) {$A_1$}
                    node[above of=a2p] (a3p) {$A_1'$}
                    node[right of=a2p] (m) {$M$}
                    node[right of=m] (y) {$Y$}
                    node[left of=a2] (c) {${\bf U}$}
                                    
                    (a1p) edge[friendlyblue] (m)
                    (a2p) edge[friendlyblue] (m)
                    (a3p) edge[friendlyblue] (m)
                    (m) edge[friendlyblue] (y)
                    (c) edge[friendlyblue, bend left] (y)
                    (c) edge[friendlyblue] (a1)
                    (c) edge[friendlyblue] (a2)
                    (c) edge[friendlyblue] (a3)
                    ;
                \end{scope}
            \end{tikzpicture}
        }
    \end{center}
    \caption{DAG corresponding to the reweighted factorization of \cref{eq:reweighted_frontdoor}.}
    \label{fig:reweighted_frontdoor}
\end{figure}
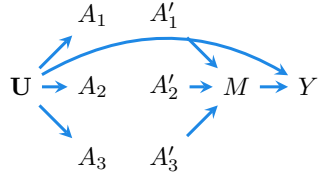

In this DAG, the simulated treatments ${\bf A}'$ are marginally unconfounded from the outcome $Y$.

Intuitively, the weights remove the dependence of the mediator $M$ on the observed treatments ${\bf A}$ and induce dependence of $M$ on the simulated treatments ${\bf A}'$ instead.
This creates directed paths from ${\bf A}'$ to $Y$.
Then, because ${\bf A}'$ is simulated and thus randomized by design, ${\bf A}'$ and $Y$ are marginally unconfounded. 
We can then fit MLEs for the conditional mean model $\E[Y \mid {\bf A}']$ in the reweighted distribution to select ${\bf A}^*$ because $\E[Y({\bf a}')] = \E[Y \mid {\bf A}']$ from \cref{prop:treatment_outcome_sufficient}.

To create the simulated treatments ${\bf A}'$, we suggest the following procedure.
For each observed treatment, choose a parametric distribution that has the same support as it.
Next, choose the parameters of this distribution such that it is fixed and does not depend on any observed variables.
Finally, sample $n$ observations of each observed treatment from its corresponding parametric distribution.

For instance, suppose we observe a binary treatment.
Then, we can create its simulated version by sampling from a Bernoulli distribution with probability of success $0.5$.
Alternatively, suppose we observe a discrete treatment with four possible categories that denote the strength of the administered drug.
Then, we can create its simulated version by sampling from a categorical (multinoulli) distribution where the probability of observing each category is $0.25$.

\subsection{Reweighted MLEs}
\label[secinapp]{app:reweighted_mle}

Here, we give a brief overview of reweighted MLEs.

Recall the notation that ${\bf V}$ represents the set of all variables in the data-generating process.
Let $X$ be the random variable we are fitting an MLE for.
Let the notation $\E_q[X]$ denote the expected value of $X$ with respect to some reweighted distribution.
Everything discussed in this subsection can be easily extended to estimating MLEs for a set of random variables ${\bf X}$.

A typical MLE seeks to find the values of the parameters $\hat{\btheta}$ that maximize the likelihood of $n$ IID observations of $X$: $\underset{\hat{\btheta}}{\text{argmax }} p(X_1, \dots, X_n \mid \hat{\btheta}) = \underset{\hat{\btheta}}{\text{argmax }} \prod_{i=1}^n p(X_i \mid \hat{\btheta})$.
For notational convenience, we will omit the argmax notation in the rest of this subsection.
However, note that each appearance of parameters with the ``hat'' notation implies that we are looking for the values of the parameters that maximize the expression.

Maximizing the likelihood $\prod_{i=1}^n p(X_i \mid \hat{\btheta})$ is equivalent to maximizing the log-likelihood $\sum_{i=1}^n \log p(X_i \mid \hat{\btheta})$, which is in turn equivalent to maximizing the mean log-likelihood $\frac{1}{n} \sum_{i=1}^n \log p(X_i \mid \hat{\btheta})$.
The mean log-likelihood is a finite-sample estimate for the expected log-likelihood of the random variable $X$.
Thus, maximizing the mean log-likelihood is using the observed data to maximize $\E[\log p(X \mid \hat{\btheta})]$.
In conclusion, the MLE attempts to find values $\hat{\btheta}$ that maximize the expected log-likelihood of a single observation of $X$.

The reweighted MLE should, thus, attempt to find values of the parameters $\hat{\btheta}_W$ that maximize the expected log-likelihood of a single observation of $X$ in the reweighted distribution $q({\bf V}, {\bf A}')$.
Let $\E_q[\log p(X \mid \hat{\btheta}_W)]$ be the expected log-likelihood of the random variable $X$ for a single observation drawn from the reweighted distribution.
The reweighted MLE seeks to maximize this quantity.

Note that, if $\btheta_W^*$ and $\btheta^*$ are the true values of the parameters in the reweighted and observed distributions, respectively, $\E_q[\log p(X \mid \btheta_W^*)] \neq \E[\log p(X \mid \btheta^*)]$ in general.
Therefore, the reweighted MLE estimates and the MLE estimates will typically not be the same.

As discussed in \cref{sec:reweighted_bic}, let $W({\bf O}, {\bf A}')$ be weights such that
\begin{align*}
    W({\bf O}, {\bf A}') p({\bf V}) &= q({\bf V}, {\bf A}'),
\end{align*}
where $q({\bf V}, {\bf A}')$ is a target reweighted distribution.
We may replace $q({\bf V}, {\bf A}')$ with any reweighted distribution, if we wish.
Then,
\begin{align*}
    \E[W({\bf O}, {\bf A}') \log p(X \mid \hat{\btheta}_W)] &= \int_{{\bf V}, {\bf A}'} W({\bf O}, {\bf A}') \log p(X \mid \hat{\btheta}_W) p({\bf V}) d{\bf v, a'} \\
    &= \int_{{\bf V}, {\bf A}'} \log p(X \mid \hat{\btheta}_W) W({\bf O}, {\bf A}') p({\bf V}) d{\bf v, a'} \\
    &= \int_{{\bf V}, {\bf A}'} \log p(X \mid \hat{\btheta}_W) q({\bf V}, {\bf A}') d{\bf v, a'} \\
    &= \E_q[\log p(X \mid \hat{\btheta}_W)]
\end{align*}
From the above equality, we can see that the reweighted MLE, which maximizes $\E_q[\log p(X \mid \hat{\btheta}_W)]$, equivalently seeks to find parameters $\hat{\btheta}_W$ that maximize $\E[W({\bf O}, {\bf A}') \log p(X \mid \hat{\btheta}_W)]$.
The finite sample estimate of this quantity is 
\begin{align*}
    \frac{1}{\sum_{i=1}^n W_i({\bf O}, {\bf A}')} \sum_{i=1}^n W_i({\bf O}, {\bf A}') \log p(X_i \mid \hat{\btheta}_W).
\end{align*}
Because the sum of the weights is fixed, maximizing this finite-sample estimate of the reweighted mean is equivalent to maximizing 
\begin{align*}
    \sum_{i=1}^n W_i({\bf O}, {\bf A}') \log p(X_i \mid \hat{\btheta}_W).
\end{align*}
In summary, we can compute the reweighted MLE by maximizing the reweighted sum of the observed log-likelihoods. 

When we have access to the ground-truth values of the weights, the reweighted MLE inherits key properties of the typical MLE, such as convergence to the true parameters at a rate of $O_p(n^{-1/2})$.
This is because the resulting estimating equation is simply an M-estimator with the criterion function $\sum_{i=1}^n W_i({\bf O}, {\bf A}') \log p(X_i \mid \hat{\btheta}_W)$ \citep{ref:van2000asymptotic}.

\subsection{Time and Space Complexity of \cref{alg:search}}
\label[secinapp]{app:complexity_search_alg}

Let $k$ be the number of variables in the set of potential treatments ${\bf A}$ and $n$ be the sample size.
The time complexity of \cref{alg:search} is $O(nk)$ because it iterates over the variables in ${\bf A}$ twice, and, at each iteration, we need to compute two MLEs and two BIC scores.
The space complexity of \cref{alg:search} is $O(n+k)$ because we need to keep track of the list of significant predictor variables ${\bf A}^*$, which has size at most $k$, as well as the data and weights, both of which are of size $n$.

\section{Proofs of Theorems and Propositions}
\label[secinapp]{app:proofs}

\subsection{Proof Intuition for \cref{prop:bic_derivation}}
\label[secinapp]{app:bic_derivation_intuition}

Here, we give intuitions for key steps of the proof for \cref{prop:bic_derivation}.
We use the same notation as that in \cref{sec:bic_and_consistency}.
We start with \cref{eq:marginal_likelihood} and rewrite it as
\begin{align*}
    p(\bD \mid M_r) &= \int \exp[\log p(\bD \mid \btheta_r, M_r)] \pi(\btheta_r \mid M_r) d\btheta_r.
\end{align*}
Next, we rewrite $\log p(\bD \mid \btheta_r, M_r)$ and $\pi(\btheta_r \mid M_r)$ using their Taylor expansions about the MLE with respect to the model $M_r$, which we denote as $\hat{\btheta}_r$.
We do so for two reasons.
First, both expansions will contain terms multiplied by increasingly higher degrees of $(\btheta_r - \hat{\btheta}_r) = O_p(n^{-1/2})$.
Thus, higher degree terms in the Taylor expansion will become negligible.
Second, most of the first and second derivatives of $\log p(\bD \mid \btheta_r, M_r)$ and $\pi(\btheta_r \mid M_r)$ evaluated at the MLE can be ignored for the finite sample approximation because they are constants.

However, there are two important non-negligible terms in the Taylor expansions: $\log p(\bD \mid \hat{\btheta}_r, M_r)$ and 
\begin{align*}
    & \frac{1}{2} (\btheta_r - \hat{\btheta}_r)^T  \frac{\partial^2 \log p(\bD \mid \btheta_r, M_r)}{\partial \btheta_r\partial \btheta^T} \bigg \vert_{\btheta_r = \hat{\btheta}_r} (\btheta_r - \hat{\btheta}_r). 
\end{align*}
They are non-negligible because the values of these terms grow with the sample size.
The first term $\log p(\bD \mid \hat{\btheta}_r, M_r)$ does not depend on the variable of integration $\btheta_r$ in \cref{eq:marginal_likelihood}, so we can factor it out of the integral.
Next, the second order term above approximately integrates to $n^{-d_r/2}$. Roughly, the $n$ comes from the term containing $n$ copies of the Fisher information of the MLE, and $d_r$ comes from the dimension of the matrix of second derivatives.

Putting this together and taking the logarithm on both sides, the marginal log-likelihood of the data can be approximated as
\begin{align*}
    \log p(\bD \mid M_r) &\approx \log \bigg[ \exp[\log p(\bD \mid \hat{\btheta}_r, M_r)] n^{-d_r/2} \bigg] \\
    &= \log p(\bD \mid \hat{\btheta}_r, M_r) - \frac{d_r}{2} \log n,
\end{align*}
which is the desired result and exactly the form that the BIC takes.
We ignore the terms resulting from the Taylor expansion of the prior $\pi(\btheta_r \mid M_r)$ because they become constants after integration.

\subsection{Proof of \cref{prop:bic_derivation}}

\bicDerivation*

\begin{proof}
    In this proof, for notational simplicity, we will drop the conditioning on the candidate model $M_r$, but note that we are conditioning on $M_r$ throughout the proof.
	
    Recall \cref{eq:marginal_likelihood}, which states that
    \begin{align*}
        p(\bD) &= \int p(\bD \mid \btheta_r) \pi(\btheta_r) d\btheta_r,
    \end{align*}
    where $\pi(\btheta_r)$ is the prior for $\btheta_r$.

    Let $\ell(\btheta_r) = \log p(\bD \mid \btheta_r)$ be the natural logarithm of the observed data likelihood conditional on the parameters $\btheta_r$.
    Then, we can rewrite \cref{eq:marginal_likelihood} as
	\begin{align*}
		p(\bD) &= \int \exp[\ell(\btheta_r)] \pi(\btheta_r) d\btheta_r
	\end{align*}

    Supposing that $\ell(\btheta_r)$ is thrice differentiable and bounded (Condition 5 in \cref{app:regularity_conditions}), consider a Taylor expansion of the log-likelihood function at the point $\hat{\btheta}_r$, the MLE for $\btheta_r$:
	\begin{align}
		\ell(\btheta_r) &= \ell(\hat{\btheta}_r) + (\btheta_r - \hat{\btheta}_r) \frac{\partial \ell(\btheta_r)}{\partial \btheta_r} \bigg \vert_{\btheta_r = \hat{\btheta}_r} + \frac{1}{2}  (\btheta_r - \hat{\btheta}_r)^T  \frac{\partial^2 \ell(\btheta)}{\partial \btheta_r\partial \btheta^T} \bigg \vert_{\btheta_r = \hat{\btheta}_r} (\btheta_r - \hat{\btheta}_r) + O_p((\btheta_r - \hat{\btheta}_r)^3) 
        \label{eq:likelihood_taylor1} \\
		&= \ell(\hat{\btheta}_r) - \frac{n}{2} (\btheta_r- \hat{\btheta}_r)^T  J(\hat{\btheta}_r) (\btheta_r- \hat{\btheta}_r) + O_p(n^{-3/2})
		\label{eq:likelihood_taylor2}
	\end{align}
    where 
    \begin{align*}
        J(\hat{\btheta}_r) = -\frac{1}{n} \frac{\partial^2 \ell(\btheta)}{\partial \btheta_r\partial \btheta^T} \bigg\vert_{\btheta_r= \hat{\btheta}_r}.
    \end{align*}
    In \cref{eq:likelihood_taylor1}, the final $O_p((\btheta_r - \hat{\btheta}_r)^3)$ term follows from the fact that all subsequent terms of the Taylor expansion are terms with $(\btheta_r - \hat{\btheta}_r) = O_p(n^{-1/2})$ at an order greater than $3$, so the cubic term dominates.
	To go from \cref{eq:likelihood_taylor1} to \cref{eq:likelihood_taylor2}, first notice that the first derivative of the log-likelihood of the data evaluated at the MLE is $0$ by definition, i.e. $\frac{\partial \ell(\btheta_r)}{\partial \btheta_r} \big\vert_{\btheta_r= \hat{\btheta}_r} = 0$.
    Thus, the second term in \cref{eq:likelihood_taylor1} is $0$.
    Next, $O_p((\btheta_r - \hat{\btheta}_r)^3) = O_p((n^{-1/2})^3) = O_p(n^{-3/2})$.
    Finally, $J(\hat{\btheta}_r)$ is the average observed Fisher information matrix per observation.

	Next, we similarly expand the prior distribution about the MLE $\hat{\btheta}_r$ up to the second term as follows:
	\begin{align}
		\pi(\btheta_r) &= \pi(\hat{\btheta}_r) + (\btheta_r- \hat{\btheta}_r) \frac{\partial \pi(\btheta_r)}{\partial \btheta_r} \bigg \vert_{\btheta_r= \hat{\btheta}_r} + O_p(n^{-1}).
		\label{eq:prior_taylor}
	\end{align}
	Now, substituting \cref{eq:likelihood_taylor2} and \cref{eq:prior_taylor} into \cref{eq:marginal_likelihood}, we get
	\begin{align}
		p(\bD) &= \int \exp \bigg [ \ell(\hat{\btheta}_r) - \frac{n}{2}  (\btheta_r - \hat{\btheta}_r)^T  J(\hat{\btheta}_r) (\btheta_r- \hat{\btheta}_r) + O_p(n^{-3/2}) \bigg ] \nonumber \\
        &\quad \quad \times \bigg ( \pi(\hat{\btheta}_r) + (\btheta_r- \hat{\btheta}_r) \frac{\partial \pi(\btheta_r)}{\partial \btheta_r} \bigg \vert_{\btheta_r= \hat{\btheta}_r} + O_p(n^{-1}) \bigg ) d\btheta_r\label{eq:combined_taylor}
	\end{align}
    Next, note that $\exp[O_p(n^{-3/2})] = 1 + O_p(n^{-3/2})$ from a Taylor expansion of $\exp[x]$ about $x=0$ as follows:
    \begin{align*}
        &\exp[x] = 1 + x + \frac{x^2}{2!} + \dots \\
        &\implies \exp[O_p(n^{-3/2})] = 1 + O_p(n^{-3/2})
    \end{align*}
    Furthermore, we can move $\exp[\ell(\hat{\btheta}_r)]$ out of the integral because this term does not depend on $\btheta_r$, the variable of integration.
    Using these two observations to continue simplifying \cref{eq:combined_taylor} gives us
	\begin{align}
		&= \exp[\ell(\hat{\btheta}_r)] \int \exp \bigg [ - \frac{n}{2}  (\btheta_r- \hat{\btheta}_r)^T  J(\hat{\btheta}_r) (\btheta_r- \hat{\btheta}_r) \bigg ] \nonumber \\
		&\quad \times \bigg ( \pi(\hat{\btheta}_r) + (\btheta_r- \hat{\btheta}_r) \frac{\partial \pi(\btheta_r)}{\partial \btheta_r} \bigg \vert_{\btheta_r= \hat{\btheta}_r} + O_p(n^{-1}) \bigg ) \bigg (1 + O_p(n^{-3/2}) \bigg )d\btheta_r\label{eq:move_log_likelihood_out} \\
		&= \exp[\ell(\hat{\btheta}_r)] \int \exp \bigg [ - \frac{n}{2}  (\btheta_r- \hat{\btheta}_r)^T  J(\hat{\btheta}_r) (\btheta_r- \hat{\btheta}_r) \bigg ] \nonumber \\
        &\quad \times \bigg ( \pi(\hat{\btheta}_r) + (\btheta_r- \hat{\btheta}_r) \frac{\partial \pi(\btheta_r)}{\partial \btheta_r} \bigg \vert_{\btheta_r= \hat{\btheta}_r} + O_p(n^{-1}) \bigg )d\btheta_r\label{eq:multiply_through_error}
	\end{align}
    \cref{eq:multiply_through_error} follows from expanding the terms inside the integral.
    All of the terms in this expansion with a coefficient of $O_p(n^{-3/2})$ are dominated by $O_p(n^{-1})$ and are thus left out of the resulting equation.
    
    We now continue to simplify and apply lineartiy of integration.
	\begin{align}
		&= \exp[\ell(\hat{\btheta}_r)] \bigg (\pi(\hat{\btheta}_r) \int \exp \bigg [ - \frac{n}{2}  (\btheta_r- \hat{\btheta}_r)^T  J(\hat{\btheta}_r) (\btheta_r- \hat{\btheta}_r) \bigg ] d\btheta_r\nonumber \\
		&\quad + \int \exp \bigg [ - \frac{n}{2}  (\btheta_r- \hat{\btheta}_r)^T  J(\hat{\btheta}_r) (\btheta_r- \hat{\btheta}_r) \bigg ] (\btheta_r- \hat{\btheta}_r) \frac{\partial \pi(\btheta_r)}{\partial \btheta_r} \bigg \vert_{\btheta_r= \hat{\btheta}_r} d\btheta_r\nonumber \\
        &\quad + \int \exp \bigg [ - \frac{n}{2}  (\btheta_r- \hat{\btheta}_r)^T  J(\hat{\btheta}_r) (\btheta_r- \hat{\btheta}_r) \bigg ] O_p(n^{-1}) d\btheta_r\bigg ) \label{eq:expand_integral}
	\end{align}

    We now consider each term in \cref{eq:expand_integral} separately.
	For the first term, we have that 
    \begin{align*}
        \int \exp \bigg [ - \frac{n}{2}  (\btheta_r- \hat{\btheta}_r)^T  J(\hat{\btheta}_r) (\btheta_r- \hat{\btheta}_r) \bigg ] d\btheta_r &= (2\pi)^{d_r/2} |(nJ(\hat{\btheta}_r))^{-1}|^{1/2} \\
        &= (2\pi)^{d_r/2} n^{-d_r/2} |J(\hat{\btheta}_r)|^{-1/2},
    \end{align*}
    where $|J(\hat{\btheta}_r)|$ is the determinant of the $d_r \times d_r$ matrix $J(\hat{\btheta}_r)$.
    The first equality follows from noticing that the left-hand side is the PDF of a multivariate normal distribution, without the normalization term, with mean $\hat{\btheta}_r$ and variance $(n J(\hat{\btheta}_r))^{-1}$. 
    The dimension of this multivariate normal distribution is $d_r$, the number of elements in $\btheta_r$.
    The normalization term of such a distribution is $(2 \pi)^{-d_r/2} |(nJ(\hat{\btheta}_r))^{-1}|^{-1/2}$.
    Thus, this integral is equal to $(2 \pi)^{d_r/2} |(nJ(\hat{\btheta}_r))^{-1}|^{1/2}$.
    For the second equality, recall that the determinant of the inverse of a matrix is the reciprocal of the determinant of the matrix.
    Therefore, we can further simplify to $(2 \pi)^{d_r/2} |(nJ(\hat{\btheta}_r))|^{-1/2}$.
    Finally, because each row in the matrix $J(\hat{\btheta}_r)$ is multiplied by $n$, we use the determinant scaling rule to pull the $n$ out of the determinant as $(2 \pi)^{d/2} (n^{d_r} |J(\hat{\btheta}_r)|)^{-1/2}$.
    One more simplification step gives the expression shown.

    For the second term, notice that $\int (\btheta_r - \hat{\btheta}_r) \exp[ - \frac{n}{2}  (\btheta_r- \hat{\btheta}_r)^T  J(\hat{\btheta}_r) (\btheta_r- \hat{\btheta}_r)] d\btheta_r= 0$ because the term inside the exponent is an even function symmetric about $\hat{\btheta}_r$ -- it is the density of a multivariate normal distribution without the normalizing term -- while the term outside of the exponent is also symmetric about $\hat{\btheta}_r$ but is odd. 
    Thus, the density at each point in the integral to the right of $\hat{\btheta}_r$ is cancelled out by the point equidistant to the left of $\hat{\btheta}_r$. 
    This, combined with the fact that the first derivative of the prior distribution is just a constant, the second term in \cref{eq:expand_integral} is $0$.

    For the third term, we apply the same logic as for the first term, except it is multiplied by $O_p(n^{-1})$.

    Continuing from \cref{eq:expand_integral},
	\begin{align}
        p(\bD) &= \exp[\ell(\hat{\btheta}_r)] \bigg ( (2\pi)^{d_r/2} n^{-d_r/2} |J(\hat{\btheta}_r)|^{-1/2} \pi(\hat{\btheta}_r) + (2\pi)^{d_r/2} n^{-d_r/2} |J(\hat{\btheta}_r)|^{-1/2} O_p(n^{-1}) \bigg ) \\
        &= \exp[\ell(\hat{\btheta}_r)] \bigg ( (2\pi)^{d_r/2} n^{-d_r/2} |J(\hat{\btheta}_r)|^{-1/2} \pi(\hat{\btheta}_r) + O_p(n^{-1-d_r/2}) \bigg ) \label{eq:marginal_simplify1} \\
		&= \exp[\ell(\hat{\btheta}_r)] \bigg ( (2\pi)^{d_r/2} n^{-d_r/2} |J(\hat{\btheta}_r)|^{-1/2} \pi(\hat{\btheta}_r) + O_p(n^{-1}) \bigg ) \label{eq:marginal_simplify2} \\
		&=  \exp[\ell(\hat{\btheta}_r)](2\pi)^{d_r/2} n^{-d_r/2} |J(\hat{\btheta}_r)|^{-1/2} \pi(\hat{\btheta}_r) + O_p(n^{-1}) \label{eq:marginal_simplify3}
	\end{align}
    In \cref{eq:marginal_simplify1}, we notice that $(2\pi)^{d_r/2}$ and $|J(\hat{\btheta}_r)|^{-1/2}$ are constants and can therefore be absorbed by the $O_p(n^{-1 - d/2})$ term.
    In \cref{eq:marginal_simplify2}, we use the constraint that the dimension $d_r > 0$ to conclude that $O_p(n^{-1 - d_r/2}) = O_p(n^{-1})$.
    In \cref{eq:marginal_simplify3}, we use the fact that $\exp[\ell(\hat{\btheta}_r)]$ is just a constant to distribute it between the two terms in the parentheses.

    Now, applying a logarithm to both sides of \cref{eq:marginal_simplify3}, we obtain
	\begin{align}
        \log p(\bD) &= \log \bigg(\exp[\ell(\hat{\btheta}_r)](2\pi)^{d_r/2} n^{-d_r/2} |J(\hat{\btheta}_r)|^{-1/2} \pi(\hat{\btheta}_r) + O_p(n^{-1}) \bigg) \\ 
        &= \log \bigg(\exp[\ell(\hat{\btheta}_r)](2\pi)^{d_r/2} n^{-d_r/2} |J(\hat{\btheta}_r)|^{-1/2} \pi(\hat{\btheta}_r) \bigg) + O_p(n^{-1}) \label{eq:log_simplify1} \\ 
		&= \log p(\bD \mid \hat{\btheta}_r) + \frac{d_r}{2} \log 2\pi - \frac{d_r}{2} \log n - \frac{1}{2} \log |J(\hat{\btheta}_r)| + \log \pi(\hat{\btheta}_r) + O_p(n^{-1})
		\label{eq:bic_with_constants}
	\end{align}
    In \cref{eq:log_simplify1}, we use $\log (A + O_p(n^{-1})) = \log A + \frac{(A + O_p(n^{-1}) - A)}{A} =  \log A + O_p(n^{-1})$ from taking the Taylor expansion of $\log (A+O_p(n^{-1}))$ about $A$, where $A$ is any term or expression greater than $0$.
    Further terms of the Taylor expansion are dominated by the second term with $O_p(n^{-1})$.

    In \cref{eq:bic_with_constants}, note that $\frac{d_r}{2} \log 2\pi$, $-\frac{1}{2} \log |J(\hat{\btheta}_r)|$, and $\log \pi(\hat{\btheta}_r)$ are all constants and thus can be combined to a single term $O_p(1)$, which dominates $O_p(n^{-1})$.
    Thus, we can rewrite \cref{eq:bic_with_constants} as
	\begin{align}
		\log p(\bD) &= \log p(\bD \mid \hat{\btheta}_r) - \frac{d_r}{2} \log n + O_p(1).
        \label{eq:bic_approximation}
	\end{align}
    Because the final term is a constant that becomes negligible as the sample size tends to infinity, we can approximate $\log p(\bD)$ as
    \begin{align*}
		\log p(\bD) &\approx \log p(\bD \mid \hat{\btheta}_r) - \frac{d_r}{2} \log n.
    \end{align*}
    This completes the proof.
\end{proof}

\subsection{Proof Intuition for \cref{thm:bic_consistency}}
\label[secinapp]{app:bic_consistency_intuition}

Here, we give a brief overview for why \cref{thm:bic_consistency} is true.
First, consider the following lemma.

\begin{restatable}[Simplification of the Log-Likelihood]{lemma}{likelihoodSimplification}\label{lm:likelihood_simplification}
    Let ${\bf X}$ be a set of random variables and suppose that ${\bf Z} = \emptyset$.
    Let $M_r \subset \Theta$ be a model, $d_r$ be the dimension of $M_r$, and $\bar{{\bf Y}}_{n} = \frac{1}{n} \sum_{i=1}^n {\bf T}({\bf X}_{i})$.
    Let $H = \sum_{i=1}^n \log h({\bf X}_{i})$, and $S_r = \sup_{\hat{\btheta}_{r} \in M_r} ( \bar{{\bf Y}}_{n} \hat{\btheta}_{r} - b(\hat{\btheta}_{r}) )$, where $\hat{\btheta}_r$ is the MLE for the model $M_r$.
    If the model $M_r$ for ${\bf X}$ is a member of the exponential family, then the BIC score $\gamma(\D, M_r)$ can be equivalently written as
    \begin{align}
        \gamma(\D, M_r) &= H + n S_r - \frac{d_r}{2} \log n.
        \label{eq:simplified_bic}
    \end{align}
\end{restatable}

\cref{lm:likelihood_simplification} follows from plugging in the definition of a density from the exponential family (\cref{eq:exp_family}) to \cref{eq:bic_definition} and simplifying.
First, notice that the term $H$ in \cref{eq:simplified_bic} does not depend on the candidate model nor its dimension.
Thus, when we compare the scores of two models by taking their difference, these terms can be ignored because they will always cancel out exactly.

Next, there are $n$ copies of $S_r$, which we call the maximization term.
This term maximizes the portion of the log-likelihood that depends on the parameters over the space specified by the candidate model.
Finally, there is a penalty term scaled by $\log n$ that reduces the score based on the dimension of the model parameters.

First, consider the circumstance where $M_2$ contains the true parameter but $M_1$ does not.
The maximization term for $M_2$ will, asymptotically, be maximized at the true parameters $\btheta^*$.
Meanwhile, the maximization term for $M_1$ will be maximized at some sub-optimal value.
Thus, the difference between the maximization terms of $M_2$ and $M_1$ will be positive.
This difference is scaled by $n$ whereas the difference between penalty terms is scaled by $\log n$, so the former will asymptotically dominate the overall difference in BIC scores.

Next, consider the circumstance where both models contain the true parameter, but $M_2$ has a smaller model dimension.
Because both models contain the true parameter, their maximization terms will be maximized at, asymptotically, arbitrary close values.
Thus, the difference in maximization terms will be negligible when compared to the difference in penalty terms, which is scaled by $\log n$.
$M_2$ has a lower model dimension, so the overall difference in BIC scores will be positive.

\subsection{Proof of \cref{thm:bic_consistency}}

We first restate the definition of the BIC and the definition of a consistent scoring criterion from \cref{sec:bic_and_consistency}.

\bicDefinition*

\correctChoice*

Now, we will prove several lemmas that will aid us in the proof of \cref{thm:bic_consistency}.

For all the lemmas and proofs in this subsection, let $\bar{{\bf Y}}_n = \frac{1}{n} \sum_{i=1}^n {\bf T}({\bf X}_i)$.
This is a vector of the sample averages of each of the sufficient statistics for the random variables ${\bf X}$.
We also assume that every model $M_r$ is a member of the exponential family and that sufficient regularity conditions hold for MLEs to be well-behaved.
For more details on these regularity conditions, we refer the reader to \cref{app:regularity_conditions} and \cite{ref:van2000asymptotic, ref:cramer1999mathematical, ref:barndorff2014information}.

\likelihoodSimplification*

\begin{proof}
    Following \cref{eq:exp_family} and the assumption that $M_r$ for ${\bf X}$ is a member of the exponential family, the log-likelihood of a set of random variable ${\bf X}$ at some arbitrary parameter $\btheta$ can be simplified as
    \begin{align}
        \log p_{\bf X}(\D \mid \btheta) &= \log p_{\bf X}({\bf X}_1, \dots, {\bf X}_n; \btheta) \\
        &= \log \bigg[ h({\bf X}_1) \exp \bigg[ {\bf T}({\bf X}_1) \btheta - b(\btheta) \bigg] \dots h({\bf X}_n) \exp \bigg[ {\bf T}({\bf X}_n) \btheta - b(\btheta) \bigg] \bigg] \label{eq:likelihood_simplification1} \\
        &= \log \bigg[ h({\bf X}_1) \dots h({\bf X}_n) \exp \bigg[ \bigg( \sum_{i=1}^n {\bf T}({\bf X}_i) \bigg) \btheta - n b(\btheta) \bigg] \bigg] \label{eq:likelihood_simplification2} \\
        &= \log \bigg[ h({\bf X}_1) \dots h({\bf X}_n) \bigg] + \log \bigg[ \exp \bigg[ \bigg( \sum_{i=1}^n {\bf T}({\bf X}_i) \bigg) \btheta - n b(\btheta) \bigg] \bigg] \label{eq:likelihood_simplification3} \\
        &= \sum_{i=1}^n \log h({\bf X}_i) + \bigg( \sum_{i=1}^n {\bf T}({\bf X}_i) \bigg) \btheta- n b(\btheta) \label{eq:likelihood_simplification4} \\
        &= \sum_{i=1}^n \log h({\bf X}_i) + n \bar{{\bf Y}}_n \btheta - n b(\btheta) \label{eq:likelihood_simplification5} \\
        &= \sum_{i=1}^n \log h({\bf X}_i) + n (\bar{{\bf Y}}_n \btheta - b(\btheta)) \label{eq:likelihood_simplification6} 
    \end{align}
    \cref{eq:likelihood_simplification1} follows from plugging in the definition for the PDF or PMF of variables that are members of the exponential family and the IID assumption.
    \cref{eq:likelihood_simplification2} follows from grouping the base measure terms $h({\bf X})$ and the rest of the terms inside the exponent separately.
    Because $b(\btheta)$ is the same for each observation, we end up with $n$ copies of it.
    In contrast, we sum over the $n$ observations of the sufficient statistic ${\bf T}({\bf X}_i)$.
    \cref{eq:likelihood_simplification3} follows from applying the product rule of logarithms: $\log(AB) = \log A + \log B$.
    \cref{eq:likelihood_simplification4} follows from applying the product rule of logarithms on the base measure terms and the identity $\log(\exp(x)) = x$.
    \cref{eq:likelihood_simplification5} follows from the definition $n \bar{{\bf Y}}_n = \sum_{i=1}^n {\bf T}({\bf X}_i)$.
    Finally, \cref{eq:likelihood_simplification6} follows from factoring out $n$ in the last two terms.

    The BIC evaluates the log-likelihood of ${\bf X}$ at the MLE within the constraints of model $M_r$.
    Thus, using the simplification above, we can equivalently write
    \begin{align*}
        \gamma(\D, M_r) &= \log p(\D \mid \hat{\btheta}_r, M_r) - \frac{d_r}{2} \log n \\
        &= \sum_{i=1}^n \log h({\bf X}_i) + n \sup_{\hat{\btheta}_r \in M_r} \bigg( \bar{{\bf Y}}_n \hat{\btheta}_r - b(\hat{\btheta}_r) \bigg) - \frac{d_r}{2} \log n
    \end{align*}
    The supremum here means that we are choosing the values of the parameters within the model $M_r$ that maximize the portion of the log-likelihood that depends on the parameters.
    Note that the base measure does not depend on the parameter, so it is not included within the supremum.
    Our desired result then follows from plugging in the definitions of $H$ and $S_r$.
\end{proof}

\begin{restatable}[Gradient of the Log-Partition Function]{lemma}{gradientLogPartition}\label{lm:gradient_log_partition}
    The gradient of the log-partition function evaluated at the true value of the parameter $\btheta^*$ is equal to the expected value of ${\bf T}({\bf X})$.
    That is,
    \begin{align*}
        \nabla b(\btheta^*) &= \E[{\bf T}({\bf X})],
    \end{align*}
    where $\nabla b(\btheta^*)$ is the gradient of the log-partition function -- the vector of partial derivatives of the log-partition function evaluated at the parameter vector $\btheta^*$.
\end{restatable}

\begin{proof}
    We start with the fact that the joint PDF or PMF of the random variables ${\bf X}$ at their true parameters $\btheta^*$ integrates to $1$.
    \begin{align}
        \int h({\bf x}) \exp[{\bf T}({\bf x}) \btheta^* - b(\btheta^*)] d{\bf x} &= 1 \\
        \frac{\partial}{\partial \btheta^*} \int h({\bf x}) \exp[{\bf T}({\bf x}) \btheta^* - b(\btheta^*)] d{\bf x} &= 0 \label{eq:gradient_log_partition1} \\
        \int \frac{\partial}{\partial \btheta^*} h({\bf x}) \exp[{\bf T}({\bf x}) \btheta^* - b(\btheta^*)] d{\bf x} &= 0 \label{eq:gradient_log_partition2} \\
        \int ({\bf T}({\bf x}) - \nabla b(\btheta^*)) h({\bf x}) \exp[{\bf T}({\bf x}) \btheta^* - b(\btheta^*)] d{\bf x} &= 0 \label{eq:gradient_log_partition3} \\
        \E[{\bf T}({\bf X}) - \nabla b(\btheta^*)] &= 0 \label{eq:gradient_log_partition4} \\
        \E[{\bf T}({\bf X})] &= \nabla b(\btheta^*) \label{eq:gradient_log_partition5} 
    \end{align}
    In \cref{eq:gradient_log_partition1}, we take the derivative with respect to $\btheta^*$ on both sides of the equation.
    In \cref{eq:gradient_log_partition2}, we use our sufficient regularity assumption to switch the order of differentiation and integration.
    In \cref{eq:gradient_log_partition3}, we evaluate the derivative.
    First, the base measure $h({\bf x})$ does not depend on $\btheta^*$, so we do not need to differentiate it.
    Then, we use the chain rule to get 
    \begin{align*}
        \frac{\partial}{\partial \btheta^*} \exp[{\bf T}({\bf x}) \btheta^* - b(\btheta^*)] &= \exp[{\bf T}({\bf x}) \btheta^* - b(\btheta^*)] \frac{\partial}{\partial \btheta^*} ({\bf T}({\bf x}) \btheta^* - b(\btheta^*)) \\
        &= \exp[{\bf T}({\bf x}) \btheta^* - b(\btheta^*)] ({\bf T}({\bf x}) - \nabla b(\btheta^*)).
    \end{align*}
    In \cref{eq:gradient_log_partition4}, we use the law of the unconscious statistician ($\E[g(X)] = \int g(x) p_X(x) dx$) to rewrite the integral as an expectation.
    Finally, in \cref{eq:gradient_log_partition5}, we use the observation that $\nabla b(\btheta^*)$ is just a vector of numbers and does not depend on ${\bf X}$.
    Thus, its expectation is equal to its value.
    This completes the proof.

    For a more detailed treatment of this fact, we refer the reader to Chapter 8 of \cite{ref:barndorff2014information}.    
\end{proof}

\begin{restatable}[Maximum Log-Likelihood of a Single Observation]{lemma}{maxLogLikelihood}\label{lm:max_log_likelihood}
    Let $f(\btheta) = \nabla b(\btheta^*) \btheta - b(\btheta)$.
    $f(\btheta)$ attains its unique maximum at $\btheta^*$.
\end{restatable}

\begin{proof}
    Consider the derivative of $f(\btheta)$ with respect to $\btheta$.
    Using the $\nabla$ notation to denote the gradient,
    \begin{align*}
        \nabla f(\btheta) &= \nabla b(\btheta^*) - \nabla b(\btheta).
    \end{align*}
    This derivative is equal to $0$ when $\nabla b(\btheta^*) = \nabla b(\btheta)$, which occurs when $\btheta = \btheta^*$.

    For a more detailed treatment of this fact, we refer the reader to Theorems 9.13 and 9.1 of \cite{ref:barndorff2014information}.    

    Combining \cref{lm:likelihood_simplification,lm:gradient_log_partition}, we can interpret $f(\btheta) = \nabla b(\btheta^*) \btheta - b(\btheta) = \E[{\bf T}({\bf X})] \btheta - b(\btheta)$ as the log-likelihood, without the base measure, of a guess $\btheta$ for the parameters evaluated at the mean of ${\bf T}({\bf X})$.
\end{proof}

Now that we have proved all of our lemmas, we can prove the consistency of the BIC for marginal exponential families.

\bicConsistency*

\begin{proof}
    We first prove the theorem assuming that ${\bf X}$ has only a single random variable $X_1 \equiv X$.
    Then, the proof can be generalized to the case where ${\bf X}$ has an arbitrary number of elements by replacing each occurrence of $X$ with ${\bf X}$.

    Because we are assuming that ${\bf X}$ has only one element, we will drop the subscript $j$ for $X$.
    We start with the proof for condition (i) where $\btheta^* \in M_2 \setminus M_1$.
    Here, $\btheta^* \not \in M_1$ by assumption.

    As in \cref{lm:max_log_likelihood}, let $f(\btheta) = \nabla b(\btheta^*) \btheta - b(\btheta)$.
    Now, we pick $\epsilon > 0$ and a neighborhood $N$ near $\btheta^*$ such that $N \cap M_1 = \emptyset$ (meaning $N$ does not overlap with the space covered in $\Theta$ by $M_1$).
    For $\btheta \not \in N$,
    \begin{align}
        f(\btheta) + \epsilon &\leq f(\btheta^*) \nonumber \\
        \nabla b(\btheta^*) \btheta - b(\btheta) + \epsilon &\leq \nabla b(\btheta^*) \btheta^* - b(\btheta^*). \label{eq:neighborhood_inequality}
    \end{align}
    Here, $\btheta$ is a point not in $N$ that is sufficiently far away from the true parameters $\btheta^*$.
    From \cref{lm:max_log_likelihood}, we know that $f(\btheta)$ attains its unique maximum at $\btheta^*$, so $f(\btheta)$ will be at least some $\epsilon$ away from $f(\btheta^*)$.

    From the above inequality, we have that
    \begin{align*}
        \sup_{\hat{\btheta}_1 \in M_1} \nabla b(\btheta^*) \hat{\btheta}_1 - b(\hat{\btheta}_1) + \epsilon \leq \nabla b(\btheta^*) \btheta^* - b(\btheta^*)
    \end{align*}
    In words, the above inequality states that the value of the parameter that maximizes $f(\btheta)$ under the model $M_1$ will always be at least some $\epsilon > 0$ away from $f(\btheta^*)$.
    This is true because $M_1$ does not contain the true parameter $\btheta^*$.

    Next, from the weak law of large numbers, $\bar{{\bf Y}}_n \xrightarrow{p} \E[{\bf T}(X)]$ as $n \rightarrow \infty$.
    From \cref{lm:gradient_log_partition}, this means that $\bar{{\bf Y}}_n \xrightarrow{p} \nabla b(\btheta^*)$ as $n \rightarrow \infty$.
    Thus,
    \begin{align*}
        \sup_{\hat{\btheta}_r \in M_r} \bar{{\bf Y}}_n \hat{\btheta}_r - b(\hat{\btheta}_r) &\xrightarrow{p} \sup_{\hat{\btheta}_r \in M_r} \nabla b(\btheta^*) \hat{\btheta}_r - b(\hat{\btheta}_r)
    \end{align*}
    as $n \rightarrow \infty$ for both models $M_1$ and $M_2$ due to the continuity of the function $g(x, M_r) = \sup_{\hat{\btheta}_r \in M_r} x \hat{\btheta}_r - b(\hat{\btheta}_r)$.
    This is an application of the continuous mapping theorem.
    Using the definition of convergence in probability, we can write that, with probability $1$ as $n \rightarrow \infty$,
    \begin{align}
        \bigg | \sup_{\hat{\btheta}_r \in M_r} \bar{{\bf Y}}_n \hat{\btheta}_r - b(\hat{\btheta}_r) - \sup_{\hat{\btheta}_r \in M_r} \nabla b(\btheta^*) \hat{\btheta}_r - b(\hat{\btheta}_r) \bigg | < \epsilon/4.
        \label{eq:supremum_distance}
    \end{align}
    \cref{eq:supremum_distance} states that, for both models $M_1$ and $M_2$, we can always find some $n$ large enough such that the distance between $g(\bar{{\bf Y}}_n, M_r)$ and $g(\nabla b(\btheta^*), M_r)$ is at most $\epsilon/4$ apart.
    The $\epsilon$ here is the same as the $\epsilon$ used in \cref{eq:neighborhood_inequality}.

    Now,
    \begin{align*}
        \sup_{\hat{\btheta}_2 \in M_2} \nabla b(\btheta^*) \hat{\btheta}_2 - b(\hat{\btheta}_2) = \nabla b(\btheta^*) \btheta^* - b(\btheta^*).
    \end{align*}
    By assumption, $\btheta^* \in M_2$, so taking the supremum over $M_2$ for $f(\btheta)$ will achieve its maximum value, which is at $\btheta^*$.
    This allows us to write
    \begin{align}
        \sup_{\hat{\btheta}_1 \in M_1} \nabla b(\btheta^*) \hat{\btheta}_1 - b(\hat{\btheta}_1) + \epsilon \leq \sup_{\hat{\btheta}_2 \in M_2} \nabla b(\btheta^*) \hat{\btheta}_2 - b(\hat{\btheta}_2)
        \label{eq:m1_inequality}
    \end{align}

    Using the definition of absolute value, \cref{eq:supremum_distance} can be alternatively written as
    \begin{align}
        -\epsilon/4 &< \sup_{\hat{\btheta}_r \in M_r} \bar{{\bf Y}}_n \hat{\btheta}_r - b(\hat{\btheta}_r) - \sup_{\hat{\btheta}_r \in M_r} \nabla b(\btheta^*) \hat{\btheta}_r - b(\hat{\btheta}_r) < \epsilon/4 \nonumber \\
        \sup_{\hat{\btheta}_r \in M_r} \bar{{\bf Y}}_n \hat{\btheta}_r - b(\hat{\btheta}_r) - \epsilon/4 &< \sup_{\hat{\btheta}_r \in M_r} \nabla b(\btheta^*) \hat{\btheta}_r - b(\hat{\btheta}_r) < \sup_{\hat{\btheta}_r \in M_r} \bar{{\bf Y}}_n \hat{\btheta}_r - b(\hat{\btheta}_r) + \epsilon/4 
        \label{eq:supremum_distance2}
    \end{align}
    with probability $1$ as $n \rightarrow \infty$ for both models $M_1$ and $M_2$.

    Starting with \cref{eq:m1_inequality} and using the left bound for $M_1$ and the right bound for $M_2$ from \cref{eq:supremum_distance2}, we get that, with probability $1$ as $n \rightarrow \infty$,
    \begin{align}
        \sup_{\hat{\btheta}_1 \in M_1} \bar{{\bf Y}}_n \hat{\btheta}_1 - b(\hat{\btheta}_1) - \epsilon/4 + \epsilon &< \sup_{\hat{\btheta}_2 \in M_2} \bar{{\bf Y}}_n \hat{\btheta}_2 - b(\hat{\btheta}_2) + \epsilon/4 \nonumber \\
        \sup_{\hat{\btheta}_1 \in M_1} \bar{{\bf Y}}_n \hat{\btheta}_1 - b(\hat{\btheta}_1) + \epsilon/2 &< \sup_{\hat{\btheta}_2 \in M_2} \bar{{\bf Y}}_n \hat{\btheta}_2 - b(\hat{\btheta}_2) \nonumber \\
        \epsilon/2 &< \sup_{\hat{\btheta}_2 \in M_2} \bar{{\bf Y}}_n \hat{\btheta}_2 - b(\hat{\btheta}_2) - \sup_{\hat{\btheta}_1 \in M_1} \bar{{\bf Y}}_n \hat{\btheta}_1 - b(\hat{\btheta}_1) \label{eq:in_prob_distance}
    \end{align}
    This inequality states that, with probability $1$ as $n \rightarrow \infty$, the maximization term maximized over the parameter space of $M_1$ will always be at least some $\epsilon/2$ less than the maximization term maximized over the parameter space of $M_2$.
    We will now consider how this fact relates to comparing the BIC scores evaluated at $M_1$ and $M_2$.

    Let us consider $\gamma(\D, M_2) - \gamma(\D, M_1)$ and apply \cref{lm:likelihood_simplification}.
    \begin{align}
        \gamma(\D, M_2) - \gamma(\D, M_1) &= \sum_{i=1}^n \log h(X_i) + n \sup_{\hat{\btheta}_2 \in M_2} \bigg( \bar{{\bf Y}}_n \hat{\btheta}_2 - b(\hat{\btheta}_2) \bigg) - \frac{d_2}{2} \log n \nonumber \\
        &\quad \quad - \bigg( \sum_{i=1}^n \log h(X_i) + n \sup_{\hat{\btheta}_1 \in M_1} \bigg( \bar{{\bf Y}}_n \hat{\btheta}_1 - b(\hat{\btheta}_1) \bigg) - \frac{d_1}{2} \log n \bigg) \nonumber \\
        &= n \bigg( \sup_{\hat{\btheta}_2 \in M_2} \bigg( \bar{{\bf Y}}_n \hat{\btheta}_2 - b(\hat{\btheta}_2) \bigg) - \sup_{\hat{\btheta}_1 \in M_1} \bigg( \bar{{\bf Y}}_n \hat{\btheta}_1 - b(\hat{\btheta}_1) \bigg) \bigg) - \frac{1}{2} (d_2 - d_1) \log n \label{eq:bic_score_diff}
    \end{align}
    We now apply \cref{eq:in_prob_distance} to get that, with probability $1$ as $n \rightarrow \infty$,
    \begin{align*}
        \gamma(\D, M_2) - \gamma(\D, M_1) &> n (\epsilon/2) - \frac{1}{2} (d_2 - d_1) \log n \\
        &> 0
    \end{align*}
    regardless of the difference $(d_2-d_1)$ because the $n(\epsilon/2)$ term will dominate asymptotically over the $\log n$ term.
    This completes the proof for condition (i).

    Now, consider condition (ii), where $\btheta^* \in M_1 \cap M_2$ with $d_2 < d_1$.
    Let
    \begin{align*}
        S_{n,i} &= \sup_{\hat{\btheta}_r \in M_r} \bar{{\bf Y}}_n \hat{\btheta}_r - b(\hat{\btheta}_r)
    \end{align*}
    for both models $M_1$ and $M_2$.
    We want to show that $| S_{n, 1} - S_{n, 2} | = O_p(1/n)$.
    Since $\nabla b(\mathrm{int} \Theta)$ is open, with probability 1 as $n \rightarrow \infty$, there exists a unique MLE $\hat{\btheta}_n$ such that
    \begin{align*}
        \sup_{\btheta \in \mathrm{int}\Theta} \bar{{\bf Y}}_n \btheta - b(\btheta) &= \bar{{\bf Y}}_n \hat{\btheta}_n - b(\hat{\btheta}_n).
    \end{align*}
    Let $U_n = \bar{{\bf Y}}_n \btheta^* - b(\btheta^*)$ and $\hat{U}_n = \bar{{\bf Y}}_n \hat{\btheta}_n - b(\hat{\btheta}_n)$.
    Then, for both models $M_1$ and $M_2$,
    \begin{align}
        0 \leq S_{n, i} - U_n \leq \hat{U}_n - U_n. \label{eq:likelihood_compare2}
    \end{align}
    \cref{eq:likelihood_compare2} follows from two observations.
    First, because $\btheta^* \in M_1$ and $\btheta^* \in M_2$, $S_{n,i}$ will attain its maximum at some point greater than or equal to $U_n$ based on the observed mean $\bar{{\bf Y}}_n$.
    Next, $\hat{U}_n$ is evaluated at the unique MLE, which will never be less than whatever value $S_{n,i}$ attains its maximum.

    Continuing to expand \cref{eq:likelihood_compare2},
    \begin{align}
        \hat{U}_n - U_n &= \bar{{\bf Y}}_n (\hat{\btheta}_n - \btheta^*) + b(\btheta^*) - b(\hat{\btheta}_n) \label{eq:difference_u}
    \end{align}
    Next, we can write a Taylor expansion for $b(\hat{\btheta}_n)$ about $\btheta^*$ as
    \begin{align}
        b(\hat{\btheta}_n) &= b(\btheta^*) + (\hat{\btheta}_n - \btheta^*) \nabla b(\btheta^*) + O_p(1/n) \nonumber \\
        b(\hat{\btheta}_n) - b(\btheta^*) &= (\hat{\btheta}_n - \btheta^*) \nabla b(\btheta^*) + O_p(1/n) \label{eq:b_taylor_expansion}
    \end{align}
    where the final term arises from MLE estimates converging at a rate of $O_p(n^{-1/2})$ under sufficient regularity conditions \citep{ref:van2000asymptotic}.

    Plugging in \cref{eq:b_taylor_expansion} to \cref{eq:difference_u},
    \begin{align*}
        \hat{U}_n - U_n &= \bar{{\bf Y}}_n (\hat{\btheta}_n - \btheta^*) - (\hat{\btheta}_n - \btheta^*) \nabla b(\btheta^*) + O_p(1/n) \\
        &= (\hat{\btheta}_n - \btheta^*) (\bar{{\bf Y}}_n - \nabla b(\btheta^*)) + O_p(1/n) \\
        &= O_p(1/\sqrt{n}) O_p(1/\sqrt{n}) + O_p(1/n) \\
        &= O_p(1/n).
    \end{align*}
    The key insight above is that $\bar{{\bf Y}}_n - \nabla b(\btheta^*) = O_p(1/\sqrt{n})$ due to the central limit theorem \citep{ref:van2000asymptotic} and \cref{lm:gradient_log_partition}.
    
    We can use the above observation to also conclude that $S_{n, i} - U_n = O_p(1/n)$ for both models $M_1$ and $M_2$ using \cref{eq:difference_u}.
    This gives us 
    \begin{align}
        S_{n,2} - S_{n,1} &= S_{n, 2} - U_n - S_{n, 1} + U_n \nonumber \\
        &= S_{n, 2} - U_n - (S_{n, 1} - U_n) \nonumber \\
        &= O_p(1/n) - O_p(1/n) \nonumber \\
        &= O_p(1/n). \label{eq:dimishing_difference_in_mle}
    \end{align}
    Now, with probability $1$ as $n \rightarrow \infty$ and using \cref{eq:bic_score_diff} directly,
    \begin{align*}
        \gamma(\D, M_2) - \gamma(\D, M_1) &= n(S_{n, 2} - S_{n, 1}) - \frac{1}{2} (d_2 - d_1) \log n \\
        &= n O_p(1/n) - \frac{1}{2} \log n (d_2 - d_1) \\
        &= O_p(1) - \frac{1}{2} \log n (d_2 - d_1) \\
        &> 0
    \end{align*}
    because $d_2 < d_1$ by assumption, and the $\log n$ term will dominate the constant term as $n \rightarrow \infty$.
    This completes the proof for condition (ii).
    We have thus proven our desired result under the assumption that ${\bf X}$ has only a single element.

    Now, we can extend the proof above to the general case where ${\bf X}$ has an arbitrary number of elements by replacing each occurrence of $X$ with ${\bf X}$.
    We have thus proven our desired result.
\end{proof}

\subsection{Proof of \cref{cor:bic_consistency_conditional}}

\bicConditionalConsistency*

\begin{proof}
    We are only given that, for each model, the distribution of ${\bf X}$ given ${\bf Z}$ is a member of the exponential family; we do not know anything regarding the distribution of ${\bf Z}$.
    However, suppose that we estimate the marginal PDF or PMF of ${\bf Z}$ using its empirical distribution \citep{ref:van2000asymptotic}:
    \begin{align*}
        p_{\bf Z}({\bf Z} = {\bf z}) = \frac{1}{n} \sum_{i=1}^n \mathbb{I}({\bf Z}_i = {\bf z})
    \end{align*}
    for all observed values of ${\bf Z}$.
    If ${\bf Z}$ has $m$ unique values in the observed data, then we need $m-1$ parameters to learn the distribution of ${\bf Z}$.
    Let $\hat{\boldeta}$ be the vector of these $m-1$ parameters.
    Note that this procedure extends readily to cases where ${\bf Z}$ has support in the real line.
    Furthermore, under this estimation strategy, $p_{\bf Z}({\bf Z} = {\bf z}) > 0$ for all possible observed values of ${\bf Z}$.

    Next, let us rewrite the BIC score for model $M_r$ when we seek to compare models for ${\bf X}$ given ${\bf Z}$.
    \begin{align}
        \gamma(\D, M_r) &= \sum_{i=1}^n \log p_{{\bf X} \mid {\bf Z}}({\bf X}_i \mid {\bf Z}_i, \hat{\btheta}({\bf Z}_i), M_r) - \frac{d_r}{2} \log n \label{eq:conditional_bic_simp1} \\
        &= \sum_{i=1}^n \log \frac{p_{{\bf X}, {\bf Z}}({\bf X}_i, {\bf Z}_i \mid \hat{\btheta}, \hat{\boldeta}, M_r)}{p_{\bf Z}({\bf Z}_i \mid \hat{\boldeta})} - \frac{d_r}{2} \log n \label{eq:conditional_bic_simp2} \\
        &= \sum_{i=1}^n \log p_{{\bf X}, {\bf Z}}({\bf X}_i, {\bf Z}_i \mid \hat{\btheta}, \hat{\boldeta}, M_r) - \sum_{i=1}^n \log p_{\bf Z}({\bf Z}_i \mid \hat{\boldeta}) - \frac{d_r}{2} \log n \label{eq:conditional_bic_simp3} 
    \end{align}
    \cref{eq:conditional_bic_simp1} follows from the definition of the BIC score for conditional distributions and the IID assumption.
    \cref{eq:conditional_bic_simp2} follows from Bayes' rule.
    The joint distribution in the numerator is completely determined by the parameters $\btheta$ and $\boldeta$.
    \cref{eq:conditional_bic_simp3} follows from simplifying the logarithm.
    Importantly, note that the marginal distribution of ${\bf Z}$ does not depend on the model $M_r$.

    We now seek to show that the joint distribution $p_{{\bf X}, {\bf Z}}({\bf x}, {\bf z} \mid \btheta, \boldeta)$ is a member of the exponential family when ${\bf X}$ given ${\bf Z}$ is a member of the exponential family and the marginal distribution of ${\bf Z}$ is estimated using its empirical distribution as described above.
    Let ${\bf z}_1, \dots {\bf z}_m$ be the $m$ possible values that the variables ${\bf Z}$ can take based on the observed distribution, and let $\eta_j = p_{\bf Z}({\bf Z} = {\bf z}_j), \forall j=1, \dots, m$.
    Note that the parameters for the conditional distribution of ${\bf X}$ given ${\bf Z}$ $\btheta({\bf z})$ can take on at most $m$ unique values.
    \begin{align*}
        p_{{\bf X}, {\bf Z}}({\bf x}, {\bf z} \mid \btheta, \boldeta) &= p_{{\bf X} \mid {\bf Z}}({\bf x} \mid {\bf z}, \btheta({\bf z})) p_{\bf Z}({\bf z} \mid \boldeta) \\
        &= h({\bf x}) \exp[{\bf T}({\bf x}) \btheta({\bf z}) - b(\btheta({\bf z}))] \eta_{\bf z} \\
        &= h({\bf x}) \exp[{\bf T}({\bf x}) \btheta({\bf z}) + \log \eta_{\bf z} - b(\btheta({\bf z}))] \\
        &= h({\bf x}) \exp[{\bf T}({\bf x}, {\bf z}) \phi],
    \end{align*}
    where
    \begin{align*}
        {\bf T}({\bf x}, {\bf z}) &= 
        \begin{pmatrix}
            \mathbb{I}({\bf z} = {\bf z}_1) {\bf T}({\bf x}) \\
            \dots \\
            \mathbb{I}({\bf z} = {\bf z}_m) {\bf T}({\bf x}) \\
            \mathbb{I}({\bf z} = {\bf z}_1) \\
            \dots \\
            \mathbb{I}({\bf z} = {\bf z}_m) 
        \end{pmatrix}   
    \end{align*}
    and 
    \begin{align*}
        \phi &= 
        \begin{pmatrix}
            \btheta({\bf z}_1) \\
            \dots \\
            \btheta({\bf z}_m) \\
            \log \eta_1 - b(\btheta({\bf z}_1))
            \\
            \dots \\
            \log \eta_m - b(\btheta({\bf z}_m)) 
        \end{pmatrix}.   
    \end{align*}
    Thus, $p_{{\bf X}, {\bf Z}}({\bf x}, {\bf z} \mid \btheta, \boldeta)$ is a member of the exponential family.
    Using this fact, we can now show that the BIC score is consistent for conditional distributions:
    \begin{align*}
        \gamma(\D, M_2) - \gamma(\D, M_1) &= \sum_{i=1}^n \log p_{{\bf X}, {\bf Z}}({\bf X}_i, {\bf Z}_i \mid \hat{\btheta}, \hat{\boldeta}, M_2) - \sum_{i=1}^n \log p_{\bf Z}({\bf Z}_i \mid \hat{\boldeta}) - \frac{d_2}{2} \log n \\
        &\quad \quad - \sum_{i=1}^n \log p_{{\bf X}, {\bf Z}}({\bf X}_i, {\bf Z}_i \mid \hat{\btheta}, \hat{\boldeta}, M_1) + \sum_{i=1}^n \log p_{\bf Z}({\bf Z}_i \mid \hat{\boldeta}) + \frac{d_1}{2} \log n \\
        &= \sum_{i=1}^n \log p_{{\bf X}, {\bf Z}}({\bf X}_i, {\bf Z}_i \mid \hat{\btheta}, \hat{\boldeta}, M_2) - \frac{d_2}{2} \log n \\
        &\quad \quad - \sum_{i=1}^n \log p_{{\bf X}, {\bf Z}}({\bf X}_i, {\bf Z}_i \mid \hat{\btheta}, \hat{\boldeta}, M_1) + \frac{d_1}{2} \log n 
    \end{align*}
    The term $\sum_{i=1}^n \log p_{\bf Z}({\bf Z}_i \mid \hat{\boldeta})$ cancels out exactly because, under any model with a fixed dataset, it will evaluate to the same value.
    Note that $d_2$ and $d_1$ are the model dimensions for the MLEs of ${\bf X}$ given ${\bf Z}$ under $M_2$ and $M_1$; the dimension of $\boldeta$ is the same under both models, so we can ignore it.
    According to the derivation above, comparing models for conditional distributions with the BIC is equivalent to comparing models for joint distributions.
    Thus, according to \cref{thm:bic_consistency}, the desired result is true.
\end{proof}

\begin{restatable}[Extensions to the Consistency Proof]{remark}{extensionsConsistency}\label{rm:extensions_consistency}
    Our proof for the consistency of the BIC extends the proof of \cite{ref:haughton1988choice} in the sense that we prove that the BIC is consistent for a more general class of variables from the exponential family and also for joint and conditional models.
    In \cite{ref:haughton1988choice}, the author considers variables $X$ with the density
    \begin{align*}
        p_X(x, \btheta) &= \exp[{\bf x} \btheta - b(\btheta)],
    \end{align*}
    where ${\bf x}$ is a $d$-dimensional vector with $d$ copies of the observed value $X=x$.
    Compared with our definition given in \cref{eq:exp_family}, the PDF above enforces (i) the base measure $h(x)=1$ and (ii) the sufficient statistic ${\bf T}(x) = {\bf x}$.
    Furthermore, the PDF above is for only a single variable $X$ and does not allow for $X$ to depend on other variables in the system.

    Sometimes, additive terms of the log-likelihood that do not depend on the parameters, namely the base measure $h(x)$, are omitted from the definition of the density of exponential families because they do not affect key operations on the log-likelihood such as maximizing it or differentiating it with respect to $\btheta$.
    It is unclear whether the authors of \cite{ref:haughton1988choice} omitted the base measure from their definition of the exponential family for this reason.
    In any case, we believe that it is beneficial to explicitly verify that the BIC is consistent without the constraint that $h(x) = 1$.

    Although the key steps of the proof for consistency remain largely the same between our proof and that of \cite{ref:haughton1988choice}, our slight extension covers key distributions not previously included in their original proof.
    For instance, the Poisson distribution -- $p_X(x, \theta) = \frac{1}{x!} \exp[x \log \theta - \theta]$ -- requires the inclusion of the base measure $h(x) = \frac{1}{x!}$.
    In addition, the chi-squared distribution, 
    \begin{align*}
        p_X(x, \theta) &= e^{-x/2} \exp[\log (x) \theta - (\log \Gamma(\theta+1) + (\theta+1) \log 2)],
    \end{align*}
    where $\Gamma(x)$ is the Gamma function, requires the base measure $h(x) = e^{-x/2}$ and the sufficient statistic $T(x) = \log x$.
\end{restatable}

\subsection{Proof of \cref{prop:treatment_outcome_sufficient}}

Recall \cref{def:confounding}, which is reproduced below.

\marginalUnconfoundedness*

Recall that ${\bf V} = {\bf C} \cup {\bf U} \cup {\bf A} \cup {\bf M} \cup \{Y\}$, where ${\bf C}$ represents observed confounders, ${\bf U}$ represents unobserved, or latent, confounders, ${\bf A} = \{A_1, A_2, \dots, A_k\}$ represents $k$ potential causes of the outcome that we must select from, ${\bf M}$ represents variables mediating treatment and outcome, and $Y$ represents the outcome of interest.
Furthermore, recall that we assume that the variables in $\G$ follow the partial topological ordering ${\bf C} \cup {\bf U} \prec {\bf A} \prec {\bf M} \prec Y$.

\treatmentOutcomeSufficient*

\begin{proof}
    By definition of marginal unconfoundedness, $Y({\bf a}) \ci {\bf A}$ \citep{ref:richardson2013single}.
    Furthermore, recall that consistency assumes that $Y({\bf a}) = Y$ when ${\bf A} = {\bf a}$.
    Thus,
    \begin{align*}
        \E[Y({\bf a})] &= \E[Y({\bf a}) \mid {\bf A}={\bf a}] \\
        &= \E[Y \mid {\bf A} = {\bf a}]
    \end{align*}
    The first equality follows from marginal unconfoundedness, and the second equality follows from consistency.    
\end{proof}

\subsection{Proof Intuition for \cref{thm:reweighted_bic_consistency}}
\label[secinapp]{app:reweighted_bic_consistency_intuition}

Here, we give a short proof intuition for \cref{thm:reweighted_bic_consistency}.

When ${\bf Z} = \emptyset$, the intuition for the proof of \cref{thm:reweighted_bic_consistency} is similar to that for \cref{thm:bic_consistency}.
First, like \cref{lm:likelihood_simplification}, the rBIC can be equivalently expressed as a function of the reweighted mean and the model dimension.
Because we have access to the ground-truth values of the weights, the reweighted mean and MLE possess the same convergence properties as the typical mean and MLE.
All of the steps of the proof for \cref{thm:reweighted_bic_consistency} then follow the same steps as the proof for \cref{thm:bic_consistency}.
For the case where ${\bf Z} \neq \emptyset$, applying the same steps as the proof of \cref{cor:bic_consistency_conditional} achieves the desired result.

\subsection{Proof of \cref{thm:reweighted_bic_consistency}}

We first restate the definition of the reweighted BIC.

\reweightedBicDefinition*

The proof for \cref{thm:reweighted_bic_consistency} follows the proof for \cref{thm:bic_consistency} almost identically.
However, the key difference is recognizing that the reweighted BIC can be rewritten as a function of the reweighted mean of the sufficient statistic ${\bf T}({\bf X})$.
For the remainder of this subsection, we will suppose that the weights have been standardized such that $\sum_{i=1}^n W_i = n$.

Like before, we will first prove several lemmas that will greatly simplify the proof of \cref{thm:reweighted_bic_consistency}.

\begin{restatable}[Simplification of the Reweighted Log-Likelihood]{lemma}{reweightedLikelihoodSimplification}\label{lm:reweighted_likelihood_simplification}
    Let ${\bf X}$ be a set of random variables and suppose that ${\bf Z} = \emptyset$.
    Let $M_r \subset \Theta$ be a model, $d_r$ be the model dimension of $M_r$, and $\bar{{\bf Y}}_{W} = \frac{1}{n} \sum_{i=1}^n W_i {\bf T}({\bf X}_{i})$.
    If the model $M_r$ for ${\bf X}$ is a member of the exponential family,
    then the reweighted BIC $\gamma_W(\D, \W, M_r)$ can be equivalently written as
    \begin{align}
        &\gamma_W(\D, \W, M_r) = \sum_{i=1}^n W_i \log h({\bf X}_{i}) + n \sup_{\hat{\btheta}_r \in M_r} \bigg( \bar{{\bf Y}}_{W} \hat{\btheta}_{r} - b(\hat{\btheta}_{r}) \bigg) - \frac{d_r}{2} n^{\alpha_d}
        \label{eq:simplified_reweighted_bic}
    \end{align}
\end{restatable}

\begin{proof}
    First, let us consider the reweighted MLE for the random variables ${\bf X}$ whose joint density is a member of the exponential family.
    From \cref{app:reweighted_mle}, the reweighted MLE seeks to find values of the parameters $\hat{\btheta}_W$ such that the criterion function
    \begin{align*}
        \sum_{i=1}^n W_i \log p({\bf X}_i \mid \hat{\btheta}_W)
    \end{align*}
    is maximized.
    Because the model for ${\bf X}$ is a member of the exponential family, we can simplify this quantity we seek to maximize as
    \begin{align}
        &\sum_{i=1}^n W_i \log p({\bf X}_i \mid \hat{\btheta}_W) \nonumber \\
        &= \sum_{i=1}^n W_i \log \bigg[ h({\bf X}_i) \exp \bigg[ {\bf T}({\bf X}_i) \hat{\btheta}_W - b(\hat{\btheta}_W) \bigg] \bigg] \label{eq:reweighted_likelihood_simplification2} \\
        &= \sum_{i=1}^n W_i \bigg( \log h({\bf X}_i) + \bigg[ {\bf T}({\bf X}_i) \hat{\btheta}_W - b(\hat{\btheta}_W) \bigg]  \bigg) \label{eq:reweighted_likelihood_simplification3} \\
        &= \sum_{i=1}^n W_i \log h({\bf X}_i) + \sum_{i=1}^n W_i {\bf T}({\bf X}_i) \hat{\btheta}_W - \sum_{i=1}^n W_i b(\hat{\btheta}_W) \label{eq:reweighted_likelihood_simplification4} \\
        &= \sum_{i=1}^n W_i \log h({\bf X}_i) + n \bar{{\bf Y}}_W \hat{\btheta}_W - n b(\hat{\btheta}_W) \label{eq:reweighted_likelihood_simplification5} \\
        &= \sum_{i=1}^n W_i \log h({\bf X}_i) + n (\bar{{\bf Y}}_W \hat{\btheta}_W - b(\hat{\btheta}_W)) \label{eq:reweighted_likelihood_simplification6} 
    \end{align}
    \cref{eq:reweighted_likelihood_simplification2} follows from plugging in the definition for the PDF or PMF of variables that are members of the exponential family.
    \cref{eq:reweighted_likelihood_simplification3} follows from applying the product rule of logarithms and the fact that the logarithm of the exponent is the identity function.
    \cref{eq:reweighted_likelihood_simplification4} follows from distributing the weights $W_i$ to each term.
    \cref{eq:reweighted_likelihood_simplification5} follows from the definition of the reweighted mean $n \bar{{\bf Y}}_W = \sum_{i=1}^n W_i {\bf T}({\bf X}_i)$ and the assumption that the sum of the weights have been standardized to sum to $n$.
    \cref{eq:reweighted_likelihood_simplification6} follows from factoring $n$ out of the last two terms.
    In conclusion, the reweighted MLE equivalently maximizes \cref{eq:reweighted_likelihood_simplification6}.

    If $M_r$ denotes the model space and $\hat{\btheta}_r$ denotes the reweighted MLE, the above derivation means that the reweighted log-likelihood evaluated at the reweighted MLE can be written as
    \begin{align}
        \sum_{i=1}^n W_i \log h({\bf X}_i) + n \sup_{\hat{\btheta}_r \in M_r} (\bar{{\bf Y}}_W \hat{\btheta}_r - b(\hat{\btheta}_r)) \label{eq:reweighted_likelihood_simplification7}
    \end{align}
    because only the second and third terms depend on the reweighted MLE.
    
    Now, applying \cref{def:reweighted_bic} and \cref{eq:reweighted_likelihood_simplification7} gives us
    \begin{align*}
        \gamma_W(\D, \W, M_r) &= \sum_{i=1}^n W_i \log p({\bf X}_i \mid \hat{\btheta}_r, M_r) - \frac{d_r}{2} n^{\alpha_d} \\
        &= \sum_{i=1}^n W_i \log h({\bf X}_i) + n \sup_{\hat{\btheta} \in M_r} (\bar{{\bf Y}}_W \hat{\btheta} - b(\hat{\btheta})) - \frac{d_r}{2} n^{\alpha_d}
    \end{align*}
    This completes the proof.
\end{proof}

Next, the reweighted versions of \cref{lm:gradient_log_partition,lm:max_log_likelihood} are straightforward to show.
Recall from \cref{app:reweighted_mle} that the notation $\E_q[{\bf X}]$ represents the expectation of ${\bf X}$ with respect to the target reweighted distribution.

\begin{restatable}[Gradient of the Log-Partition Function at the Reweighted Parameter]{lemma}{gradientLogPartitionReweighted}\label{lm:gradient_log_partition_reweighted}
    The gradient of the log-partition function evaluated at the true value of the reweighted parameters $\btheta^*_W$ is equal to the expected value of ${\bf T}({\bf X})$ with respect to the reweighted distribution.
    That is,
    \begin{align*}
        \nabla b(\btheta^*_W) &= \E_q[{\bf T}({\bf X})],
    \end{align*}
    where $\nabla b(\btheta^*_W)$ is the gradient of the log-partition function -- the vector of partial derivatives of the log-partition function evaluated at the parameter vector $\btheta^*_W$.
\end{restatable}

\begin{proof}
    The proof follows from the proof of \cref{lm:gradient_log_partition} by replacing $\btheta^*$ with $\btheta^*_W$ and noting that 
    \begin{align*}
        \int ({\bf T}({\bf x}) - \nabla b(\btheta^*_W)) h({\bf x}) \exp[{\bf T}({\bf x}) \btheta^*_W - b(\btheta^*_W)] d{\bf x} &= \E_q[({\bf T}({\bf x}) - \nabla b(\btheta^*_W))]
    \end{align*}
    because the PDF or PMF in the integral is with respect to the parameters in the reweighted distribution.
\end{proof}

\begin{restatable}[Maximum Log-Likelihood of a Single Observation at the Reweighted Parameter]{lemma}{maxLogLikelihoodReweighted}\label{lm:max_log_likelihood_reweighted}
    Let $f(\btheta) = \nabla b(\btheta^*_W) \btheta - b(\btheta)$.
    $f(\btheta)$ attains its unique maximum at $\btheta^*_W$.
\end{restatable}

\begin{proof}
    The proof follows from the proof of \cref{lm:max_log_likelihood} by replacing $\btheta^*$ with $\btheta^*_W$.
\end{proof}

Finally, let us consider some important properties of the reweighted mean $\bar{{\bf Y}}_W$.

\begin{restatable}[Convergence of the Reweighted Mean]{lemma}{reweightedMean}\label{lm:reweighted_mean}
    Let $\bar{{\bf Y}}_W = \frac{1}{n} \sum_{i=1}^n W_i {\bf T}({\bf X}_i)$.
    Let $\sigma^2 = Var[W {\bf T}({\bf X})] < 0$.
    Then,
    \begin{enumerate}
        \item $\bar{{\bf Y}}_W \xrightarrow{p} \E_q[{\bf T}({\bf X})]$
        \item $\sqrt{n}(\bar{{\bf Y}}_W - \E_q[{\bf T}({\bf X})]) \xrightarrow{d} \mathcal{N}(0, \sigma^2)$
    \end{enumerate}
    
    In words, 1. $\bar{{\bf Y}}_W$ converges in probability to $\E_q[{\bf T}({\bf X})]$.
    Furthermore, 2. $\sqrt{n}(\bar{{\bf Y}}_W - \E_q[{\bf T}({\bf X})])$ converges in distribution to a normal distribution with mean $0$ and variance $\sigma^2$.
\end{restatable}

\begin{proof}
    First, let us consider the expected value of $\bar{{\bf Y}}_W$.
    \begin{align*}
        \E[\bar{{\bf Y}}_W] &= \E[\frac{1}{n} \sum_{i=1}^n W_i {\bf T}({\bf X}_i)] \\
        &= \E[W {\bf T}({\bf X})] \\
        &= \int_{{\bf V}, {\bf A}'} W {\bf T}({\bf X}) p({\bf V}) d{\bf v}, {\bf a}' \\
        &= \int_{{\bf V}, {\bf A}'} {\bf T}({\bf X}) W p({\bf V}) d{\bf v}, {\bf a}' \\
        &= \int_{{\bf V}, {\bf A}'} {\bf T}({\bf X}) q({\bf V}, {\bf A}') d{\bf v}, {\bf a}' \\
        &= \E_q[{\bf T}({\bf X})]
    \end{align*}
    This shows that the expected value of $\bar{{\bf Y}}_W$ is equal to the true mean of ${\bf T}({\bf X})$ with respect to the reweighted distribution $q({\bf V}, {\bf A}')$.
    By the law of large numbers, the first part of the lemma is true.
    By the central limit theorem, the second part of the lemma is true \cite{ref:van2000asymptotic}.
\end{proof}

An important consequence of \cref{lm:reweighted_mean} is that $\bar{{\bf Y}}_W$ converges to $\E_{q}[{\bf T}({\bf X})]$ at the standard $O_p(1/\sqrt{n})$ rate.

Having established these lemmas, we can now prove the consistency of the reweighted BIC with oracle weights by directly invoking the proof of the consistency of the classical BIC.

\reweightedBicConsistency*

\begin{proof}
    We first start with the case where ${\bf Z} = \emptyset$.
    Then we will show how consistency of the rBIC for joint distributions implies consistency of the rBIC for conditional distributions.
    
    Using \cref{lm:reweighted_likelihood_simplification} to simplify the difference $\gamma_W(\D, \W, M_2) - \gamma_W(\D, \W, M_1)$,
    \begin{align*}
        \gamma_W(\D, \W, M_2) - \gamma_W(\D, \W, M_1) &= \sum_{i=1}^n W_i \log h({\bf X}_i) + n \sup_{\hat{\btheta}_2 \in M_2} (\bar{{\bf Y}}_W \hat{\btheta}_2 - b(\hat{\btheta}_2)) - \frac{d_2}{2} n^{\alpha_d} \\
        &\quad \quad - \bigg( \sum_{i=1}^n W_i \log h({\bf X}_i) + n \sup_{\hat{\btheta}_1 \in M_1} (\bar{{\bf Y}}_W \hat{\btheta}_1 - b(\hat{\btheta}_1)) - \frac{d_1}{2} n^{\alpha_d} \bigg) \\
        &= n \bigg( \sup_{\hat{\btheta}_2 \in M_2} \bigg(\bar{{\bf Y}}_W \hat{\btheta}_2 - b(\hat{\btheta}_2)\bigg) - \sup_{\hat{\btheta}_1 \in M_1} \bigg(\bar{{\bf Y}}_W \hat{\btheta}_1 - b(\hat{\btheta}_1)\bigg) \bigg) - \frac{1}{2}(d_2 - d_1) n^{\alpha_d}
    \end{align*}
    The desired result follows from first taking the proof of \cref{thm:bic_consistency} and replacing $\btheta^*$ with $\btheta^*_W$, \cref{lm:max_log_likelihood} with \cref{lm:max_log_likelihood_reweighted}, $\bar{{\bf Y}}_n$ with $\bar{{\bf Y}}_W$, $\E[{\bf T}({\bf X})]$ with $\E_q[{\bf T}({\bf X})]$, \cref{lm:gradient_log_partition} with \cref{lm:gradient_log_partition_reweighted}, \cref{lm:likelihood_simplification} with \cref{lm:reweighted_likelihood_simplification}, and $\gamma(\D, M_2) - \gamma(\D, M_1)$ with $\gamma_W(\D, \W, M_2) - \gamma_W(\D, \W, M_1)$.
    Then, under condition (i), with probability $1$ as $n \rightarrow \infty$,
    \begin{align*}
        \gamma_W(\D, \W, M_2) - \gamma_W(\D, \W, M_1) &> n(\epsilon/2) - \frac{1}{2} (d_2 - d_1) n^{\alpha_d} \\
        &> 0
    \end{align*}
    because $n(\epsilon/2)$ asymptotically dominates the difference between the model dimension, which is scaled by $n^{\alpha_d}$ where $\alpha_d < 1$.

    Next, \cref{app:reweighted_mle} shows that the reweighted MLE converges to $\btheta^*_W$ at a rate of $O_p(n^{-1/2})$, and \cref{lm:reweighted_mean} shows that the reweighted mean converges to $\E_q[{\bf T}(X)]$ at a rate of $O_p(n^{-1/2})$.
    Thus, under condition (ii),
    \begin{align*}
        \gamma_W(\D, \W, M_2) - \gamma_W(\D, \W, M_1) &= O_p(1) - \frac{1}{2} (d_2 - d_1) n^{\alpha_d} \\
        &> 0
    \end{align*}
    because $d_2 < d_1$, and the $n^{\alpha_d}$ term asymptotically dominates the constant term since $\alpha_d > 0$.
    This completes the proof for the case where ${\bf Z}=\emptyset$.

    Next, using the same proof technique as the proof for \cref{cor:bic_consistency_conditional}, we will show that the rBIC for comparing conditional distributions is equivalent to comparing joint distributions, which we already proved above is consistent.
    Suppose that the empirical distribution of ${\bf Z}$, the parameters $\boldeta$, and the estimation procedure for $\boldeta$ are as they are described in the proof of \cref{cor:bic_consistency_conditional}, with the slight change that the empirical distribution is reweighted:
    \begin{align*}
        q_{\bf Z}({\bf Z} = {\bf z}) = \frac{1}{n} \sum_{i=1}^n W_i \mathbb{I}({\bf Z}_i = {\bf z}).
    \end{align*}
    Consider the rBIC for a conditional distribution:
    \begin{align*}
        \gamma(\D, \W, M_r) &= \sum_{i=1}^n W_i \log p_{{\bf X} \mid {\bf Z}}({\bf X}_i \mid {\bf Z}_i, \hat{\btheta}({\bf Z}_i), M_r) - \frac{d_r}{2} n^{\alpha_d} \\
        &= \sum_{i=1}^n W_i \log \frac{p_{{\bf X}, {\bf Z}}({\bf X}_i, {\bf Z}_i \mid \hat{\btheta}, \hat{\boldeta}, M_r)}{p_{\bf Z}({\bf Z}_i \mid \hat{\boldeta})} - \frac{d_r}{2} n^{\alpha_d} \\
        &= \sum_{i=1}^n W_i \log p_{{\bf X}, {\bf Z}}({\bf X}_i, {\bf Z}_i \mid \hat{\btheta}, \hat{\boldeta}, M_r) - \sum_{i=1}^n W_i \log p_{\bf Z}({\bf Z}_i \mid \hat{\boldeta}) - \frac{d_r}{2} n^{\alpha_d} 
    \end{align*}
    When we use the rBIC to compare two models, the term $\sum_{i=1}^n W_i \log p_{\bf Z}({\bf Z}_i \mid \hat{\boldeta})$ cancels out exactly because it does not depend on any of the models, only the weights and the observed data.
    Furthermore, using the same argument as in the proof of \cref{cor:bic_consistency_conditional}, the assumption that ${\bf X}$ given ${\bf Z}$ is a member of the exponential family is sufficient to show that the joint distribution of ${\bf X}$ and ${\bf Z}$ is also a member of the exponential family when the distribution of ${\bf Z}$ is estimated using its reweighted empirical distribution.
    Thus, comparing conditional distributions is equivalent to comparing joint distributions, so the rBIC is also consistent for conditional distributions as well.
    We have thus arrived at our desired result.
\end{proof}

\subsection{Proof of \cref{thm:consistency_estimated_weights}}

\consistencyEstimatedWeights*

\begin{proof}
    Let $\bar{{\bf Y}}_{\hat{W}} = \frac{1}{n} \sum_{i=1}^n \hat{W}_i {\bf T}({\bf X}_i)$.
    This is the reweighted mean estimated using the estimated weights rather than the ground truth weights.
    We start by showing that $\bar{{\bf Y}}_{\hat{W}} - \bar{{\bf Y}}_W = O_p(n^{-\alpha_W})$.
    \begin{align}
        \bar{{\bf Y}}_{\hat{W}} - \bar{{\bf Y}}_W &= \frac{1}{n} \sum_{i=1}^n \hat{W}_i {\bf T}({\bf X}_i) - \frac{1}{n} \sum_{i=1}^n W_i {\bf T}({\bf X}_i) \nonumber \\
        &= \frac{1}{n} \sum_{i=1}^n (\hat{W}_i - W_i) {\bf T}({\bf X}_i) \label{eq:reweighted_mean_convergence1} \\
        &\leq \bigg( \frac{1}{n} \sum_{i=1}^n (\hat{W}_i - W_i)^2 \bigg)^{1/2} \bigg( \frac{1}{n} \sum_{i=1}^n {\bf T}({\bf X}_i)^2 \bigg)^{1/2} \label{eq:reweighted_mean_convergence2} \\
        &= o_p(n^{-\alpha_W}) O_p(1) \label{eq:reweighted_mean_convergence3} \\
        &= O_p(n^{-\alpha_W}) \label{eq:reweighted_mean_convergence4}
    \end{align}
    \cref{eq:reweighted_mean_convergence1} follows from regrouping terms.
    \cref{eq:reweighted_mean_convergence2} follows from the Cauchy-Schwarz inequality.
    \cref{eq:reweighted_mean_convergence3} follows from the initial assumption that the estimated weights converge to the oracle weights at a rate of $o_p(n^{-\alpha_W})$ in RMSE.
    Furthermore, the second moment of the sufficient statistic is just  a vector of finite numbers, so the empirical estimate of the second moment will converge to a vector of numbers by the law of large numbers.
    Thus, the estimated reweighted mean converges to the oracle reweighted mean at a rate of $o_p(n^{-\alpha_W})$.
    \cref{eq:reweighted_mean_convergence4} follows from $o_p(n^{-\alpha_W}) O_p(1) = o_p(n^{-\alpha_W}) = O_p(n^{-\alpha_W})$.

    Let $g({\bf Z}, M_r) = \sup_{\hat{\btheta}_r \in M_r} ({\bf Z} \hat{\btheta}_r - b(\hat{\btheta}_r))$, where $b(\btheta)$ is the log-partition function for the model $M_r$, who is a member of the exponential family.
    This function takes as input a vector of values ${\bf Z}$, finds the value $\hat{\btheta}_r \in M_r$ that maximizes the expression, and outputs the value of the maximized expression.

    Recall that $\btheta_W^* \in \mathrm{int} \Theta$ is the true parameter in the reweighted distribution.
    By \cref{lm:reweighted_mean,lm:gradient_log_partition_reweighted}, $\bar{{\bf Y}}_W \xrightarrow{p} \nabla b(\btheta^*_W)$.
    This means that, as the amount of data increases to infinity, $\bar{{\bf Y}}_W$ will get arbitrarily close to $\nabla b(\btheta^*_W)$.
    Let $\btheta_0 \in M_r$ be the maximizing value of $g(\nabla b(\btheta^*_W), M_r)$.
    Note that $\btheta_0$ here is a finite vector in the interior of $\Theta$ that is either $\btheta_W^*$ itself (according to \cref{lm:max_log_likelihood_reweighted}), or the point in $M_r$ closest to $\btheta_W^*$ if $\btheta_W^* \not \in M_r$.
    
    We now note that the maximizing value of $\hat{\btheta}_r$ in $g(\bar{{\bf Y}}_W, M_r)$ will be bounded in a neighborhood near $\btheta_0$.
    This is due to the fact that $\bar{{\bf Y}}_W$ will itself be in a neighborhood near $\nabla b(\btheta^*_W)$ as the former is converging in probability to the latter.
    
    Next, let us consider the derivative of $g({\bf Z}, M_r)$ with respect to ${\bf Z}$ in a neighborhood near $\bar{{\bf Y}}_W$.
    This derivative will be $\hat{\btheta}_r$, the maximizing value of $g({\bf Z}, M_r)$.
    We can thus conclude that the first derivative of $g$ will be bounded in a neighborhood near $\bar{{\bf Y}}_W$ because we showed above that $\hat{\btheta}_r$ is bounded.
    
    We can then use the mean value theorem to show that there exists a neighborhood $U$ around $\bar{{\bf Y}}_W$ and a constant $K < \infty$ such that $\forall {\bf Z}_1, {\bf Z}_2 \in U$,
    \begin{align*}
        |g({\bf Z}_1, M_r) - g({\bf Z}_2, M_r)| \leq K |{\bf Z}_1 - {\bf Z}_2|
    \end{align*}
    That is, $g({\bf Z}, M_r)$ is locally Lipschitz near $\bar{{\bf Y}}_W$.
    As $n \rightarrow \infty$, $\bar{{\bf Y}}_{\hat{W}}$ will get arbitrarily close to $\bar{{\bf Y}}_W$ because $\bar{{\bf Y}}_{\hat{W}} - \bar{{\bf Y}}_W = O_p(n^{-\alpha_W})$, so $\bar{{\bf Y}}_{\hat{W}}$ will be in the desired neighborhood $U$ with sufficient data.

    This allows us to conclude
    \begin{align*}
        |g(\bar{{\bf Y}}_{\hat{W}}, M_r) - g(\bar{{\bf Y}}_W, M_r)| &\leq K |\bar{{\bf Y}}_{\hat{W}} - \bar{{\bf Y}}_W| \\
        &\leq K O_p(n^{-\alpha_W}) \\
        &= O_p(n^{-\alpha_W})
    \end{align*}
    Thus, $g(\bar{{\bf Y}}_{\hat{W}}, M_r) - g(\bar{{\bf Y}}_W, M_r) = O_p(n^{-\alpha_W})$.

    As a reminder, from \cref{lm:reweighted_likelihood_simplification},
    \begin{align*}
        \gamma_W(\D, \W, M_r) &= \sum_{i=1}^n W_i \log h({\bf X}_i) + n g(\bar{{\bf Y}}_W, M_r) - \frac{d_r}{2} n^{\alpha_d}
    \end{align*}
    Following the same steps as \cref{lm:reweighted_likelihood_simplification},
    \begin{align*}
        \gamma_W(\D, \hat{\W}, M_r) &= \sum_{i=1}^n \hat{W}_i \log h({\bf X}_i) + n \sup_{\hat{\btheta}_r \in M_r} (\bar{{\bf Y}}_{\hat{W}} \hat{\btheta}_r - b(\hat{\btheta}_r)) - \frac{d_r}{2} n^{\alpha_d} \\
        &= \sum_{i=1}^n \hat{W}_i \log h({\bf X}_i) + n g(\bar{{\bf Y}}_{\hat{W}}, M_r) - \frac{d_r}{2} n^{\alpha_d}
    \end{align*}
    Taking the difference between the reweighted BIC with estimated weights and oracle weights under the same model $M_r$ gives us
    \begin{align}
        \gamma_W(\D, \hat{\W}, M_r) - \gamma_W(\D, \W, M_r) &= \sum_{i=1}^n \hat{W}_i \log h({\bf X}_i) + n g(\bar{{\bf Y}}_{\hat{W}}, M_r) - \frac{d_r}{2} n^{\alpha_d} \nonumber \\
        &\quad \quad - \bigg( \sum_{i=1}^n W_i \log h({\bf X}_i) + n g(\bar{{\bf Y}}_W, M_r) - \frac{d_r}{2} n^{\alpha_d} \bigg) \nonumber \\
        &= \sum_{i=1}^n (\hat{W}_i - W_i) \log h({\bf X}_i) + n (g(\bar{{\bf Y}}_{\hat{W}}, M_r) - g(\bar{{\bf Y}}_W, M_r)) \nonumber \\
        &= \sum_{i=1}^n (\hat{W}_i - W_i) \log h({\bf X}_i) + n O_p(n^{-\alpha_W}) \nonumber \\
        &= \sum_{i=1}^n (\hat{W}_i - W_i) \log h({\bf X}_i) + O_p(n^{1-\alpha_W}) \label{eq:diff_est_ora_weights}
    \end{align}
    Let $A = \sum_{i=1}^n (\hat{W}_i - W_i) \log h(X_i)$.
    Note that $A$ does not depend on the model, just the oracle weights, estimated weights, and the observed data.

    Finally, let us consider the difference between the reweighted BIC evaluated at two models $M_1$ and $M_2$ using the estimated weights.
    \begin{align}
        &\gamma_W(\D, \hat{\W}, M_2) - \gamma_W(\D, \hat{\W}, M_1) \nonumber \\
        &= \gamma_W(\D, \hat{\W}, M_2) - \gamma_W(\D, \W, M_2) + \gamma_W(\D, \W, M_2) - (\gamma_W(\D, \hat{\W}, M_1) - \gamma_W(\D, \W, M_1) + \gamma_W(\D, \W, M_1)) \nonumber \\
        &= A + O_p(n^{1-\alpha_W}) + \gamma_W(\D, \W, M_2) - (A + O_p(n^{1-\alpha_W}) + \gamma_W(\D, \W, M_1)) \label{eq:remove_bic_est_weights} \\
        &= \gamma_W(\D, \W, M_2) - \gamma_W(\D, \W, M_1) + O_p(n^{1-\alpha_W}) \nonumber
    \end{align}
    \cref{eq:remove_bic_est_weights} follows from two applications of \cref{eq:diff_est_ora_weights}.

    Following the proof of \cref{thm:reweighted_bic_consistency}, under condition (i), with probability 1 as $n \rightarrow \infty$,
    \begin{align*}
        \gamma_W(\D, \W, M_2) - \gamma_W(\D, \W, M_1) + O_p(n^{1-\alpha_W}) &> n(\epsilon/2) - \frac{1}{2}(d_2 - d_1) n^{\alpha_d} + O_p(n^{1-\alpha_W}) \\
        &> 0
    \end{align*}
    because $n(\epsilon/2)$ dominates both other terms as $1-\alpha_W < \alpha_d < 1$ by construction.

    Again, following the proof of \cref{thm:reweighted_bic_consistency}, under condition (ii), with probability 1 as $n \rightarrow \infty$,
    \begin{align*}
        \gamma_W(\D, \W, M_2) - \gamma_W(\D, \W, M_1) + O_p(n^{1-\alpha_W}) &= O_p(1) - \frac{1}{2}(d_2 - d_1) n^{\alpha_d} + O_p(n^{1-\alpha_W}) \\
        &> 0
    \end{align*}
    because $d_2 < d_1$, and the $n^{\alpha_d}$ term asymptotically dominates both other terms due to $\alpha_d > 0$ and $\alpha_d > 1-\alpha_W$ by construction.
    This completes the proof.
\end{proof}

\section{Simulation Details and Additional Results}
\label[secinapp]{app:sims}

\subsection{Data-Generating Process Details}
\label[secinapp]{app:sims_dgp}

We first remind the reader of the expit function $\text{expit}(x) = \frac{1}{1+e^{-x}}$.

The DGP for the graph shown in \cref{fig:example_dag_backdoor} is:

\begin{align*}
    C &\sim \mathcal{N}(0, 0.25) \\
    A1 &\sim \text{Bernoulli}(\text{expit}(C)) \\
    A2 &\sim \text{Bernoulli}(\text{expit}(C)) \\
    A3 &\sim \text{Bernoulli}(\text{expit}(C)) \\
    Y &\sim 1.5 C + 1.5 A_1 + 1.5 A_2 + \mathcal{N}(0, 1)
\end{align*}

The DGP that follows \cref{fig:example_dag_frontdoor} is:

\begin{align*}
    U &\sim \mathcal{N}(0, 0.01) \\
    A1 &\sim \text{Bernoulli}(\text{expit}(U)) \\
    A2 &\sim \text{Bernoulli}(\text{expit}(U)) \\
    A3 &\sim \text{Bernoulli}(\text{expit}(U)) \\
    M &\sim \sqrt{1.5} A_1 - \sqrt{1.5} A_3 + \mathcal{N}(0, 2.875) \\
    Y &\sim 1.5 U + \sqrt{1.5} M + \mathcal{N}(0, 1)
\end{align*}

Note that $U$ is unobserved in this setting, so none of the model selection methods trialed in the simulation for the frontdoor graph have access to it.
The coefficients of $A_1$ and $A_3$ in the conditional mean model $\E[Y \mid U, A_1, A_3]$ are both $1.5$.
In both DGPs, the variances for the normally distributed random variables are chosen to keep the computation of the weights stable in finite samples.

\subsection{Computation of Weights}
\label[secinapp]{app:sims_weights}

In this subsection, we start by describing the computation of weights in the backdoor setting (\cref{fig:example_dag_backdoor}).
Then, we will describe computing the weights in the frontdoor setting (\cref{fig:example_dag_frontdoor}).

Recall that the weights in the backdoor setting are
\begin{align*}
    W = \frac{p^*({\bf A})}{p({\bf A} \mid C)} &= \frac{p^*(A_1, A_2, A_3)}{p(A_1, A_2, A_3 \mid C)} \\
    &= \frac{p^*(A_1) p^*(A_2) p^*(A_3)}{p(A_1 \mid C) p(A_2 \mid C) p(A_3 \mid C)}.
\end{align*}
$p^*(A_1, A_2, A_3)$ are user-defined probabilities for the potential treatments, so we can simply define them to be independent of each other.
Next, the treatments are independent of each other conditional on $C$ in our synthetic DGP, which implies $p(A_1, A_2, A_3 \mid C) = p(A_1 \mid C) p(A_2 \mid C) p(A_3 \mid C)$.
In general, if treatments are not independent of each other conditional on $C$, we can simply use the chain rule to get $p(A_1, A_2, A_3 \mid C) = p(A_1 \mid A_2, A_3, C) p(A_2 \mid A_3, C) p(A_3 \mid C)$.

We define $p^*(A_j = 1) = 0.5$ for $j=1,2,3$.
Thus, the numerator of the weight for each individual is $0.5^3$, regardless of the observed values of the treatments.
To estimate the probabilities in the denominator, we fit a linear logistic regression model $\hat{p}(A_j \mid C; \eta_j)$ for $j=1,2,3$, where $\eta_j$ are the learned parameters from the linear logistic regression.
Then, for each individual, we use the fitted models to predict their conditional probability of getting assigned their observed treatment.
That is, the weight for individual $i$ is
\begin{align*}
    W_i &= \frac{0.5^3}{\hat{p}(A_1 = a_{1i} \mid C = c_i; \eta_1) \hat{p}(A_2 = a_{2i} \mid C = c_i; \eta_2) \hat{p}(A_3 = a_{3i} \mid C = c_i; \eta_3)}
\end{align*}
where $a_{1i}, a_{2i}, a_{3i},$ and $c_i$ are the observed values of $A_1, A_2, A_3,$ and $C$ for individual $i$.

Next, we describe the computation of weights in the frontdoor setting.
The weights in the frontdoor setting are
\begin{align*}
    \frac{p^*({\bf A}') p(M \mid {\bf A})}{p(M \mid {\bf A}')} &= \frac{p^*(A_1', A_2', A_3') p(M \mid A_1, A_2, A_3)}{p(M \mid A_1', A_2', A_3')} \\
    &= \frac{p^*(A_1') p^*(A_2') p^*(A_3') p(M \mid A_1, A_2, A_3)}{p(M \mid A_1', A_2', A_3')},
\end{align*}
where $A_1', A_2',$ and $A_3'$ are binary random variables generated according to $A_j' \sim \text{Bernoulli}(0.5)$ for $j=1,2,3$.
Thus, $p^*(A_1') p^*(A_2') p^*(A_3') = 0.5^3$ regardless of the observed values of these simulated treatments.

Because $M$ is a normal random variable, the probabilities involving $M$ in the weights are replaced with conditional densities at observed values of the treatments and simulated treatments.
$p(M \mid A_1, A_2, A_3)$ is computed as follows.
We first fit a linear regression for $M$ as a function of $A_1, A_2,$ and $A_3$.
We then use the fitted regression to make a prediction $\hat{m}_i$ of $M$ for each individual $i$ based on the observed values of $A_1, A_2,$ and $A_3$.
Let $m_i$ be the observed value of $M$ for individual $i$.
We then estimate the variance of $M$ using the MLE
\begin{align*}
    \hat{\sigma}^2 &= \frac{1}{n} \sum_{i=1}^n (m_i - \hat{m}_i)^2.
\end{align*}
The density $p(M \mid A_1, A_2, A_3)$ for individual $i$ is
\begin{align*}
    \hat{p}(M = m_i \mid A_1=a_{1i}, A_2=a_{2i}, A_3=a_{3i}) &= \frac{1}{\sqrt{2 \pi \hat{\sigma}^2}} \exp (\frac{-(m_i - \hat{m}_i)^2}{2 \hat{\sigma}^2}).
\end{align*}
Next, $p(M \mid A_1', A_2', A_3')$ is computed as follows.
Using the linear regression learned above, we make a prediction $\hat{m}_i'$ of $M$ for each individual $i$ based on the observed values of the simulated treatments $A_1', A_2',$ and $A_3'$.
The density in the denominator is thus
\begin{align*}
    \hat{p}(M = m_i \mid A_1'=a_{1i}', A_2'=a_{2i}', A_3'=a_{3i}') &= \frac{1}{\sqrt{2 \pi \hat{\sigma}^2}} \exp (\frac{-(m_i - \hat{m}_i')^2}{2 \hat{\sigma}^2}),
\end{align*}
where $a_{1i}', a_{2i}',$ and $a_{3i}'$ are the observed values of the simulated treatments $A_1', A_2',$ and $A_3'$ for individual $i$.
The weight for individual $i$ is thus
\begin{align}
    W_i = \frac{0.5^3 \hat{p}(M = m_i \mid A_1=a_{1i}, A_2=a_{2i}, A_3=a_{3i})}{\hat{p}(M = m_i \mid A_1'=a_{1i}', A_2'=a_{2i}', A_3'=a_{3i}')}. \label{eq:frontdoor_weights}
\end{align}

Because the densities of $M$ can be small, especially at values of predicted $M$ far from the observed $M$, our simulation computes $W_i$ using the logarithm of the ratio of densities before transforming it back to its exponent as follows:
\begin{align*}
    \text{log-ratio}_i &= \frac{- ((m_i - \hat{m}_i)^2 - (m_i - \hat{m}_i')^2)}{2 \hat{\sigma}^2} \\
    W_i &= 0.5^3 \exp(\text{log-ratio}_i)
\end{align*}
The normalization term $\frac{1}{\sqrt{2 \pi \hat{\sigma}^2}}$ is omitted from the ratio of densities because it cancels out exactly. 
Computing the weights using this procedure is equivalent to computing the weights using \cref{eq:frontdoor_weights}.

\subsection{Adaptive LASSO and Regression with SCAD Penalty}
\label[secinapp]{app:sims_alasso_scad}

Before going into details, this paragraph gives the ``too long, did not read'' (TLDR) summary on the adaptive LASSO and regression with SCAD penalty methods.
Adaptive LASSO first fits an ordinary regression then uses the inverse of those learned parameters in the L1 penalty of a second regression.
This forces parameters that are close to $0$ in the first regression to be greatly penalized in the second regression.
The SCAD penalty term penalizes parameters during gradient descent using a piecewise function whose value is based on the absolute value of the estimated parameters at each step.
Large values of estimated parameters are not penalized whereas smaller values of estimated parameters receive increasingly large penalties, which force them to $0$.
We chose hyperparameters for both of these methods that satisfy the conditions for their oracle properties to take effect.

We fit the adatpive LASSO using the following steps \citep{ref:zou2006adaptive}.
\begin{enumerate}
    \item Run a linear regression of $Y$ on $A_1, A_2, A_3$, and covariates $C$, if they are available.
    Save $\tilde{\beta}_j$ for $j=1,\dots,p$, where $\tilde{\beta}_j$ is the estimated coefficient for the $j$th variable in the regression, and $p$ is the number of coefficients learned.
    \item Define the weights $\hat{\omega}_j = \frac{1}{|\tilde{\beta}_j|^\gamma}$, where $\gamma > 0$ is a user-specified hyperparameter.
    \item Let ${\bf X}$ be the design matrix for the linear regression. 
    Let ${\bf Y}$ be the vector of outcomes $Y$.
    Let $\beta$ be the vector of coefficients used in the adaptive LASSO regression.
    The adaptive LASSO estimates are given by
    \begin{align*}
        \underset{\beta}{\text{argmin}} \frac{1}{n} ({\bf Y} - {\bf X} \beta)^2 + \lambda_n \sum_{j=1}^p \hat{\omega}_j | \beta_j |,
    \end{align*}
    where $\lambda_n$ is a user-specified hyperparameter.
\end{enumerate}
Algorithm 1 of \cite{ref:bahamyirou2022doubly} gives a succint summary of the method under a similar setting.
In order for adaptive LASSO to achieve its oracle property, we need $\lambda_n/\sqrt{n} \rightarrow 0$ and $\lambda_n n^{(\gamma-1)/2} \rightarrow \infty$ \citep{ref:zou2006adaptive}.
In our simulations, we choose $\gamma=2$ and $\lambda_n = n^{-1/4}$.

We now describe regression with the SCAD penalty using the same notation as above \citep{ref:fan2001variable}.
SCAD stands for smoothly clipped absolute deviation.
The estimates of the regression with the SCAD penalty are given by
\begin{align*}
    \underset{\beta}{\text{argmin}} \frac{1}{n} ({\bf Y} - {\bf X} \beta)^2 + \sum_{j=1}^p p_{\lambda_n}(|\beta_j|),
\end{align*}
where $p_{\lambda_n}(|\beta_j|)$ is the penalty function whose derivative is defined as
\begin{align*}
    p_{\lambda_n}'(\beta) &= \lambda_n \bigg( \mathbb{I}(\beta \leq \lambda_n) + \frac{(a \lambda_n - \beta)_+}{(a-1) \lambda_n} \mathbb{I}(\beta > \lambda_n) \bigg),
\end{align*}
where $\mathbb{I}(\cdot)$ is the indicator function and $(x)_+$ is equal to $x$ if $x \geq 0$ and $0$ otherwise.
$a$ and $\lambda_n$ are user-specified hyperparameters.
Breaking down the derivative of the penalty function, assuming $\beta > 0$,
\begin{enumerate}
    \item $p_{\lambda_n}'(\beta) = \lambda_n$ if $\beta \leq \lambda_n$.
    \item $p_{\lambda_n}'(\beta) = \frac{a \lambda_n - \beta}{(a-1)}$ if $\lambda_n < \beta < a \lambda_n$.
    \item $p_{\lambda_n}'(\beta) = 0$ if $\beta \geq a \lambda_n$.
\end{enumerate}
The derivative of the penalty term becomes increasingly small as the value of the estimated coefficient becomes increasingly large.

\cite{ref:fan2001variable} suggest $a=3.7$, so this is the value we use in our simulations.
Regression with the SCAD penalty achieves the oracle property when $\lambda_n \rightarrow 0$ and $\sqrt{n} \lambda_n \rightarrow \infty$ \citep{ref:fan2001variable}.
Thus, we choose $\lambda_n = n^{-1/4}$ in our simulation.

To ensure fair comparisons with the reweighted BIC, we use reweighted versions of the loss functions when fitting parameters for adaptive LASSO and regression with the SCAD penalty.

\subsection{Simulation Results with Oracle Weights}
\label[secinapp]{app:results_oracle_weights}

Here, we give additional simulation results executed under the same settings as those described in \cref{sec:simulation_study}, except we use the ground-truth weights to fit penalized regressions models and for computation of the rBIC.

As shown in \cref{fig:backdoor_oracle_results}, the rBIC with the $\log n$ and $n^{1/2}$ penalties select ${\bf A}^*$ correctly with high probability across all sample sizes.
The rBIC with the $n^{3/4}$ penalty selects the correct ${\bf A}^*$ with probability $1$ starting at sample size $n=1000$.
In the backdoor graph setting, results with oracle weights appear roughly the same as results with estimated weights.
This is because the convergence of the estimated weights is fast enough that the estimated weights behave like the oracle weights at small sample sizes.

\begin{figure}[h]
    \centering
    \includegraphics[scale=0.5]{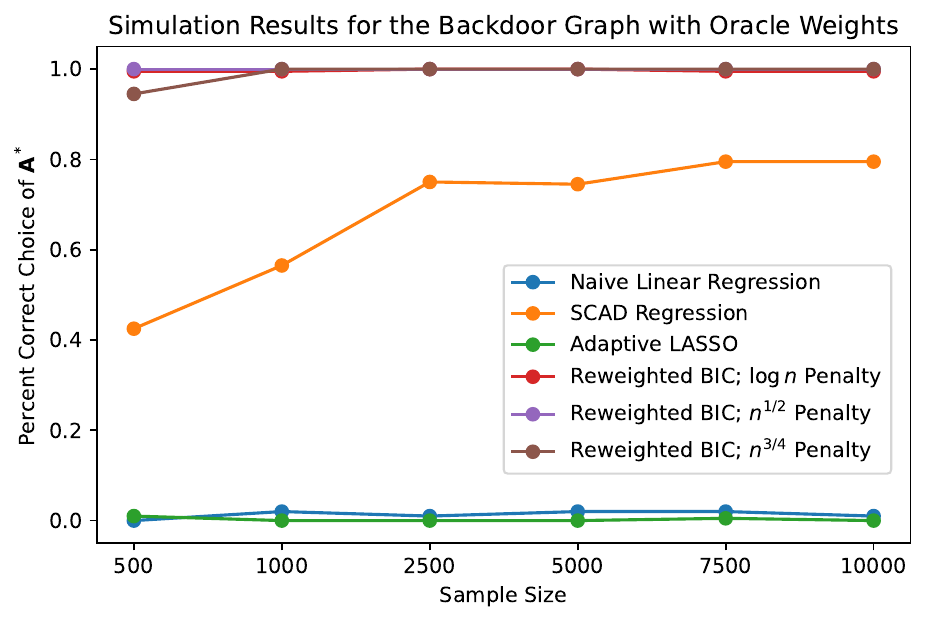}
    \caption{Simulation results for the backdoor graph setting when ground-truth weights are used.}
    \label{fig:backdoor_oracle_results}
\end{figure}

Next, consider the frontdoor graph setting.
Recall from \cref{fig:results} that the rBIC with the $\log n$ penalty selects the true model roughly $80\%$ of the time across sample sizes when we use estimated weights.
As shown in \cref{fig:frontdoor_oracle_results}, when we use the oracle weights, the $\log n$ penalty's accuracy rises to $90\%$, and we expect it to increase as the sample size increases further, which aligns with our theoretical results.
Due to computational constraints at such large sample sizes, however, we did not run the experiments at larger sample sizes.
Like the simulations with estimated weights, the accuracy of the rBIC with the $n^{3/4}$ penalty increases as the sample size increases.

\begin{figure}[h]
    \centering
    \includegraphics[scale=0.5]{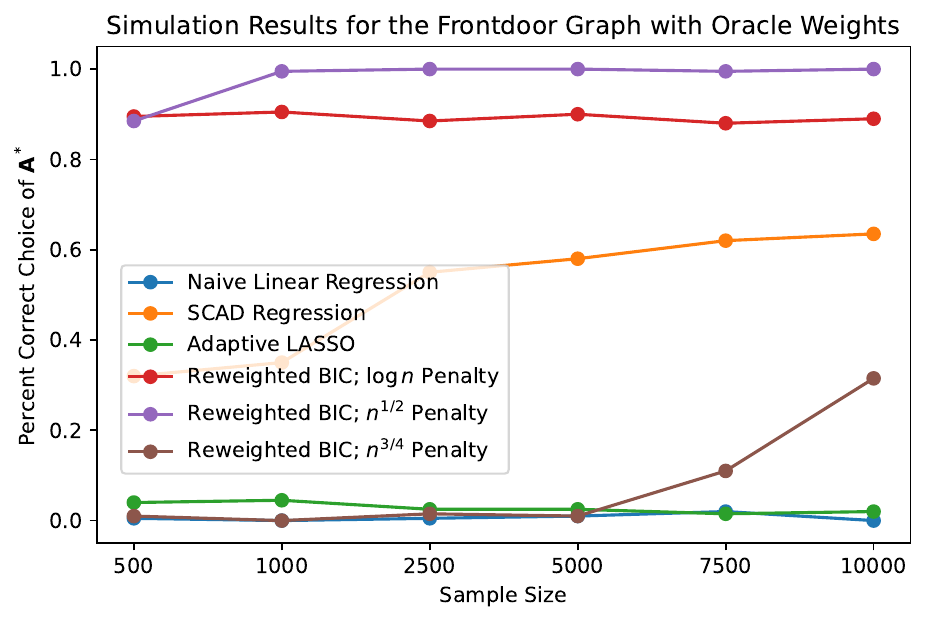}
    \caption{Simulation results for the frontdoor graph setting when ground-truth weights are used.}
    \label{fig:frontdoor_oracle_results}
\end{figure}

\subsection{Hardware Used for Simulation Study}
\label[secinapp]{app:sims_hardware}

The cluster used for the simulation study is comprised of 87 compute nodes, providing over 4000 64-bit cores/threads and 53 TB of DDR-SDRAM.
The cluster also has 7 GPU nodes with 23 Nvidia H100, V100, and A100 GPUs in support of numerous AI, ML, and LLM efforts.
The storage environment on the cluster contains a total of 20PB of storage with a 2PB Lustre storage cluster, 16PB of ZFS NAS storage, and a 2PB ZFS NAS storage array for offsite disk-to-disk backup, all built on the “dirt cheap” storage philosophy.
In addition, a 25 TB NVME SSD based “fast scratch” storage array is used to support short-term storage of high-IO processing. 

\end{document}